\documentclass[aps,prx,reprint,11pt,tightenlines,onecolumn,notitlepage,preprintnumbers,superscriptaddress,nofootinbib,longbibliography]{revtex4-1}

\usepackage{amsmath}
\usepackage{amssymb}
\usepackage{amsfonts}
\usepackage{amsthm}
\usepackage{comment}

\usepackage[T1]{fontenc}

\usepackage[utf8]{inputenc}

\usepackage[hmargin=1.25in,vmargin=1in]{geometry}

\usepackage{mathrsfs}\usepackage{graphicx}\usepackage{bm}\usepackage{xcolor}
\usepackage{enumitem}

\usepackage{mathptmx}

\usepackage{hyperref}
\colorlet{linkcolor}{blue!50!black}\hypersetup{bookmarksnumbered=false,bookmarksopen=false,bookmarksopenlevel=1,breaklinks=true,pdfborder={0 0 0},colorlinks=true,anchorcolor=linkcolor,citecolor=linkcolor,filecolor=linkcolor,linkcolor=linkcolor,menucolor=linkcolor,runcolor=linkcolor,urlcolor=linkcolor}

 \usepackage[nameinlink,capitalize]{cleveref}

\newtheorem{theorem}{Theorem}
\newtheorem{definition}{Definition}
\newtheorem{lemma}{Lemma}

\newtheorem{proposition}{Proposition}

\usepackage{mathrsfs}
\usepackage{dsfont}
\usepackage{cancel}
\begin{document}

\title{Asymptotic Reversibility of Thermal Operations for Interacting Quantum Spin Systems via Generalized Quantum Stein's Lemma}

\author{Takahiro Sagawa}
\affiliation{Department of Applied Physics, The University of Tokyo,
Tokyo 113-8656, Japan}

\author{Philippe Faist}
\affiliation{Institute for Quantum Information and Matter,
California Institute of Technology, Pasadena, CA 91125, USA}
\affiliation{Institute for Theoretical Physics, ETH Zurich, 8093 Switzerland}
\affiliation{Dahlem Center for Complex Quantum Systems, Freie Universität Berlin, 14195 Berlin, Germany}

\author{Kohtaro Kato}
\affiliation{Institute for Quantum Information and Matter,
California Institute of Technology, Pasadena, CA 91125, USA}

\author{Keiji~Matsumoto}
\affiliation{National Institute of Informatics, 
Tokyo 101-8430, Japan}

\author{Hiroshi Nagaoka}
\affiliation{The University of Electro-Communications, 
Tokyo, 182-8585, Japan}

\author{Fernando G. S. L. Brand\~ao}
\affiliation{Institute for Quantum Information and Matter,
California Institute of Technology, Pasadena, CA 91125, USA}

\date{\today}

\begin{abstract}
    For quantum spin systems in any spatial dimension with a local, translation-invariant Hamiltonian, we prove that asymptotic state convertibility from a quantum state to another one by a thermodynamically feasible class of quantum dynamics, called thermal operations, is completely characterized by the Kullback-Leibler (KL) divergence rate, if the state is translation-invariant and spatially ergodic.
    Our proof consists of two parts and is phrased in terms of a branch of the quantum information theory called the resource theory.
    First, we prove that any states, for which the min and max R\'enyi divergences collapse approximately to a single value, can be approximately reversibly converted into one another by thermal operations with the aid of a small source of quantum coherence.
    Second, we prove that these divergences collapse asymptotically to the KL divergence rate for any translation-invariant ergodic state.  We show this via a generalization of the quantum Stein's lemma for quantum hypothesis testing beyond independent and identically distributed (i.i.d.) situations. 
    Our result implies that the KL divergence rate serves as a thermodynamic potential that provides a complete characterization of thermodynamic convertibility of ergodic states of quantum many-body systems in the thermodynamic limit, including out-of-equilibrium and fully quantum situations.
\end{abstract}

\maketitle

\tableofcontents

\section{Introduction}

Reversibility and irreversibility of dynamics in classical and quantum physics, especially in thermodynamics, is characterized thanks to the concept of entropy.
It is a salient feature of macroscopic equilibrium thermodynamics that entropy does not only have the non-decreasing property but also provides a complete characterization of convertibility between thermal equilibrium states~\cite{BookCallen_Thermo}, which is represented by the second law of thermodynamics.
Lieb and Yngvason constructed an axiomatic formulation of this phenomenology, and within their mathematical framework, rigorously proved that entropy provides a necessary and sufficient condition for state conversion, and furthermore, that such an entropy function is essentially unique~\cite{Lieb1999_secondlaw}.

The connection between microscopic information entropy and thermodynamic entropy has been extensively studied both in terms of statistical mechanics~\cite{Sagawa2012LQC_secondlawlike,Parrondo2015NPhys_thermodynamics}
and the thermodynamic resource theory~\cite{Goold2016JPA_review,Chitambar2019RMP_resource,SagawaBook}.
In the latter formalism which we adopt in this article, so-called one-shot entropy measures have provided tools to quantify resource costs of physical operations in quantum information settings including quantum thermodynamics~\cite{Goold2016JPA_review,Chitambar2019RMP_resource,Brandao2013_resource,Aberg2013_worklike,Horodecki2013_ThermoMaj,Brandao2015PNAS_secondlaws,Gour2018NatComm_entropic,Faist2015NatComm,Faist2018PRX_workcost,Weilenmann2016PRL_axiomatic,PhDMirjam2017,Weilenmann2018_smooth,SagawaBook}. 

Our understanding of the macroscopic behavior of the entropy has been sharpened by fundamental theorems proving asymptotic
equipartition properties (AEP). Rougly speaking, an AEP states that
in the long sequence limit of a stochastic process, some relevant quantities
concentrate to definite values.
For instance, the Shannon-McMillan theorem states that an ergodic process
satisfies an AEP with the Shannon entropy
rate~\cite{Shannon1948BSTJ,BookCoverThomas2006_InfTheory}.  This has been generalized to a
stronger form known as the Shannon-McMillan-Breiman theorem as well as to a
relative version for an ergodic process with respect to a Markov
process~\cite{Algoet1988AoP_sandwich}.  A quantum version of the Shannon-McMillan theorem
proves a similar AEP for quantum ergodic processes with the von Neumann entropy
rate~\cite{Bjelakovic2004IM_lattice,Bjelakovic2005QIP_compression,Ogata2013LMP_shannonmcmillan}.

Closely related to AEP theorems is Stein's lemma, which relates the asymptotic error rate of hypothesis
testing  for distinguishing two quantum states to the KL divergence rate.
Classically, Stein's lemma is a straightforward consequence of the relative AEP.
However, its quantum counterpart is more
involved~\cite{Hiai1991CMP_proper,Ogawa2000IEEETIT_Stein,Bjelakovic2004CMP_ergodic,Brandao2010CMP_Stein,Bjelakovic2003arXiv_revisted}.
Hiai and Petz~\cite{Hiai1991CMP_proper} first addressed the quantum Stein's lemma and
provided a partial proof for a completely ergodic quantum state with respect to
an i.i.d.\@ state.  The proof of the quantum Stein's lemma was completed for the
case where both states are i.i.d.\@ by Ogawa and Nagaoka~\cite{Ogawa2000IEEETIT_Stein}, by
proving the strong converse of the Hiai-Petz theorem for that case.  A more
general form of the quantum Stein's lemma for an ergodic
state with respect to an i.i.d.\@ state was proved in
Ref.~\cite{Bjelakovic2004CMP_ergodic}, which is regarded as a quantum analog of the
relative AEP.

In this work, we go beyond the non-interacting or i.i.d.\@ regime, and investigate an entropy function that provides a thermodynamic characterization of physically relevant, interacting many-body quantum systems.
We consider quantum spin systems on the lattice $\mathbb Z^d$ with an arbitrary number $d$ of spatial dimensions. 
Under certain general conditions, we rigorously prove that the necessary and sufficient condition for asymptotic state conversion from one ergodic state to another state by thermodynamically feasible quantum dynamics, called thermal operations~\cite{Brandao2013_resource}, is characterized by  the Kullback-Leibler (KL) divergence rate of the state relative to the Gibbs state.
The KL divergence rate is shown to determine the work cost for state transformations, and thus plays a role of the proper thermodynamic potential.
Our central assumptions are that (i) the quantum state is  translation-invariant and spatially ergodic and (ii) the Hamiltonian is translation-invariant and local.
Physically, the assumption (i) implies that a quantum state does not exhibit any macroscopic fluctuations if one looks at translation-invariant observables~\cite{BookCoverThomas2006_InfTheory,BookBratteliRobinson_OpAlgQStatMech1,BookBratteliRobinson_OpAlgQStatMech2,BookIsrael_ConvexityTheoryLatticeGases,BookRuelle_StatMechRigorous}, and the assumption (ii) guarantees the sound thermodynamic limit of the Gibbs state.  
Importantly, a spatially ergodic state~--- in contrast to a temporarily ergodic state~--- is not necessarily a thermal equilibrium state, and thus our result is applicable to out-of-equilibrium  situations.

To achieve an operationally robust notion of a thermodynamic potential, we
resort to the resource theory of thermal operations.  The resource theory of
thermal operations is an established model for thermodynamics in the quantum
regime~\cite{Brandao2013_resource,Brandao2015PNAS_secondlaws,Goold2016JPA_review,BookBinder2018_ThermoQuantumRegime}.  This approach allows us to study the
thermodynamic behavior of arbitrary quantum states in a way that inherently
accounts for the fluctuations in the work requirement of state transformations.
This model for thermodynamics is tightly related to measures of information
introduced in quantum information theory based on the quantum R\'enyi
divergences~\cite{Renyi1960_MeasOfEntrAndInf}.  Two quantities in particular,
the R\'enyi-$0$ divergence (or min-divergence) and R\'enyi-$\infty$ divergence
(or max-divergence), play a special role in determining the work requirement of
state transformations~\cite{Lieb2013_entropy_noneq,Weilenmann2016PRL_axiomatic}.  For instance, the work
that can be extracted from any state that is block-diagonal in the energy
eigenspaces is given by the R\'enyi-0 divergence.  For our main result, we
consider the asymptotic version of these quantities for large system sizes,
which corresponds to the thermodynamic limit.
The asymptotic min and max R\'enyi divergences are also called the upper and lower spectral divergence rates in the theory of information
spectrum, and we will use both terms interchangeably in this paper~\cite{BookHan_InfSpecMethods,Han2000IEEETIT_hypothesis,Nagaoka2007IEEETIT_hypothesis,Datta2009IEEE_InfSpec,Datta2009IEEE_minmax,Bowen2006ISIT_beyondiid,Bowen2006arXiv_arbitrary,Schoenmakers2007ISIT_Renyi}.

\

\paragraph{Main result}
Our main result is that ergodic states can be reversibly interconverted into one
another in the resource theory of thermal operations in the thermodynamic limit.
Roughly speaking, if the Hamiltonian is local and translation invariant, then there exists a thermodynamic potential $F(\rho)$ that is defined for all
translation invariant and ergodic states $\rho$ on a lattice of $d$ spatial dimensions with the following property:
For any two translation invariant and ergodic states $\rho,\rho'$, there exists a (generalized) thermal operation
that can carry out the transformation $\rho\to\rho'$ by investing work at a rate
of $F(\rho') - F(\rho)$ per subsystem and that uses a negligible amount of
coherence per subsystem.  
Furthermore, $F(\rho)$ is given by the KL divergence rate between $\rho$ and the Gibbs state $\sigma$ of the Hamiltonian, divided by the temperature of the heat bath.

Our main result is proved in the following two steps.  They are discussed in \cref{sec:thermo-reversibility} and \cref{sec:main-result-for-ergodic-local-Gibbs}, where the main theorems are \cref{thm:asympt-equipartition-implies-asympt-TO-reversibility} and \cref{quantum_theorem_main2t}, respectively.  Both of them can be of independent interest.

First,  we prove that  any state for
which the min and max  R\'enyi divergences coincide approximately~\cite{PhDRenner2005_SQKD,BookTomamichel2016_Finite} can approximately be
converted reversibly to and from the Gibbs state by thermal operations, using a small source of
quantum coherence~\cite{Lostaglio2015PRX_coherence}.  
In this case, the resource theory becomes reversible, i.e., the work required for
a state transformation is equal to the negative work required for the reverse
transformation.  In consequence, if these divergences coincide to a single  value in the asymptotic limit, then it defines  a thermodynamic potential that completely characterizes the possible state transformations in the fully quantum regime. 
This is a result that applies broadly to the resource theory of
thermal operations in general settings, even for states that are non-classical, i.e., that are not block-diagonal in the energy basis.  This intermediate result, which is independent of the assumptions~(i) and~(ii), can be of independent interest.

Second, we prove that the min and max R\'enyi divergences indeed collapse to the KL divergence rate under the assumptions~(i) and~(ii).
To this end, we prove a generalization of the quantum Stein's lemma to the setting with (i) and (ii).
The main idea of our proof, inspired by Refs.~\cite{Bjelakovic2004CMP_ergodic,Bjelakovic2003arXiv_revisted}, is to construct typical projectors that are adapted to the assumptions (i) and (ii).  Our formulation uses semidefinite programming to simplify some parts of the proof.

\

\paragraph{Structure of the paper}
In \cref{sec:preliminaries}, we introduce preliminary definitions and notation, including the relevant divergences and entropy measures. 
In \cref{sec:thermo-reversibility}, we introduce our thermodynamic framework of thermal operations, giving a rigorous meaning to the work cost of a transformation from one state to another, and prove our first main theorem on asymptotic thermal operations
(\cref{thm:asympt-equipartition-implies-asympt-TO-reversibility}).
In \cref{sec:main-result-for-ergodic-local-Gibbs}, we rigorously formulate ergodicity,  and  prove our second main theorem on the generalized quantum Stein's lemma 
(\cref{quantum_theorem_main2t}).  
We conclude with remarks and an outlook in \cref{sec:discussion}.  In the appendices, we remark on some technical lemmas, Gibbs-preserving maps,  a more rigorous approach to ergodicity formulated using $C^\ast$-algebras, an alternative proof of our second main theorem for the one dimensional case, and purely classical implications of our results.

\section{Preliminaries}
\label{sec:preliminaries}

Consider a Hilbert space $\mathscr{H}$ of finite dimension $D$, and let ${\mathcal{S}}(\mathscr{H})$ be
the set of density operators (quantum states) on $\mathscr{H}$, satisfying
$\hat\rho\geqslant0$ and $\operatorname{tr}[ \hat\rho] = 1$ for $\hat\rho\in {\mathcal{S}}(\mathscr{H})$.
We also define the set of \emph{subnormalized states}, which we denote by
${\mathcal{S}_\leq}(\mathcal{H})$, and which is the set of all operators $\hat\rho\geqslant0$
that satisfy $\operatorname{tr}[\hat\rho]\leqslant1$.  For two Hilbert spaces $\mathscr{H}_A$ and $\mathscr{H}_B$
representing systems $A$ and $B$, we write $A\simeq B$ when the Hilbert spaces
are isomorphic; by convention, the identity mapping $A\to B$ maps the canonical
basis of $A$ onto the canonical basis of $B$.

The set of quantum states carries a natural metric given by the \emph{trace
  distance}~\cite{BookNielsenChuang2000}, defined as
$D(\hat\rho,\hat\rho') = (1/2) \lVert {\hat\rho-\hat\rho'}\rVert _{1}$ for any
$\hat\rho,\hat\rho'\in{\mathcal{S}}(\mathscr{H})$, where $\lVert {\cdot}\rVert _{1}$ is the
Schatten 1-norm.  This metric can be extended to subnormalized states
$\hat\rho,\hat\rho'\in{\mathcal{S}_\leq}(\mathscr{H})$ as the \emph{generalized trace distance}~\cite{Tomamichel2010IEEE_Duality,PhDTomamichel2012}, defined
as
\begin{equation}
D(\hat\rho,\hat\rho') = \frac{1}{2} \lVert {\hat\rho-\hat\rho'}\rVert _{1} + \frac{1}{2} \lvert {\operatorname{tr}(\hat\rho) - \operatorname{tr}(\hat\rho')}\rvert .
\end{equation}
We also define the \emph{fidelity}~\cite{BookNielsenChuang2000} as
$F(\hat X,\hat Y) = \lVert {\hat X^{1/2}\hat Y^{1/2}}\rVert _{1}$ for any
$\hat X, \hat Y\geqslant0$.

\subsection{Entropy and divergence}

Thermodynamic properties of microscopic quantum systems can be described using
entropy measures that generalize the usual Shannon or von Neumann entropy to the
so-called ``one-shot''
regime~\cite{PhDRenner2005_SQKD,PhDTomamichel2012,BookTomamichel2016_Finite}.
More specifically, in the presence of thermodynamic reservoirs, we need to
consider a family of relative entropies, or \emph{divergences}.
For $\hat\rho\in{\mathcal{S}_\leq}(\mathscr{H})$ and $\hat\sigma\geqslant0$, the KL divergence
(R\'enyi-$1$ divergence) is defined as:
\begin{align}
  {S}_{1}(\hat\rho\,\Vert\,\hat\sigma) = \operatorname{tr}[\hat\rho\ln \hat\rho- \hat\rho\ln \hat\sigma]\ .
\end{align}
Throughout this paper, we assume that the first argument of the divergences
considered (here $\hat\rho$) lies within the support of the second argument
(here $\hat\sigma$).  This assumption is physically justified when
$\hat\sigma$ is a Gibbs state, which necessarily has full rank.
The min divergence (R\'enyi-$0$ divergence), or the min relative entropy, is defined as
\begin{align}
  {S}_{0}(\hat\rho\,\Vert\,\hat\sigma) = -\ln  \operatorname{tr}[ \hat P_{\rho} \hat\sigma] \ ,
\end{align} 
where $\hat P_\rho$ is the projection onto the support of $\hat\rho$.  We also
define an alternative measure of the min divergence (R\'enyi-$1/2$ divergence)
as 
\begin{align}
  {S}_{1/2}(\hat\rho\,\Vert\,\hat\sigma) =
  -  \ln \, \bigl\lVert { \hat\rho^{1/2}\hat\sigma^{1/2} }\bigr\rVert _{1} ^2\ , 
\end{align}
Finally, the max divergence (R\'enyi-$\infty$
divergence), or the max relative entropy, is defined as
 \begin{align}
   {S}_{\infty}(\hat\rho\,\Vert\,\hat\sigma) = 
   \ln\, \bigl\lVert { \hat\sigma^{-1/2} \, \hat\rho\, \hat\sigma^{-1/2} }\bigr\rVert _{\infty} =
   \ln \min_{\hat\rho\leqslant\lambda \hat\sigma} \lambda \ ,
\end{align}
where $\lVert {\cdot}\rVert _{\infty}$ is the operator norm.

These quantities are special cases of the R\'enyi-$\alpha$ divergences.  Here,
we avoid technicalities and issues in the general definitions of the quantum
R\'enyi divergences caused by the noncommutativity of the
arguments~\cite{Hiai2011RMathPh_fdivergences,Wilde2014CMP_converse,BookTomamichel2016_Finite}, by focusing on
the quantities above which are sufficient for our purposes.
These divergences satisfy
\begin{align}
  \label{eq:divergences-are-ordered-in-alpha}
  -\ln\operatorname{tr}(\hat\sigma) \leqslant{S}_{0}(\hat\rho\,\Vert\,\hat\sigma)
  \leqslant{S}_{1/2}(\hat\rho\,\Vert\,\hat\sigma)
  \leqslant{S}_{1}(\hat\rho\,\Vert\,\hat\sigma) \leqslant{S}_{\infty}(\hat\rho\,\Vert\,\hat\sigma)\ .
\end{align}

From these divergences we can define corresponding entropy measures as the
divergence with respect to the identity operator $\hat{I}$: For
$\alpha=0,1/2,1,\infty$ we define
\begin{align}
  {S}_{\alpha}({\hat\rho}) := -{S}_{\alpha}(\hat\rho\,\Vert\,\hat{I})\ .
\end{align}
We note the following explicit forms of the von Neumann entropy (R\'enyi-$1$
entropy) ${S}_{1}(\hat\rho)$, the max entropy (R\'enyi-$0$
entropy) ${S}_{0}^{}({\hat\rho})$, and the min entropy (R\'enyi-$\infty$ entropy) ${S}_{\infty}^{}({\hat\rho})$,
\begin{align}
  {S}_{1}(\hat\rho) &= - \operatorname{tr}[\hat\rho\ln\hat\rho]\ ;
  &
    {S}_{0}^{}({\hat\rho}) &= \ln\operatorname{rank}(\hat\rho)\ ;
  &
    {S}_{\infty}^{}({\hat\rho}) &= -\ln\,\lVert {\hat\rho}\rVert _{\infty}\ .
\end{align}
The entropies are ordered as
\begin{equation}
0 \leqslant S_\infty (\hat\rho) \leqslant S_1 (\hat\rho) \leqslant S_0 (\hat\rho) \leqslant\ln(D)\ .
\end{equation}

These divergences satisfy the data processing inequality, i.e., they are
monotonous under the action of a completely-positive (CP) and trace-preserving (TP) map
$E$:
\begin{equation}
  \label{eq:Dalpha-dpi}
  {S}_{\alpha}(\hat\rho\,\Vert\,\hat\sigma)
  \geqslant{S}_{\alpha}(E( \hat\rho)\,\Vert\,E( \hat\sigma))\ .
\end{equation}
For $\alpha = 0,1$, see for example  Lemma~7 of Ref.~\cite{Datta2009IEEE_minmax}.
The case of $\alpha = 1$ is equivalent to the strong subadditivity of the von
Neumann entropy~\cite{BookNielsenChuang2000,Lieb1973JMP_SSA}.
Consequently, the entropies do not decrease under the action of a CPTP map $E$ that is unital, i.e.,
$E(\hat{I}) = \hat{I}$,
\begin{equation}
  {S}_{\alpha}({\hat\rho}) \leqslant{S}_{\alpha}({E(\hat\rho)})\ .
\end{equation}

A useful property of these divergences is a monotonicity property for the
semidefinite ordering of the second argument: If $\sigma\leqslant\sigma'$, then for
each $\alpha=0,1/2,1,\infty$,
\begin{align}
  \label{eq:Dalpha-semidef-ordering-second-argument}
  {S}_{\alpha}(\hat\rho\,\Vert\,\hat\sigma') \leqslant{S}_{\alpha}(\hat\rho\,\Vert\,\hat\sigma)\ .
\end{align}

The divergences obey a scaling property in the second argument.  For
$\alpha=0,1/2,1,\infty$, we have for any $a>0$,
\begin{align}
  \label{eq:Dalpha-scaling}
  {S}_{\alpha}(\hat\rho\,\Vert\,a\hat\sigma)
  = {S}_{\alpha}(\hat\rho\,\Vert\,\hat\sigma) - \ln(a)\ .
\end{align}

Under tensor product states, the divergences become additive.  For
$\alpha=0,1/2,1,\infty$, we have for any
$\hat\rho\in{\mathcal{S}_\leq}(\mathscr{H}),\hat\rho'\in{\mathcal{S}_\leq}(\mathscr{H}')$,
$\hat\sigma\geqslant0, \hat\sigma'\geqslant0$,
\begin{align}
  \label{eq:Dalpha-additive-under-tensor-products}
  {S}_{\alpha}(\hat\rho\otimes\hat\rho'\,\Vert\,\hat\sigma\otimes\hat\sigma')
  = 
  {S}_{\alpha}(\hat\rho\,\Vert\,\hat\sigma) + {S}_{\alpha}(\hat\rho'\,\Vert\,\hat\sigma')\ .
\end{align}

To ensure that the operational quantities represented by these entropies and
divergences do not significantly depend on events that only appear with
vanishingly small probability, we ``smoothe'' these entropies and divergences
over a ball of states that are close to the original
state~\cite{PhDRenner2005_SQKD,Datta2009IEEE_minmax}.  First, we define the
$\varepsilon$-ball of states around a subnormalized state $\hat\rho\in{\mathcal{S}_\leq}(\mathscr{H})$
as
\begin{align}
  B^\varepsilon(\hat\rho) := \{
  \hat\tau\in{\mathcal{S}_\leq}(\mathscr{H})  : \  D(\hat\tau, \hat\rho) \leqslant\varepsilon\}\ .
\end{align}

\begin{definition}[Smooth divergences~\protect\cite{Datta2009IEEE_minmax}]
\label{def:smooth-divergences}
  The smooth divergences are defined as follows,
  \begin{subequations}
    \begin{align}
      {S}_{\infty}^{\varepsilon}(\hat\rho\,\Vert\,\hat\sigma)
      &:= \min_{\hat\tau\in B^\varepsilon(\hat\rho) } {S}_{\infty}(\hat\tau\,\Vert\,\hat\sigma)\ ;
        \label{eq:def-smooth-Dmax}
      \\
      {S}_{0}^{\varepsilon}(\hat\rho\,\Vert\,\hat\sigma)
      &:= \max_{\hat\tau\in B^\varepsilon(\hat\rho) } {S}_{0}(\hat\tau\,\Vert\,\hat\sigma)\ ;
        \label{eq:def-smooth-Dminz}
      \\
      {S}_{1/2}^{\varepsilon}(\hat\rho\,\Vert\,\hat\sigma)
      &:= \max_{\hat\tau\in B^\varepsilon(\hat\rho) } {S}_{1/2}(\hat\tau\,\Vert\,\hat\sigma)\ .
        \label{eq:def-smooth-Dminf}
    \end{align}
  \end{subequations}
  The smooth entropies are defined correspondingly as
  \begin{align}
    {S}_{0}^{\varepsilon}({\hat\rho}) &:= -{S}_{0}^{\varepsilon}(\hat\rho\,\Vert\,\hat{I})\ ;
    &
    {S}_{\infty}^{\varepsilon}({\hat\rho}) &:= -{S}_{\infty}^{\varepsilon}(\hat\rho\,\Vert\,\hat{I})\ .
  \end{align}
\end{definition}

We introduce a further convenient divergence (relative entropy) that is based on
hypothesis testing~\cite{Tomamichel2013_hierarchy,Dupuis2013_DH,Faist2018PRX_workcost}.  This
divergence allows to interpolate between the min- and max-divergences in a
different fashion than the R\'enyi entropies, along with a simple formulation
and a collection of useful properties.  For a subnormalized state $\hat\rho$
and $\hat\sigma\geqslant0$, we define for any $0<\eta\leqslant\operatorname{tr}(\hat\rho)$,
\begin{align}
  \label{def_SH}
  {S}_{\mathrm{H}}^{\eta}(\hat\rho\,\Vert\,\hat\sigma)
  := - \ln \left( \eta^{-1}\min_{0 \leqslant\hat Q \leqslant\hat{I}, \ \operatorname{tr}[\hat\rho\hat Q] \geqslant\eta }
    \operatorname{tr}[\hat\sigma\hat Q] \right)\ .
\end{align}
The hypothesis testing divergence owes its name to the fact that if
$\hat\rho,\hat\sigma$ are two quantum states,
$\eta\exp(-{S}_{\mathrm{H}}^{\eta}(\hat\rho\,\Vert\,\hat\sigma))$ represents the probability of
mistakenly reporting $\hat\rho$ in a hypothesis test between the two states, if
we carry out a strategy that mistakenly reports $\hat\sigma$ with probability
at most $1-\eta$.

The hypothesis testing divergence satisfies the data processing
inequality~\cite{Wang2012PRL_oneshot}: For any subnormalized state $\hat\rho$, for any
$\hat\sigma\geqslant 0$, for any CP and trace-nonincreasing map
$E$,  and for any $0<\eta\leqslant\operatorname{tr}(E(\hat\rho))$, 
the hypothesis testing divergence is monotonic,
\begin{align}
  {S}_{\mathrm{H}}^{\eta}(\hat\rho\,\Vert\,\hat\sigma) \geqslant{S}_{\mathrm{H}}^{\eta}(E(\hat\rho)\,\Vert\,E(\hat\sigma))\ .
  \label{eq:DHyp-dpi}
\end{align}
The hypothesis testing entropy also obeys a scaling property in the second
argument: For any subnormalized state $\hat\rho$, for any
$\hat\sigma\geqslant 0$, and for any $0<\eta\leqslant\operatorname{tr}(\hat\rho)$,
\begin{align}
  {S}_{\mathrm{H}}^{\eta}(\hat\rho\,\Vert\,a\hat\sigma) = {S}_{\mathrm{H}}^{\eta}(\hat\rho\,\Vert\,\hat\sigma) - \ln(a)\ ,
  \label{eq:DHyp-scaling}
\end{align}
as can be directly seen from~\eqref{def_SH}.  Also, for any
$\hat\sigma,\hat\sigma'\geqslant0$ for which $\hat\sigma\leqslant\hat\sigma'$, the
hypothesis testing entropy satisfies
\begin{align}
  {S}_{\mathrm{H}}^{\eta}(\hat\rho\,\Vert\,\hat\sigma') \leqslant{S}_{\mathrm{H}}^{\eta}(\hat\rho\,\Vert\,\hat\sigma)\ ,
  \label{eq:DHyp-semidef-ordering-second-argument}
\end{align}
for any subnormalized state $\hat\rho$ and for any $0<\eta\leqslant\operatorname{tr}(\hat\rho)$.
Furthermore, if $D(\hat\rho',\hat\rho)\leqslant\varepsilon$, then
$\hat\rho'\geqslant\hat\rho-\hat \Delta$ for some $\hat \Delta\geqslant0$ with
$\operatorname{tr}(\hat \Delta)\leqslant\varepsilon$ and hence for any
$0<\eta\leqslant\eta+\varepsilon\leqslant\operatorname{tr}(\hat\rho)$,
\begin{align}
  {S}_{\mathrm{H}}^{\eta+\varepsilon}(\hat\rho\,\Vert\,\hat\sigma)
  \leqslant{S}_{\mathrm{H}}^{\eta}(\hat\rho'\,\Vert\,\hat\sigma)
  + \ln\Bigl(\frac{\eta+\varepsilon}{\eta}\Bigr)\ .
  \label{eq:DHyp-perturbation-state}
\end{align}

A useful property of the hypothesis testing divergence is that it interpolates
between the min and max divergences, which are approximately recovered in the
regimes $\eta \simeq 0$ and $\eta \simeq 1$,
respectively~\cite{Dupuis2013_DH}:
\begin{proposition}
  \label{Faist_prop}
  Let $\hat\rho$ be a (normalized) quantum state and let $\hat\sigma\geqslant0$.  For any
  $0<\varepsilon<1/2$,
  \begin{subequations}
    \begin{gather}
      {S}_{\mathrm{H}}^{{1-\varepsilon^2 / 6}}(\hat\rho\,\Vert\,\hat\sigma)
      - \ln \left(\frac{1-\varepsilon^2 / 6}{\varepsilon^2 / 6}\right)
      \leqslant{S}_{0}^{\varepsilon}(\hat\rho\,\Vert\,\hat\sigma)
      \leqslant{S}_{\mathrm{H}}^{{1-\varepsilon}}(\hat\rho\,\Vert\,\hat\sigma) - \ln (1-\varepsilon)\ ;
      \\
      {S}_{\mathrm{H}}^{2\varepsilon}(\hat\rho\,\Vert\,\hat\sigma) - \ln(2)
      \leqslant{S}_{\infty}^{\varepsilon}(\hat\rho\,\Vert\,\hat\sigma)
      \leqslant{S}_{\mathrm{H}}^{{\varepsilon^2 / 2}}(\hat\rho\,\Vert\,\hat\sigma)\ .
    \end{gather}
  \end{subequations}
\end{proposition}
\begin{proof}
  The proof of~\cite[Lemma~40]{Faist2018PRX_workcost} carries through even for the slightly
  different smoothing of ${S}_{0}$ and ${S}_{\infty}$, except for the upper
  bound on ${S}_{\infty}$.  There, we may
  apply~\cite[Proposition~4.1]{Dupuis2013_DH} directly.
\end{proof}

Finally, we note a pair of inequalities which establishes the approximate
equivalence of the two kinds of
min-divergences~\cite{Tomamichel2011TIT_LeftoverHashing,Dupuis2013_DH,BookTomamichel2016_Finite}.
\begin{proposition}
  \label{prop:equiv-Dzero-Dhalf}
  Let $\hat\rho$ be a normalized state and let $\hat\sigma\geqslant0$.  For any $\varepsilon>0$,
  \begin{align}
    {S}_{1/2}^{2\varepsilon}(\hat\rho\,\Vert\,\hat\sigma)
    \geqslant
    {S}_{0}^{2\varepsilon}(\hat\rho\,\Vert\,\hat\sigma)
    \geqslant {S}_{1/2}^{\varepsilon}(\hat\rho\,\Vert\,\hat\sigma) - 6\ln\left(\frac3{\varepsilon}\right)\ .
  \end{align}
\end{proposition}
\begin{proof}
  The first inequality follows because
  of~\eqref{eq:divergences-are-ordered-in-alpha}.  For the second inequality,
  let $\hat\rho'\in B^\varepsilon(\hat\rho)$ such that
  ${S}_{1/2}^{\varepsilon}(\hat\rho\,\Vert\,\hat\sigma)={S}_{1/2}(\hat\rho'\,\Vert\,\hat\sigma)$.
  Then from~\cite[Proposition~4.2]{Dupuis2013_DH}, we have
  ${S}_{1/2}(\hat\rho'\,\Vert\,\hat\sigma) \leqslant{S}_{\mathrm{H}}^{1-\varepsilon'}(\hat\rho'\,\Vert\,\hat\sigma) - \ln(\varepsilon'^2)$ for any
  $\varepsilon'>0$; choosing $\varepsilon'=\varepsilon^2/6$ and using \cref{Faist_prop},
  we find
  ${S}_{1/2}(\hat\rho'\,\Vert\,\hat\sigma) \leqslant{S}_{0}^{\varepsilon}(\hat\rho'\,\Vert\,\hat\sigma) + \ln[(1-\varepsilon')/\varepsilon'] -
  \ln(\varepsilon'^2)$.  The claim follows by noting that
  ${S}_{0}^{\varepsilon}(\hat\rho'\,\Vert\,\hat\sigma) \leqslant{S}_{0}^{2\varepsilon}(\hat\rho\,\Vert\,\hat\sigma)$ along with
  $(1-\varepsilon')/(\varepsilon'^3) \leqslant(6\varepsilon^{-2})^{3} \leqslant(3\varepsilon^{-1})^{6}$.
\end{proof}

\subsection{Asymptotic spectral divergence rates}

In statistical mechanics one is often interested in the thermodynamic limit,
where the behavior of the system as it becomes arbitrarily large often no longer
depends on microscopic details.
The action of taking the thermodynamic limit is formalized by considering a
sequence of states $\widehat{P}:= \{ \hat\rho_n \}_{n\in\mathbb{N}}$, where
$\hat\rho_n$ is a quantum state on $\mathscr{H}^{\otimes n}$.

The von Neumann entropy rate is defined as
\begin{align}
  \label{entropy_rate_q}
  {S}_{1}(\widehat{P}) := \lim_{n \to \infty} \frac{1}{n} {S}_{1}(\hat\rho_n)\ ,
\end{align}
and the KL divergence rate with respect to the sequence of
positive operators $\widehat{\Sigma}:= \{ \hat\sigma_n \}_{n \in \mathbb N}$ is
defined as
\begin{align}
  \label{divergence_rate_q}
  {S}_{1}(\widehat{P}\,\Vert\,\widehat{\Sigma})
  := \lim_{n \to \infty} \frac{1}{n} {S}_{1}(\hat\rho_n\,\Vert\,\hat\sigma_n)\ .
\end{align}
We note that these limits do not necessarily exist in general.

We now introduce the spectral divergence rates, which are natural extensions of
the min and max divergences to the thermodynamic limit.

\begin{definition}[Spectral divergence rates]
  \label{def_spectral}
  Let $\widehat{P}= \{ \hat\rho_n \}$ be a sequence of states and let
  $\widehat{\Sigma}= \{ \sigma_n \}$ be a sequence of positive operators.  We define
  the upper spectral divergence rate,
  \begin{align}
    {\overline{S}}(\widehat{P}\,\Vert\,\widehat{\Sigma})
    := \lim_{\varepsilon\to +0} \limsup_{n \to \infty}
    \frac{1}{n} {S}_{\infty}^{\varepsilon}(\hat\rho_n\,\Vert\,\hat\sigma_n)\ ,
  \end{align}
  and the lower spectral divergence rate,
  \begin{align}
    {\underline{S}}(\widehat{P}\,\Vert\,\widehat{\Sigma})
    := \lim_{\varepsilon\to +0}\liminf_{n \to \infty}
    \frac{1}{n} {S}_{0}^{\varepsilon}(\hat\rho_n\,\Vert\,\hat\sigma_n)\ .
  \end{align}
\end{definition}

These quantities have been introduced in Ref.~\cite{Nagaoka2007IEEETIT_hypothesis} in an equivalent but different  expression:
\begin{subequations}
  \begin{align}
    {\overline{S}}(\widehat{P}\,\Vert\,\widehat{\Sigma})
    &= \inf \left\{ a \ : \ \limsup_{n \to \infty} 
      \operatorname{tr}\bigl[ \operatorname{Proj}\left\{ \hat\rho_n - e^{na} \hat\sigma_n \geqslant0 \right\} \hat\rho_n \bigr]
      = 0 \right\}\ ,
      \label{spectral_Nagaoka1}
    \\
    {\underline{S}}(\widehat{P}\,\Vert\,\widehat{\Sigma})
    &= \sup \left\{ a \ : \ \liminf_{n \to \infty}
      \operatorname{tr}\bigl[ \operatorname{Proj}\left\{ \hat\rho_n - e^{na} \hat\sigma_n \geqslant0 \right\} \hat\rho_n \bigr]
      = 1 \right\}\ ,
      \label{spectral_Nagaoka2}
  \end{align}
\end{subequations}
where $\operatorname{Proj}\bigl\{ \hat X \geqslant0\bigr\}$ represents the projector onto the
eigenspaces of $\hat X$ corresponding to nonnegative eigenvalues.  The equivalence
of these two definitions has been proved in Theorems~2 and~3 of
Ref.~\cite{Datta2009IEEE_minmax}.  We note that
\begin{align}
  {\underline{S}}(\widehat{P}\,\Vert\,\widehat{\Sigma}) \leqslant{\overline{S}}(\widehat{P}\,\Vert\,\widehat{\Sigma})\ .
  \label{eq:min-max-asympt-divergence-rates-ordered}
\end{align}

As a special case, we introduce the lower and the upper spectral entropy rates,
which are respectively given by
\begin{align}
  \underline{S}({\widehat{P}}) &:= -{\overline{S}}(\widehat{P}\,\Vert\,\widehat{\mathrm{ID}})\ ;
  &
  \overline{S}({\widehat{P}}) &:= -{\underline{S}}(\widehat{P}\,\Vert\,\widehat{\mathrm{ID}})\ ,
\end{align}
where $\widehat{\mathrm{ID}} := \{ \hat{I}^{\otimes n} \}_{n \in \mathbb N}$ is
the sequence consisting of identity operators on $\mathscr{H}^{\otimes n}$.

We can also define the hypothesis testing divergence rate
\begin{align}
  {S}_{\mathrm{H}}^{\eta}(\widehat{P}\,\Vert\,\widehat{\Sigma})
  := \lim_{n \to \infty} \frac{1}{n} {S}_{\mathrm{H}}^{\eta}(\hat\rho_n\,\Vert\,\hat\sigma_n)\ ,
\end{align}
noting that the limit does not necessarily exist.  From \cref{Faist_prop}, in
general, ${S}_{\mathrm{H}}^{\varepsilon}(\hat\rho\,\Vert\,\hat\sigma)$ and
${S}_{\mathrm{H}}^{1-\varepsilon}(\hat\rho\,\Vert\,\hat\sigma)$ respectively give the same lower and
upper spectral divergence rates as those given by
${S}_{\infty}^{\varepsilon}(\hat\rho\,\Vert\,\hat\sigma)$ and
${S}_{0}^{\varepsilon}(\hat\rho\,\Vert\,\hat\sigma)$:
\begin{subequations}
  \label{hypothesis_spectral_eq}
  \begin{align}
    \lim_{\varepsilon\to +0}\limsup_{n \to \infty}
    \frac{1}{n} {S}_{\mathrm{H}}^{\varepsilon}(\hat\rho_n\,\Vert\,\hat\sigma_n)
    &= {\overline{S}}(\widehat{P}\,\Vert\,\widehat{\Sigma})\ ,\\
    \lim_{\varepsilon\to +0}\liminf_{n \to \infty}
    \frac{1}{n} {S}_{\mathrm{H}}^{1-\varepsilon}(\hat\rho_n\,\Vert\,\hat\sigma_n)
    &= {\underline{S}}(\widehat{P}\,\Vert\,\widehat{\Sigma})\ .
  \end{align}
\end{subequations}

\section{Asymptotic state convertibility by thermal operations}
\label{sec:thermo-reversibility}

In this section, we formulate thermal operations and prove our first main theorem on asymptotic state convertibility (\cref{thm:asympt-equipartition-implies-asympt-TO-reversibility}).
Importantly, in the microscopic regime, state transformations are not
reversible in general, not even approximately. 
For general states
$\hat\rho,\hat\rho'$, it might happen that
$\hat\rho$ can be approximately converted to  $ \hat\rho'$ with work extraction $w$, but that an approximate
transformation from $\hat\rho'$ to $\hat\rho$ requires much more work than
$w$~\cite{Horodecki2013_ThermoMaj}.

Then we can ask the question, under which conditions is reversibility restored?
This is an important question, because reversibility implies that the optimal
work cost derives from a potential, which in turn means that macroscopic
thermodynamic behavior is restored.
Here, we consider in fact a marginally stronger property.  Under which
conditions is a state reversibly convertible to the thermal state?  Clearly, any
two states that have this property can reversibly be converted into one another.
This slightly stronger statement ensures that the thermodynamic potential is
well defined for the thermal state itself, a desirable feature that allows the
thermal state to take on the role of a ``reference state.''

\subsection{Thermodynamic operations}
\label{sec:thermodynamic-operations}

We now introduce our thermodynamic framework.  The simple model we introduce
captures the relevant features of thermodynamics at the microscopic scale, while
providing a simple, abstract, and general formalism for analyzing the resource
cost of transforming one quantum state into another~\cite{Goold2016JPA_review}.

The goal is the following. Given a system $S$, and two states
$\hat\rho_S,\hat\rho'_S$, we would like to quantify the resources required in
order to convert $\hat\rho_S$ to $\hat\rho'_S$ in some reasonable thermodynamic
model. 
The resource theory of thermal operations is an established model that is particularly useful in such a context.
It specifies the set of transformations that can be carried out for free, without the involvement of external resources such as thermodynamic work.  In the model of thermal operations, one is allowed to carry out for free any unitary on the system and a heat bath at fixed background temperature, as long as the unitary commutes with the overall noninteracting Hamiltonian of the system and the bath.
Here we introduce a slightly generalized notion of thermal operations, where different input and output systems are allowed.

\begin{definition}[(Generalized) thermal Operation]
  \label{def:thermal-operation}
  Consider systems $S,S'$ with corresponding Hamiltonians $\hat{H}_S, \hat{H}'_{S'}$.
  Then a CP and trace-nonincreasing map $\Phi^{[\mathrm{TO}]}_{S\to S'}(\cdot)$
  is a \emph{thermal operation} at inverse temperature $\beta > 0$ if it can be
  written as
  \begin{align}
    \label{eq:def-TO}
    \Phi^{[\mathrm{TO}]}_{S\to S'}(\cdot)
    = \operatorname{tr}_{B} \left[ \hat V_{SB\to S'B}
    \, \left( (\cdot)\otimes
    \frac{{e}^{-\beta \hat{H}_{B}}}{\operatorname{tr}\bigl(e^{-\beta \hat{H}_B}\bigr)}
    \right)
    \, \hat V_{SB\leftarrow S'B}^\dagger \right]\ ,
  \end{align}
  for some ancilla system $B$ of finite dimension with some
  corresponding Hamiltonian $\hat{H}_B$, and for some partial isometry
  $\hat V_{SB\to S'B}$ such that
  $\hat V_{SB\to S'B} \, ({\hat{H}_S+\hat{H}_B})=({\hat{H}'_{S'} + \hat{H}_B}) \,
  \hat V_{SB\to{}S'B}$.

  If there exists a thermal operation that maps $\hat\rho_S$ to
  $\hat\rho'_{S'}$,
  we write $(\hat\rho_S,\hat{H}_S) \xrightarrow[\mathrm{TO}]{} (\hat\rho'_{S'},\hat{H}'_S)$.  We may
  omit the Hamiltonians if they are clear from context.

  Furthermore, a process that is achieved in the limit of processes of the
  form~\eqref{eq:def-TO} with arbitrarily large but finite bath systems, is also
  called a thermal operation.
\end{definition}

The last condition is required to enable processes that decrease the rank of the
input state, for instance, a process consisting of Landauer erasure of a single
bit compensated by a suitable energy shift~\cite{Horodecki2013_ThermoMaj}.

An operator $\hat V$ is a partial isometry if it is an isometry on its support,
or equivalently if $\hat V^\dagger \hat V$ and $\hat V \hat V^\dagger$ are
projectors.
We allow $\hat V$ in the definition above to be a partial isometry instead of a
unitary as considered in
Refs.~\cite{Brandao2013_resource,Horodecki2013_ThermoMaj,Brandao2015PNAS_secondlaws} because they are more convenient when considering
input and output systems of different dimension.  Physically, this corresponds
to specifying only a part of the process happening on an input subspace.
Importantly, any partial isometry that conserves energy can be dilated to a full
unitary that conserves energy on a larger
system~\cite{Faist2021CMP_impl}, as illustrated in Fig.~\ref{fig}.  We prove a corresponding general statement  as
\cref{lemma:dilation-energy-conserving-partial-isometry} in
\cref{appx:thermodynamic-lemmas}. 

\begin{figure}[h]
 \begin{center}
  \includegraphics[width=8cm]{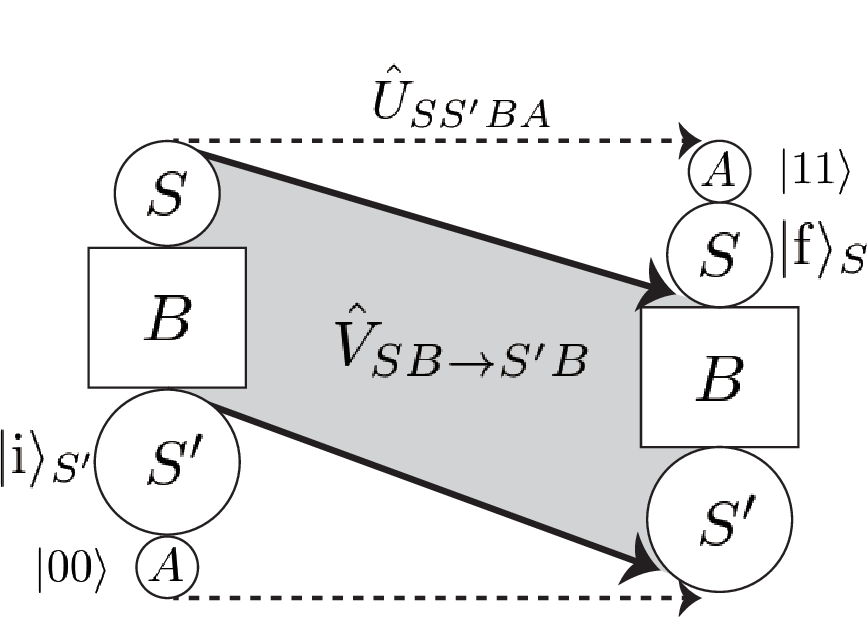}
 \end{center}
 \caption{A schematic of a generalized thermal operation (Definition~\ref{def:thermal-operation}). The partial isometry $\hat V_{SB \to S'B}$ can be embedded into the energy-conserving unitary $\hat U_{SS'BA}$ of the composite system $SS'B$
 along with a 2-qubit auxiliary system $A$.
 The initial state of $S'A$ is chosen as a pure state $| {\rm i} \rangle_{S'}| 00 \rangle$ and the final state of $SA$ is projected onto $| {\rm f} \rangle_S | 11 \rangle$.
 Here, $|{\rm i}\rangle_{S'}$ and $| {\rm f} \rangle_S$ can be arbitrary energy eingenstates of $S'$ and $S$, and the Hamiltonian of $A$ is chosen as a function of these energies to ensure the global energy conservation.
 }
\label{fig}
\end{figure}

There are no known general conditions under which state transformations are
possible with thermal operations in the quantum regime.  For semiclassical states, i.e.\@ states that
are block-diagonal in energy, such conditions are provided in the form of
\emph{thermomajorization}, a generalization of matrix
majorization~\cite{Horodecki2013_ThermoMaj}. 

Now we introduce an alternative model known as \emph{Gibbs-preserving maps}.
This model has a simple technical formulation which makes it more convenient to
prove some properties.  Because any thermal operation is in particular a
Gibbs-preserving map, all properties obeyed by Gibbs-preserving maps are
inherited by thermal operations.
As for thermal operations, it is technically more convenient to consider
trace-nonincreasing maps; furthermore we allow these maps to be
Gibbs-sub-preserving in the sense of the following definition.

\begin{definition}[Gibbs-sub-preserving map]
  Consider systems $S,S'$ with corresponding Hamiltonians $\hat{H}_S, \hat{H}'_{S'}$.
  Then a 
CP and trace-nonincreasing
map
  $\Phi^{[\mathrm{GPM}]}_{S\to S'}(\cdot)$ is said to be a
  \emph{Gibbs-sub-preserving map}  for some fixed
  inverse temperature $\beta$ if
  \begin{align}
    \label{eq:def-GPM}
    \Phi^{[\mathrm{GPM}]}_{S\to S'}\bigl(e^{-\beta \hat{H}_S}\bigr)
    \leqslant e^{-\beta\hat{H}'_{S'}}\ .
  \end{align}
  When there exists a Gibbs-sub-preserving map that maps $\hat\rho_S$ to
  $\hat\rho'_S$, we write
  $(\hat\rho_S;\hat{H}_S) \xrightarrow[\mathrm{GPM}]{} (\hat\rho'_{S'};\hat{H}'_{S'})$.  We may omit
  the Hamiltonians if they are clear from context.
\end{definition}

We note that any Gibbs-sub-preserving map can be dilated into a fully
trace-preserving map on a larger system which furthermore has the thermal state
as a fixed point~\cite[Proposition~2]{Faist2018PRX_workcost}.

\begin{lemma}
  \label{lemma:TO-implies-GPM} 
  Any thermal operation is also a Gibbs-sub-preserving map.
\end{lemma} 
\begin{proof}
  A thermal operation $\Phi^{[\mathrm{TO}]}_{S\to S'}$ can be written in the
  form~\eqref{eq:def-TO}.  We abbreviate $\hat V_{SB\to S'B}$ as $\hat V$.  Then
  with $Z_B = \operatorname{tr}(e^{-\beta\hat{H}_B})$, we have
  \begin{align}
  \Phi^{[\mathrm{TO}]}_{S\to S'}(e^{-\beta\hat{H}_S}) = Z_B^{-1}\,\operatorname{tr}_B\bigl[\hat V e^{-\beta(\hat{H}_S + \hat{H}_B)} \hat V^\dagger \bigr] = Z_B^{-1} \operatorname{tr}_B\bigl[ e^{-\beta \, \hat V(\hat{H}_S + \hat{H}_B) \hat V^\dagger} \bigr] \leqslant e^{-\beta\hat{H}'_{S'}},
  \end{align}
  where we have invoked \cref{prop:energypres-partial-isom-V-commutes-with-H} to
  see that $\hat V^\dagger \hat V$ commutes with $\hat{H}_S+\hat{H}_B$ (for the second
  equality) and that $\hat V \hat V^\dagger$ commutes with $\hat{H}'_{S'}+\hat{H}_B$
  (for the final inequality).
\end{proof}

While any thermal operation is a Gibbs-sub-preserving map as shown in \cref{lemma:TO-implies-GPM}, the converse is not true~\cite{Faist2015NJP_Gibbs}.
A notable difference between thermal operations and Gibbs-preserving maps is the
way the two models handle coherent superpositions of energy states. 
Thermal operations cannot create any coherent superpositions of energy levels
because they commute with time evolution.  However, there exist Gibbs-preserving
maps that can generate coherent superpositions of energy
levels~\cite{Faist2015NJP_Gibbs}.

The divergences defined above play an important role in our thermodynamic
framework as they are monotones under thermodynamic transformations.
In the following, we exploit the scaling property~\eqref{eq:Dalpha-scaling} of the
divergences to write the expression
${S}_{\alpha}(\hat\rho_S\,\Vert\,e^{-\beta\hat{H}_S}/Z_S) - \ln(Z_S) =
{S}_{\alpha}(\hat\rho_S\,\Vert\,e^{-\beta\hat{H}_S})$ more compactly by absorbing the
system free energy into the divergence term.

\begin{proposition}[Monotonicity of divergences~\cite{Horodecki2013_ThermoMaj,Brandao2015PNAS_secondlaws,Faist2018PRX_workcost,BookTomamichel2016_Finite}]
  \label{prop:relative_monotone}
  Consider systems $S,S'$ with corresponding Hamiltonians $\hat{H}_S,\hat{H}'_{S'}$.
  If $\hat\rho_S,\hat\rho_{S'}'$ are (normalized) quantum states that satisfy
  $\hat\rho_S \xrightarrow[*]{} \hat\rho_{S'}'$, where $*$ stands for either TO or GPM,
  then
  \begin{align}
    {S}_{\alpha}(\hat\rho_S\,\Vert\,e^{-\beta\hat{H}_S}) &\geqslant{S}_{\alpha}(\hat\rho_{S'}'\,\Vert\,e^{-\beta\hat{H}'_{S'}})\ ;
    &
    &\text{and}
    &
    {S}_{\mathrm{H}}^{\eta}(\hat\rho_S\,\Vert\,e^{-\beta\hat{H}_S}) &\geqslant{S}_{\mathrm{H}}^{\eta}(\hat\rho_{S'}'\,\Vert\,e^{-\beta\hat{H}'_{S'}})\ ,
  \end{align}
  where $\alpha$ may be any of $0$, $1/2$, $1$, or $\infty$ and where
  $0<\eta\leqslant1$.
\end{proposition}

The proof of \cref{prop:relative_monotone} is essentially an application of the
data processing inequality~\eqref{eq:Dalpha-dpi}.  The full proof requires a
dilation of the trace-nonincreasing map into a trace-preserving one, and it is
presented in \cref{appx:thermodynamic-lemmas}.

Now that we have specified the free operations, we need to specify how we can
provide resources for thermodynamic operations that are not free, or how we can
extract such resources from states.

\emph{Thermodynamic work} can be provided with the help of an external work
storage system, often called a ``battery.''  This can be any system which starts
in a definite energy level and finishes in a different energy level; the
difference in energy is then the amount of work furnished or extracted.  In
fact, a large collection of different battery models are
equivalent~\cite{Brandao2015PNAS_secondlaws,Faist2018PRX_workcost}.

Thermal operations necessarily commute with the free time evolution, as can be seen
from~\eqref{eq:def-TO}.
This means that it is impossible to create any state that has a coherent
superposition of energy levels, even with an arbitrary amount of work, without
access to another resource that provides
coherence~\cite{Lostaglio2015PRX_coherence}.  Coherence is thus a valuable
resource that should be accounted
for~\cite{Aberg2014PRL_catalytic,Lostaglio2015PRX_coherence,Korzekwa2016NJP_extraction,Winter2016PRL_coherence,Marvian2020NC_distillation}.  Here, we adopt a rudimentary, \emph{ad hoc}
model.  We suppose that we have access to an additional system $C$ initialized
into a pure state of our choosing.  Crucially, we assume that the range of
energy values that can be stored into the system $C$ is bounded by some
parameter $\eta$, i.e., $\lVert {\hat{H}_C}\rVert _{\infty}\leqslant \eta$ where $\hat{H}_C$ is the
Hamiltonian of $C$.  The system $C$ must be restored to a state that is close to
a pure state.  The bound on the norm of the Hamiltonian forbids any embezzlement
of work of more than of the order of $\eta$~\cite{Brandao2015PNAS_secondlaws}.
The requirement that the final state on
$C$ is close to a pure state is necessary because there is no constraint on the
dimensionality of $C$; with a suitable highly degenerate system, starting from a
pure state and finishing in the maximally mixed state would allow to extract an
arbitrary amount of work that is not controlled by $\eta$.

This crude model for accounting for coherence suffices for our purposes, as the
protocols we construct only require an ancilla system $C$ with a parameter
$\eta$ that is negligibly small compared to the overall work cost of the
transformation.  Note that this scheme differs from
catalysis~\cite{Brandao2015PNAS_secondlaws,NellyNg2015NJP_limits,Lostaglio2015PRL_absence} as we do not require the final
state to be related in any way to the initial state.

\begin{definition}[Work/coherence-assisted process]
\label{def:workcoherence}
  Consider systems $S,S'$ with corresponding Hamiltonians $\hat{H}_S,\hat{H}'_{S'}$ and let $*$ stand for TO or GPM.  We say that a CP and trace-nonincreasing map $\Phi_{S\to S'}$ is a
  \emph{$(w,\eta)$-work/coherence-assisted $*$ operation},
  if there
  exist systems $W,C,W',C'$ with respective Hamiltonians $\hat{H}_W, \hat{H}_C,\hat{H}_{W'},\hat{H}_{C'}$ satisfying
  $\lVert {\hat{H}_C}\rVert _{\infty}\leqslant\eta$, $\lVert {\hat{H}_{C'}}\rVert _{\infty}\leqslant\eta$,
  and if there exist two energy eigenstates
  $\lvert {E}\rangle _W,\lvert {E'}\rangle _{W'}$ of $\hat{H}_W,\hat{H}_{W'}$ respectively whose energies
  $E$ and $E'$ satisfy $E - E' = w$, 
  and if there exist two pure states $\lvert {\zeta}\rangle _C, \lvert {\zeta'}\rangle _{C'}$, 
  and if there exists a $*$ operation
  $\tilde\Phi^{[*]}_{SCW\to S'C'W'}$, such that
  \begin{align}
    \Phi_{S\to S'}(\hat \rho_S) = \operatorname{tr}_{C'W'}\mathopen{}\left[
    \lvert {E'}\rangle\hspace*{-0.25ex}\langle{E'}\rvert _{W'} \otimes \lvert {\zeta'}\rangle\hspace*{-0.25ex}\langle{\zeta'}\rvert _{C'} \;
    \tilde\Phi^{[*]}_{SCW\to S'C'W'}
    \Bigl( \hat \rho_S \otimes \lvert {E}\rangle\hspace*{-0.25ex}\langle{E}\rvert _W\otimes \lvert {\zeta}\rangle\hspace*{-0.25ex}\langle{\zeta}\rvert _C \Bigr)
    \right]\ .
    \label{eq:workcoherence-process}
  \end{align}
  Here, we allow infinite-dimensional Hilbert spaces for $C$ and $C'$ for
    technical reasons related to how to construct $\lvert {\zeta}\rangle _C$ states.
\end{definition}

A $(w,\eta)$-work/coherence-assisted thermal operation is thus simply a free process that
is assisted by ancillas that provide an amount of work $w$ and an ``amount of
coherence'' that is at most $\eta$.  If $w$ is negative, then this measures the
amount of work that is extracted by the process.

\begin{definition}[Approximate thermodynamic process using work and coherence]
\label{def:workcoherence-trans}
  Consider systems $S,S'$ with Hamiltonians $\hat{H}_S,\hat{H}'_{S'}$ and let $*$ stand for TO or GPM.  We say that the state $\hat\rho_S$ is
  \emph{$(w,\eta,\varepsilon)$-transformable into $\hat\rho'_{S'}$ by a $*$ process}, which we denote by
  $(\hat\rho_S;\hat{H}_S) \xrightarrow[\mathrm{*}]{w,\eta,\varepsilon} (\hat\rho_S';\hat{H}'_{S'})$, if
  there exists a $(w,\eta)$-work/coherence-assisted $*$ process $\Phi_{S \to S'}$  such that $D(\Phi_{S \to S'}(\hat\rho_S), \hat\rho'_S) \leqslant \varepsilon$.  We may
  omit the Hamiltonians if they are clear from context.
\end{definition}

The hypothesis testing divergence is a relatively good (quasi) monotone under
assisted thermodynamic operations: It can only decrease, except for correction
terms that depend on $w,\eta,\varepsilon$.  Because the proof is not particularly
insightful, we defer it to \cref{appx:thermodynamic-lemmas}.

\begin{proposition}[Quasi-monotonicity of the hypothesis testing divergence
  under resource-assisted transformations]
  \label{prop:quasi-monotonicity-dhyp-assisted-thermal-transformation}
  Consider systems $S,S'$ with respective Hamiltonians $\hat{H}_S, \hat{H}'_{S'}$.
  For a quantum state $\hat\rho_S$ and a subnormalized state $\hat\rho'_{S'}$,
  suppose $\hat\rho_S\xrightarrow[\mathrm{*}]{w,\,\eta,\,\varepsilon} \hat\rho'_{S'}$, where $*$
  stands for TO or GPM.  Then for any
  $0<\xi\leqslant\xi+\varepsilon\leqslant\operatorname{tr}(\hat\rho'_{S'})$,
  \begin{align}
    {S}_{\mathrm{H}}^{\xi}(\hat\rho_S\,\Vert\,e^{-\beta\hat{H}_S}) + \beta(w + 2\eta)
    + \ln\left(\frac{\xi+\varepsilon}{\xi}\right)
    \geqslant{S}_{\mathrm{H}}^{\xi+\varepsilon}(\hat\rho'_{S'}\,\Vert\,e^{-\beta\hat{H}'_{S'}})\ .
  \end{align}
\end{proposition}

Finally, we define asymptotic transformations.  These are transformations in the
thermodynamic limit for which we are interested in the work cost rate, and which
use only a sublinear amount of coherence.

\begin{definition}[Asymptotic thermodynamic process]
\label{def:asymptotic_transformation}
  Consider two sequences of states $\widehat{P}= \{ \hat\rho_n \}$ and
  $\widehat{P}' = \{ \hat\rho'_n \}$ and two sequences of Hamiltonians
  $\widehat{\mathcal{H}}= \{ \hat{H}_n \}$, $\widehat{\mathcal{H}}' = \{ \hat{H}'_n \}$.   
  Let $*$ stand for TO or GPM.
  We say that $\widehat{P}$ can be asymptotically transformed into $\widehat{P}'$
  by an asymptotic $*$ process at a work rate $w$, which we denote by
  $(\widehat{P},\widehat{\mathcal{H}}) \xrightarrow[*]{w} (\widehat{P}',\widehat{\mathcal{H}}')$, if there exists sequences
  $w_n, \eta_n, \varepsilon_n$ such that
  $\hat\rho_n \xrightarrow[*]{w_n,\,\eta_n,\,\varepsilon_n} \hat\rho_n'$ for all $n$ and
  such that
  \begin{align}
    \lim_{n\to\infty} \frac{w_n}{n} &= w\ ;
    & &&
      \lim_{n\to\infty} \frac{\eta_n}{n} &= 0\ ; &&\text{and}
    & &&
      \lim_{n\to\infty} \varepsilon_n &= 0\ .
  \end{align}
\end{definition}

The spectral rates are monotones under asymptotic transformations:
\begin{proposition}[Monotonicity of spectral rates~\cite{Bowen2006ISIT_beyondiid}]
\label{prop:monotonicity_spectral}
  Consider two sequences of states $\widehat{P}= \{ \hat\rho_n \}$ and
  $\widehat{P}' = \{ \hat\rho'_n \}$ and two sequences of Hamiltonians
  $\widehat{\mathcal{H}}= \{ \hat{H}_n \}$, $\widehat{\mathcal{H}}' = \{ \hat{H}'_n \}$.   Define the sequences of Gibbs weight operators
  $\widehat{\Sigma}= \{ e^{-\beta \hat{H}_n } \}$ and
  $\widehat{\Sigma}' = \{ e^{-\beta\hat{H}'_n} \}$.  Let $w\in\mathbb{R}$ be such that
  $\widehat{P}\xrightarrow[*]{w} \widehat{P}'$ where $*$ may stand for either TO or GPM.  Then
  \begin{align}
    {\underline{S}}(\widehat{P}\,\Vert\,\widehat{\Sigma}) + \beta w &\geqslant{\underline{S}}(\widehat{P}'\,\Vert\,\widehat{\Sigma}')\ ;
    &&\text{and}
    &
      {\overline{S}}(\widehat{P}\,\Vert\,\widehat{\Sigma}) + \beta w &\geqslant{\overline{S}}(\widehat{P}'\,\Vert\,\widehat{\Sigma}')\ .
  \end{align}
\end{proposition}

\begin{proof}
  This follows by applying
  \cref{prop:quasi-monotonicity-dhyp-assisted-thermal-transformation} and taking
  the asymptotic limit using the expressions~\eqref{hypothesis_spectral_eq} of
  the asymptotic divergences.
\end{proof}

The monotonicity of the spectral rates implies that if a transformation is
reversible at a given work cost rate, then that rate is necessarily optimal:
\begin{proposition}
\label{prop:optimal}
  Consider two sequences of states $\widehat{P}= \{ \hat\rho_n \}$ and
  $\widehat{P}' = \{ \hat\rho'_n \}$ and two sequences of Gibbs weight operators
  $\widehat{\Sigma}= \{ e^{-\beta \hat{H}_n } \}$ and
  $\widehat{\Sigma}' = \{ e^{-\beta\hat{H}'_n} \}$.   Then if
  $w\in\mathbb{R}$ is such that $\widehat{P}\xrightarrow[*]{w}\widehat{P}'$ and
  $\widehat{P}'\xrightarrow[*]{-w}\widehat{P}$, then for all $w'<w$,
  $\widehat{P}\cancel{\xrightarrow[*]{w'\,}} \widehat{P}'$.
\end{proposition}

This is an expression of  the second law of thermodynamics, or Kelvin's principle, which states that one cannot extract a positive amount of work from a single heat bath by a cyclic protocol.

\subsection{State convertibility by thermal operations}

We now describe our main theorem for state convertibility by thermal operations.  We first derive a sufficient condition for state conversion which is applicable to non-asymptotic cases.  We then take the asymptotic limit and obtain a necessary and sufficient condition for asymptotic state conversion.  The proofs of these theorems will be provided in the next subsection because of their technical nature.

First, we provide a new sufficient criterion for when a general non-semiclassical state
can be approximately reversibly converted to the thermal state using
thermal operations.  Because thermal operations cannot create superpositions of energy eigenstates, arbitrary state transformations generally require a source of coherence.
Here, we show that for any state whose
min- and max-divergences are close, only a small source of coherence is
needed to carry out a transformation to  Gibbs state.

\begin{theorem}
  \label{prop:approx-equipartition-implies-reversibility-under-TO}
  Let $\hat\rho$ be any quantum state on a system with Hamiltonian $\hat{H}$, and
  denote by $\Delta({\hat{H}})$ the spectral range of $\hat{H}$, i.e., the
  difference between the maximum and minimum eigenvalue of $\hat{H}$.  Let
  $\hat\gamma''=1$ be the trivial thermal state on a trivial system with Hilbert
  space $\mathbb{C}$ with Hamiltonian $\hat{H}''=0$.  Let
  $0\leqslant\varepsilon<1/100$.
  Suppose that there exists $S\in\mathbb{R}$ and $\Delta>0$ such that
  \begin{align}
    {S}_{\infty}^{\varepsilon}(\hat\rho\,\Vert\,e^{-\beta \hat{H}}) \leqslant S + \Delta\ ;
    \qquad\text{and}\qquad
    {S}_{0}^{\varepsilon}(\hat\rho\,\Vert\,e^{-\beta \hat{H}}) \geqslant S - \Delta\ .
  \end{align}
  Let $\delta>0$, $q \geqslant2$, and $m=\lceil\Delta({\hat{H}})/\delta\rceil$.    Then we
  have
  \begin{gather}
    \hat\rho\ \xrightarrow[\mathrm{TO}]{
      w = \beta^{-1}(-S+\Delta) + \delta + \beta^{-1}\ln(2m^2(36/\varepsilon)^{3})\,,\ \ 
      \eta = 3q^2\delta\,,\ \ 
      \bar{\varepsilon}= 11\sqrt{\varepsilon}+2/q
    } \ \hat\gamma''\ ,
  \end{gather}
  and
  \begin{gather}
  \hat\gamma''\  \xrightarrow[\mathrm{TO}]{ {\substack{
      w' = \beta^{-1}(S+\Delta)+\delta+\beta^{-1}\ln(2qm^3) +
      16q(\Delta+\beta\delta+\ln(2m))^2/(\beta^2\delta)\ ,\\
      \eta' = 32q^3(\Delta+\beta\delta+\ln(2m))^2/(\beta^2\delta)\ ,\\
      \bar{\varepsilon}' = 10\sqrt{\varepsilon}+7/(2q) + m^2{e}^{-(\Delta+\beta\delta+\ln(m))}
    }}
    } \ \hat\rho\ .
  \end{gather}
\end{theorem}

\cref{prop:approx-equipartition-implies-reversibility-under-TO} allows us to prove the emergence of a thermodynamic
potential in the macroscopic regime.  That is, there is a single quantity that
characterizes exactly when a  transformation by an asymptotic thermal operation is possible.

\begin{theorem}
  \label{thm:asympt-equipartition-implies-asympt-TO-reversibility}
  For sequences of states $\widehat{P}= \{ \hat\rho_n \}$,
  $\widehat{P}' = \{ \hat\rho'_n \}$ and sequences of Hamiltonians
  $\widehat{\mathcal{H}}= \{ \hat{H}_n \}$, $\widehat{\mathcal{H}}' = \{ \hat{H}'_n \}$.  Suppose that the spectral rates collapse for these states
  into a single monotone, i.e.:
  \begin{align}
    {\underline{S}}(\widehat{P}\,\Vert\,\widehat{\Sigma}) = {\overline{S}}(\widehat{P}\,\Vert\,\widehat{\Sigma})
    &=: {S}(\widehat{P}\,\Vert\,\widehat{\Sigma})\ ;
    &
    {\underline{S}}(\widehat{P}'\,\Vert\,\widehat{\Sigma}') = {\overline{S}}(\widehat{P}'\,\Vert\,\widehat{\Sigma}')
    &=: {S}(\widehat{P}'\,\Vert\,\widehat{\Sigma}')\ ,
  \end{align}
  with the sequences $\widehat{\Sigma}= \{ e^{-\beta \hat{H}_n } \}$ and
  $\widehat{\Sigma}' = \{ e^{-\beta\hat{H}'_n} \}$.
  Then
  \begin{align}
     \widehat{P}\xrightarrow[\mathrm{TO}]{\beta^{-1}[{S}(\widehat{P}'\,\Vert\,\widehat{\Sigma}') -
    {S}(\widehat{P}\,\Vert\,\widehat{\Sigma})]} \widehat{P}'. 
  \end{align}
  Equivalently,
  $\widehat{P}\xrightarrow[\mathrm{TO}]{} \widehat{P}'$ if and only if
  ${S}(\widehat{P}\,\Vert\,\widehat{\Sigma}) \geqslant{S}(\widehat{P}'\,\Vert\,\widehat{\Sigma}')$.
\end{theorem}

Crucially, these theorems are applicable even if the state is fully quantum.  On
the other hand, if the state is semiclassical, i.e., if it is block-diagonal in
the energy basis, then the condition for state convertibility in
\cref{prop:approx-equipartition-implies-reversibility-under-TO} reduces to the
known conditions of Refs.~\cite{Aberg2013_worklike,Horodecki2013_ThermoMaj} in
terms of state preparation and work distillation as characterized, e.g., by
thermo-majorization.  In such cases, no source of coherence is required.

Indeed, for semiclassical states, the min-divergence quantifies the amount of
work that can be extracted from a state when transforming it to the thermal
state and the max-divergence quantifies the amount of work that is required to
prepare the state out of the thermal state.  If these divergences collapse, the
state is reversibly convertible to and from the thermal states.  For quantum
states that are not semiclassical, the proof cannot proceed in the same way:
Preparing a general state $\hat\rho$ starting from the thermal state requires
an external source of coherence, and thus the work requirement of state
preparation cannot be given by the max-divergence in same way as for
semiclassical states.  For the proof of
\cref{prop:approx-equipartition-implies-reversibility-under-TO} we need the fact
that the min and the max divergences collapse approximately in order to conclude
that the state can be approximately reversibly transformed to and from the
thermal state.

\Cref{thm:asympt-equipartition-implies-asympt-TO-reversibility}
generalizes and unifies several known situations.  For i.i.d.\@ states and
Gibbs-preserving maps, our theorem reproduces the results of
Ref.~\cite{Matsumoto2010arXiv_reverse}.  In the case of a trivial Hamiltonian, we recover the
results of Ref.~\cite{Jiao2018JMP_convertibility}.  
Our theorem also provides a concrete
application of the general results provided in
Refs.~\cite{Weilenmann2016PRL_axiomatic,PhDMirjam2017,Weilenmann2018_smooth}, in the context of the
axiomatic thermodynamic framework of Lieb and Yngvason~\cite{Lieb1999_secondlaw,Lieb2013_entropy_noneq}.

We note that reversibility only applies to the leading order of the work cost
rate and coherence rate.  Consider two sequences of states $\widehat{P},\widehat{P}'$
that satisfy
${\underline{S}}(\widehat{P}\,\Vert\,\widehat{\Sigma})={\overline{S}}(\widehat{P}\,\Vert\,\widehat{\Sigma}) =
{\underline{S}}(\widehat{P}'\,\Vert\,\widehat{\Sigma}) = {\overline{S}}(\widehat{P}'\,\Vert\,\widehat{\Sigma})$, which are
asymptotically reversibly interconvertible thanks
to~\cref{thm:asympt-equipartition-implies-asympt-TO-reversibility}.  It is still
in general necessary to invest a sublinear amount of work and coherence in the
transformation $\widehat{P}\to\widehat{P}'$ which cannot be recovered in general in the
reverse transformation $\widehat{P}'\to\widehat{P}$.  In our definition of an asymptotic
transformation (\cref{def:asymptotic_transformation}) we deliberately allow
sublinear work and coherence costs for this reason, noting that these quantities
are negligible with respect to the overall work cost of the transformation.

\subsection{Proof of \cref{prop:approx-equipartition-implies-reversibility-under-TO,thm:asympt-equipartition-implies-asympt-TO-reversibility}}

\sloppy

Here we provide the proof of
\cref{prop:approx-equipartition-implies-reversibility-under-TO} and its
asymptotic counterpart,
\cref{thm:asympt-equipartition-implies-asympt-TO-reversibility}.  We proceed in
sequential steps through several lemmas:
\cref{prop:approx-equipartition-implies-reversibility-under-TO} is proved
through \cref{subsub:1} to \cref{subsub:4}, and
\cref{thm:asympt-equipartition-implies-asympt-TO-reversibility} is proved in
\cref{subsub:5}.

In order to simplify the notation and ease readability, we omit the hat symbols
on operators in this subsection.

\subsubsection{Discretizing the Hamiltonian}
\label{subsub:1}

The first simplification that we do is to change the Hamiltonian from $H$ to
a slightly different Hamiltonian $H'$ where the eigenvalues are
``coarse-grained'' into blocks.  That is, given $\delta>0$, we subdivide the
spectrum of $H$ into $m=\lceil\Delta({H})/\delta\rceil$ bins of width
$\delta$, where $\Delta({H})$ is the spectral range of $H$, and we then clamp
all eigenvalues in the bin to a single value which is a multiple of $\delta$.
This yields a Hamiltonian $H'$ with $[H,H']=0$ and
$\lVert {H- H'}\rVert _\infty\leqslant\delta$.  Furthermore, $H'$ only has $m$
distinct eigenvalues, which we denote by $\{ E_k \}$; let also $\{ P_k \}$
denote the projectors onto the corresponding eigenspaces.  We may thus write
\begin{align}
  H' = \sum_{k=0}^{m-1} E_k\, P_k\ ,
  \label{discrete_Hamiltonian}
\end{align}
with $E_k = (k+k_0)\delta$ for some fixed $k_0\in\mathbb{Z}$.

Physically, the transformation $H\to H'$ can be done by turning on a
perturbation of magnitude at most $\delta$.  Furthermore, the perturbation
commutes with the original Hamiltonian.

We note that ${e}^{-\beta H} \leqslant {e}^{-\beta H'+\beta\delta}$ and
${e}^{-\beta H'} \leqslant {e}^{-\beta H+ \beta\delta}$, where the operator
inequalities hold because both sides commute with each other.  This implies
that, for any $\rho$ and for any $\varepsilon>0$, we have
\begin{align}
  {S}_{1/2}^{\varepsilon}(\rho\,\Vert\,{e}^{-\beta H'})
  \geqslant
  {S}_{1/2}^{\varepsilon}(\rho\,\Vert\,{e}^{-\beta H}) - \beta\delta\ ;
  \\
  {S}_{\infty}^{\varepsilon}(\rho\,\Vert\,{e}^{-\beta H'})
  \leqslant
  {S}_{\infty}^{\varepsilon}(\rho\,\Vert\,{e}^{-\beta H}) + \beta\delta\ ;
\end{align}
We also define the dephasing operation for any Hermitian operator $X$ as a
pinching in the energy blocks:
\begin{align}
  \mathcal{D}_{H'}(X) = \sum_{k} P_k X P_k\ .
\end{align}

The following proposition asserts that this perturbation $H_S\to H'_S$ can
be carried out with a $(0,(q^2+1)\delta)$-work/coherence-assisted thermal operation, for
any value of $q>0$ which impacts the accuracy of the process as $1/q$.

\begin{proposition}
  \label{prop:perturbing-Hamiltonian-level-by-level}
  Consider a system $S$ with Hamiltonian $H_S$ and a copy $S'\simeq S$ with a
  Hamiltonian $H'_{S'}$.
  Suppose that $[H'_{S'},{\mathrm{id}}_{S\to S'}(H_S)]=0$ and let $\delta\geqslant 0$
  such that
  $\bigl\lVert {{\mathrm{id}}_{S\to S'}(H_S) - H'_{S'}}\bigr\rVert _{\infty}\leqslant \delta$. 
  Then for any $q>0$ there exists a $(0, (q^2+1)\delta)$-work/coherence-assisted
  transformation $\Phi_{S\to S'}$ such that for any state $\rho_{SR}$ (with any
  reference system $R$), we have
  \begin{align}
    D(\rho_{SR}, \Phi_{S\to S'}(\rho_{SR})) \leqslant \frac{1}{q}\ .
    \label{ineq_1q}
  \end{align}
\end{proposition}
\begin{proof}
  Let $\lvert {k}\rangle _S$ be a simultaneous eigenbasis of ${\mathrm{id}}_{S'\to S}(H'_{S'})$
  and of $H_S$, and write $\lvert {k}\rangle _{S'} = {\mathrm{id}}_{S\to S'}(\lvert {k}\rangle _S)$.  Then
  $H_S\lvert {k}\rangle _S = E_{k}\lvert {k}\rangle _S$ and $H'_{S'}\lvert {k}\rangle _{S'} = E'_{k}\lvert {k}\rangle _{S'}$
  for corresponding eigenvalues $E_k$ and $E'_k$ including multiplicities, i.e.,
  the $E_k$ (resp.\@ $E'_k$) need not be all different.  The condition
  $\lVert {{\mathrm{id}}_{S\to S'}(H_S) - H'_{S'}}\rVert _\infty \leqslant \delta$ implies that
  $\lvert {E_k - E'_k}\rvert  \leqslant \delta$.

  Let $L:= q^2  \delta$. 
  Let $C$, $C'$ be a particle on the intervals $[0,L]$, $[-\delta, L+\delta]$ in $\mathbb R$, respectively,
  which are described by the Hilbert spaces $L^2 ([0,L])$, $L^2 ([-\delta,L+\delta])$.
  There are natural embeddings $L^2 ([0,L]) \subset L^2 ([-\delta,L+\delta]) \subset L^2(\mathbb R)$.

  Let $\chi_I (x)$ be the indicator function for a closed interval $I \subset \mathbb R$.
  We define the Hamiltonians of $C$ and $C'$ by $H_C := x \chi_{[0,L]} (x)$ and $H_{C'} := x \chi_{[-\delta,L+\delta]} (x)$, 
  which are regarded as self-adjoint operators acting on $L^2 ([0,L])$ and $L^2 ([-\delta,L+\delta])$, respectively.
  Obviously, $\| H_C \|_\infty = L$, $\| H_{C'}\|_\infty = L+\delta$.

  We also define the initial state of $C$ by $\zeta (x) := \chi_{[0,L]} (x) / \sqrt{L} \in L^2 ([0,L])$. We can also regard $\zeta (x)$ as an element of $L^2 ([-\delta,L+\delta])$, for which we use the same notation.

  For $a \in \mathbb R$ with $|a| \leqslant\delta$, we define the translation operator $V(a) : L^2 ([0,L]) \to L^2 ([-\delta,L+\delta])$ by $V(a) \varphi (x) := \varphi (x-a)$.
  This is an isometry, where its adjoint $V(a)^\dagger$ is defined on $ L^2 ([-\delta,L+\delta])$ by $V(a)^\dagger \psi (x) = \chi_{[0,L]} (x) \psi (x+a)$ for $\psi (x) \in L^2 ([-\delta,L+\delta])$, because
  \begin{align}
      \int_{-\delta}^{L+\delta} \psi^\ast (x) \varphi (x-a) dx = \int_0^L   \psi^\ast (x+a) \varphi(x) dx.
  \end{align}
Now we define the isometry
  \begin{align}
    V_{SC\to S'C'} := \sum_{k} \lvert {k}\rangle _{S'}\langle {k}\rvert _S\otimes V(E_k - E_k').
  \end{align}
  We can show that  $V_{SC\to S'C'}(H_S+H_C) = (H'_{S'}+H_{C'}) V_{SC\to S'C'}$ by acting with $V_{SC\to S'C'}$ on $\lvert {k}\rangle _S \otimes \varphi (x)$ for any $\varphi(x) \in L^2([0,L])$. 
  Then, we define the CP and trace-nonincreasing map
  \begin{align}
    \Phi_{S \to S'}(\cdot) := \operatorname{tr}_{C'}\bigl[ \lvert {\zeta}\rangle\hspace*{-0.25ex}\langle{\zeta}\rvert \, V_{SC\to S'C'}\,((\cdot)\otimes \lvert {\zeta}\rangle\hspace*{-0.25ex}\langle{\zeta}\rvert )\,V_{SC\leftarrow S'C'}^\dagger \bigr] \ .
  \end{align}
  By construction, $\Phi_{S \to S'}$ is a $(0,(q^2+1)\delta)$-work/coherence-assisted
  thermal operation.
  
  Let $\rho_{SR}$ be any state with any reference system.  Without loss of
  generality, assume that $\rho_{SR}$ is in fact a pure state (or consider a
  larger reference system $R$; the statement will still hold because trace
  distance can only decrease under partial trace). 
  We remark that the fidelity and the trace distance can be defined for infinite-dimensional and  Hilbert spaces, and satisfy the same fundamental properties as in finite dimensions~\cite{Belavkin2004,Furrer2011CMP_infinite}.
  Then, with $\rho_{S'R} = {\mathrm{id}}_{S\to S'}(\rho_{SR})$,
  \begin{align}
    \hspace*{3em}
    &\hspace*{-3em}
    F^2(V_{SC \to S'C'}(\lvert {\rho}\rangle _{SR}\otimes\lvert {\zeta}\rangle ), \lvert {\rho}\rangle _{S'R}\otimes\lvert {\zeta}\rangle )
      \nonumber\\
    &\geqslant
    \operatorname{Re}\bigl\{ (\langle {\rho}\rvert _{S'R}\otimes\langle {\zeta}\rvert )\,
    V_{SC \to S'C'}
    \,(\lvert {\rho}\rangle _{SR}\otimes\lvert {\zeta}\rangle )
    \bigr\}
      \nonumber\\
     &= \sum_{k}  \operatorname{Re}\Bigl\{
       \bigl[ \langle {\rho}\rvert _{S'R} \, \lvert {k}\rangle _{S'}\langle {k}\rvert _S\, \lvert {\rho}\rangle _{SR} \bigr]
       \langle {\zeta}\hspace*{0.2ex}\vert\hspace*{0.2ex}{V(E_k-E'_k)}\hspace*{0.2ex}\vert\hspace*{0.2ex}{\zeta }\rangle \Bigr\}
      \nonumber\\
     &= \sum_{k}  \langle {k}\rvert _S \rho_S \lvert {k}\rangle _S\,
       \langle {\zeta}\hspace*{0.2ex}\vert\hspace*{0.2ex}{V(E_k-E'_k)}\hspace*{0.2ex}\vert\hspace*{0.2ex}{\zeta}\rangle \ ,
       \label{eq:lgiujpbjkojvjogufyov}
  \end{align}
  where the term on $C$ is real because $\zeta (x)$ is real.  We can
  calculate for $| a|  \leqslant\delta$
  \begin{align}
    \langle {\zeta}\hspace*{0.2ex}\vert\hspace*{0.2ex}{V(a)}\hspace*{0.2ex}\vert\hspace*{0.2ex}{\zeta
    }\rangle =  \int_{\mathbb R} dx \, \zeta (x) \, \zeta (x-a)  \geqslant1- \frac{\delta}{L}.
  \end{align}
  Hence, since $\lvert {E_k - E'_k}\rvert \leqslant \delta$, 
  \begin{align}
    \text{\eqref{eq:lgiujpbjkojvjogufyov}}
    \geqslant  \left( 1-\frac{\delta}{L} \right) \sum_{k}  \langle {k}\hspace*{0.2ex}\vert\hspace*{0.2ex}{\rho}\hspace*{0.2ex}\vert\hspace*{0.2ex}{k}\rangle 
    \geqslant 1-\frac{ \delta}{L}\ .
  \end{align}
  Recalling that $D(\rho,\rho')\leqslant\sqrt{1 - F^2(\rho,\rho')}$, and that
  the fidelity can only increase under partial trace, we have
  \begin{align}
    D(\Phi_{S}(\rho_{SR}), \rho_{SR})
    \leqslant \sqrt{\frac{\delta}{L}}\ .
    \tag*\qed\end{align}
\end{proof}

\subsubsection{Manipulating coherence in the state}
\label{subsub:2}

For any state $\rho$ on any system with any Hamiltonian $H$, we can
decompose $\rho$ into modes of coherence~\cite{Lostaglio2015PRX_coherence} as
\begin{align}
  \label{eq:managing-coherence-proof-coh-modes-rhoS}
  \rho = \sum_{\omega} \rho^{(\omega)}\ ,
\end{align}
where $\rho^{(\omega)}$ are general operators satisfying
\begin{align}
  {e}^{-iH t}\rho^{(\omega)}{e}^{iH t} = {e}^{-i\omega t}\rho^{(\omega)}\ ,
\end{align}
for all $t$.  The $\rho^{(\omega)}$ are simply the off-diagonal elements of
$\rho$ that connect two energy levels that differ by $\omega$.  
For the Hamiltonian $H'$ constructed in \eqref{discrete_Hamiltonian}, with only energies that are
multiples of $\delta$, we have that the $\omega$
in~\eqref{eq:managing-coherence-proof-coh-modes-rhoS} range over all possible
differences of energies in $H'$, i.e., over all multiples of $\delta$.

The following lemma states that if the large coherence modes in the state are
suppressed, then it is possible to carry out the dephasing operation by mixing
only a few differently time-evolved versions of $\rho$.

\begin{lemma}
  \label{lemma:entropy-of-dephasing}
  Let $\rho$ be any state on any system with a Hamiltonian $H'$ whose
  energies are multiples of $\delta$ as in \eqref{discrete_Hamiltonian}.  Let $\rho^{(\omega)}$ denote the
  coherence modes in the decomposition of $\rho$ as above.  Let $K'>0$.  Suppose
  that there exists $\xi>0$ such that for all $k$ with $\lvert {k}\rvert \geqslant K'$ we
  have
  \begin{align}
    \label{eq:entropy-of-dephasing-large-coh-modes-small}
    \bigl\lVert {\rho^{(k\delta)}}\bigr\rVert _1 \leqslant \xi\ .
  \end{align}
  Define
  \begin{align}
    \bar\rho =
    \frac1{K'}\sum_{n=0}^{K'-1} {e}^{-\frac{2\pi i n}{K'\delta} H'}
    \rho {e}^{\frac{2\pi i n}{K'\delta} H'}\ .
  \end{align}
  Then, if $m$ denotes the number of distinct eigenvalues of $H'$, we
  have that
  \begin{align}
    D\mathopen{}\left( \bar\rho, \mathcal{D}_{H'}(\rho) \right)
    \leqslant \frac12 m\xi\ .
  \end{align}
\end{lemma}

\begin{proof}
  For any $t>0$, we write
  \begin{align}
    \rho(t) = {e}^{-iH't} \, \rho \, {e}^{iH't}\ ,
  \end{align}
  such that
  \begin{align}
    \bar\rho =
    \frac1{K'}\sum_{n=0}^{K'-1} \rho\mathopen{}\left(\frac{2\pi n}{K'\delta}\right)\ .
  \end{align}
  Recall that $\omega$ in the modes decomposition of $\rho$ is a multiple of
  $\delta$ and ranges over all off-diagonals of $\rho$; i.e.,
  $\omega=k\delta$ for $k=-m+1, \ldots, m-1$.  Furthermore, we may split
  the sum over the modes as a sum over modes in $k=-K'+1,...,K'-1$ and a separate
  sum over the higher order modes.  We can thus calculate:
  \begin{align}
    \bar{\rho}
    &=
      \frac1K \sum_{\omega} \sum_{n=0}^{K'-1} {e}^{-i\omega\frac{2\pi n}{K'\delta}}
      \rho^{(\omega)}
      \nonumber\\
    &= \frac1K\sum_{k=-K'+1}^{K'-1} \left(\sum_{n=0}^{K'-1} {e}^{-2\pi i \frac{n k}{K'}}\right)
      \rho^{(k\delta)}
      + \sum_{\lvert {k}\rvert \geqslant K'} \frac1K\sum_{n=0}^{K'-1} {e}^{-2\pi i \frac{n k}{K'}}
      \rho^{(k\delta)}
      \nonumber\\
    &= \frac1K\sum_{k=-K'+1}^{K'-1}  \delta_{k,0} \,\rho^{(k\delta)}
      + \sum_{\lvert {k}\rvert \geqslant K'} \frac1K\sum_{n=0}^{K'-1} {e}^{-2\pi i \frac{n k}{K'}}
      \rho^{(k\delta)}
      \nonumber\\
    &= \mathcal{D}_{H'}(\rho) + G\ ,
  \end{align}
  where we recall that $\mathcal{D}_{H'}(\rho) = \rho^{(\omega=0)}$
  and where we have defined $G$ as the second sum in the before-to-last
  line.  We can bound the norm of $G$ as follows:
  \begin{align}
    \bigl\lVert {G}\bigr\rVert _1
    \leqslant \sum_{\lvert {k}\rvert \geqslant K'}
    \frac1K \sum_{n=0}^{K'-1} \bigl\lVert { \rho^{(k\delta)} }\bigr\rVert _1
    \leqslant m \xi\ ,
  \end{align}
  where $m$ is a crude upper bound for the total number of terms in the first
  sum, and where each term $\lVert {\rho^{(k\delta)}}\rVert _1$ is individually bounded
  thanks to the
  assumption~\eqref{eq:entropy-of-dephasing-large-coh-modes-small}.  We may
  conclude that $\bar{\rho}$ and $\mathcal{D}_{H'}(\rho)$ are close in trace
  distance:
  \begin{align}
    D\bigl(\bar{\rho}, \mathcal{D}_{H'}(\rho) \bigr)
    = \frac12\bigl\lVert {\bar{\rho} - \mathcal{D}_{H'}(\rho)}\bigr\rVert _1
    \leqslant \frac12 m \xi\ .
    \tag*\qed\end{align}
\end{proof}

Importantly, the min- and max-divergences are only known to quantify the
extractable work and the work cost of formation for semiclassical states, i.e.,
those that commute with the Hamiltonian.  For states that are not semiclassical,
we need a more general statement.  Here, we show a lemma that shows that the
min- and max-divergences also accurately quantify the extractable work
and the work cost of formation for general quantum states, as long as their
large coherence modes are suppressed.

\begin{lemma}
  \label{prop:managing-coherence}
  Let $\rho$ be any quantum state on a system with a Hamiltonian $H'$ whose
  energies are multiples of $\delta$ as in \eqref{discrete_Hamiltonian},
  and let $c\geqslant\beta\delta$.  
  Let $\hat\gamma''=\lvert {0}\rangle\hspace*{-0.25ex}\langle{0}\rvert $ be the thermal state of a trivial system with Hamiltonian $H''=0$ as in \cref{prop:approx-equipartition-implies-reversibility-under-TO}.
  Suppose that there exists $\xi'>0$ such that
  for any $k,k'$ with $\beta\lvert {E_k-E_{k'}}\rvert  \geqslant c$ we have
  \begin{align}
    \label{eq:managing-coherence-large-coh-modes-small}
    \bigl\lVert {P_{k} \,\rho\, P_{k'}}\bigr\rVert _1 \leqslant \xi'\ .
  \end{align}
  Then, for any $\varepsilon'\geqslant m^2\xi'$, we have
  \begin{align}
    \rho \xrightarrow[\mathrm{TO}]{
    \beta^{-1}[-{S}_{1/2}(\rho\,\Vert\,{e}^{-\beta H'}) + \ln(m(6/\varepsilon')^{6})],\ 
    0,\ 
    \varepsilon'
    } \gamma''\ .
    \label{eq:prop-managing-coherence--work-distillation}
  \end{align}
  Conversely, for any integer $q>0$, we have
  \begin{align}
    \gamma'' \xrightarrow[\mathrm{TO}]{
    {\beta^{-1}{S}_{\infty}(\rho\,\Vert\,{e}^{-\beta H'})+
    4qc^2/(\beta^2\delta)+\beta^{-1}\ln(qm^2)},\ 
    {4q(q^2+2)c^2/(\beta^2\delta)},\ 
    {3/(2q)+m\xi'/2}
    } \rho\ .
    \label{eq:prop-managing-coherence--state-formation}
  \end{align}
\end{lemma}
\begin{proof}
  First, note that~\eqref{eq:managing-coherence-large-coh-modes-small} asserts
  that the coherence modes $\rho^{(\omega)}$ of $\rho$ are small for large
  $\omega$.  More precisely: Let $K=\lceil c/(\beta\delta)\rceil$, such that
  $(\lvert {k-k'}\rvert  \geqslant K) \Rightarrow (\beta\lvert {E_k-E_{k'}}\rvert  \geqslant c)$.
  Then for all $\omega=k\delta$ such that $\lvert {k}\rvert \geqslant K$, we have
  \begin{align}
    \bigl\lVert {\rho^{(\omega)}}\bigr\rVert _1 \leqslant m\,\xi'\ ,
    \label{eq:uygiojklbhjk}
  \end{align}
  because the coherence modes are simply the combination of all the blocks in
  the $k$-th off-diagonal of $\rho$, whose individual norm is bounded by our
  assumption~\eqref{eq:managing-coherence-large-coh-modes-small}.  We may invoke
  \cref{lemma:entropy-of-dephasing} to deduce that
  \begin{align}
    D\bigl(\bar\rho, \mathcal{D}_{H'}(\rho)\bigr)
    \leqslant \frac12 m^2\xi'\ ,
  \end{align}
  where $\bar\rho$ is defined in \cref{lemma:entropy-of-dephasing} with
  $K'=K$ and $\xi=m\xi'$.

  \emph{Work extraction from $\rho$.}  Now we construct a strategy to
  transform $\rho$ into the trivial thermal state $\gamma''$.  First, we decohere the
  state in the energy blocks, effecting the transformation
  $\rho\to\mathcal{D}_{H'}(\rho)$ at no work nor coherence cost (this can be
  done by averaging over time, which is a thermal operation).  Then we apply the
  incoherent work extraction protocol
  (\cref{thm:work-distillation-state-formation-TO-semiclassical} in \cref{appx:thermodynamic-lemmas}) to transform
  $\mathcal{D}_{H'}(\rho)\to\gamma''$ with an error parameter
  $\varepsilon'\geqslant m^2\xi'$, while extracting an amount of work equal to
  ${S}_{0}^{\varepsilon'}(\mathcal{D}_{H'}(\rho)\,\Vert\,{e}^{-\beta H'})$, and at no coherence
  cost.  Hence, we have
  $\rho \xrightarrow[\mathrm{TO}]{
    -\beta^{-1}{S}_{0}^{\varepsilon'}(\mathcal{D}_{H'}(\rho)\,\Vert\,{e}^{-\beta H'}),\; 0,\;
    \varepsilon'} \gamma''$.  Using \cref{prop:equiv-Dzero-Dhalf}, observe that
  \begin{align}
    {S}_{0}^{\varepsilon'}(\mathcal{D}_{H'}(\rho)\,\Vert\,{e}^{-\beta H'})
    &\geqslant 
    {S}_{1/2}^{\varepsilon'/2}(\mathcal{D}_{H'}(\rho)\,\Vert\,{e}^{-\beta H'})
    - 6 \ln\bigl(3(\varepsilon'/2)^{-1}\bigr)
      \nonumber\\
    &\geqslant {S}_{1/2}(\bar{\rho}\,\Vert\,{e}^{-\beta H'}) - 6\ln\bigl(6\varepsilon'^{-1}\bigr)\ ,
  \end{align}
  since $\bar\rho$ is a candidate in the optimization that defines the smooth
  min-divergence.  Then we invoke the property of the fidelity that
  $F(A+B,C)\leqslant F(A,C)+F(B,C)$
  (cf.~\cite[Lemma~4.9]{Audenaert2014JMP_discrimination}), to see that
  \begin{align}
    {S}_{1/2}(\bar{\rho}\,\Vert\,{e}^{-\beta H'})
    &= -2\ln F\bigl(\bar\rho, {e}^{-\beta H'}\bigr)
      \nonumber\\
    &\geqslant -2\ln \sum_{n=0}^{K-1}
      F\mathopen{}\left(\frac1K\rho\Bigl(\frac{2\pi n}{K\delta}\Bigr), {e}^{-\beta H'}\right)
      \nonumber\\
    &= -2\ln \sum_{n=0}^{K-1} \frac1{\sqrt K}
      F\mathopen{}\left(\rho, {e}^{-\beta H'}\right)
      \nonumber\\
    &= -\ln(K) + {S}_{1/2}(\rho\,\Vert\,{e}^{-\beta H'})\ .
  \end{align}
  With the crude bound $K\leqslant m$ we finally see that
  \begin{align}
    {S}_{0}^{\varepsilon'}(\mathcal{D}_{H'}(\rho)\,\Vert\,{e}^{-\beta H'})
    \geqslant 
    {S}_{1/2}(\rho\,\Vert\,{e}^{-\beta H'})
    - \ln\bigl(m(6/\varepsilon')^{6}\bigr)\ ,
  \end{align}
  which shows~\eqref{eq:prop-managing-coherence--work-distillation}.

  \emph{Formation of the state $\rho$.}
  We now devise a procedure to construct the state $\rho$ starting
  from the trivial thermal state $\gamma''$.  
  In the following, we refer to the system as $S$, and write $\rho$ and $H'$ as $\rho_S$ and $H'_S$.

  The full protocol consists in three steps.  The
  strategy will be to prepare a completely incoherent state
  $\mathcal{D}_{H'_S+H_C}(\rho_S\otimes\eta_C)$ on the system $S$ along with an
  ancilla system $C$ in such a way that the system $C$ serves as a
  \emph{reference frame} that can be used to induce coherence in $S$.  Then, in
  the second and third steps, we ``externalize'' the reference frame by using
  $C$ to ``induce'' the necessary coherence modes in $S$~\cite{Bartlett2007_refframes}.

  Let $q>0$ be an integer.
  Let $C$ be an ancilla system of dimension $d_C = qK^2$
  and with a Hamiltonian consisting of evenly $\delta$-spaced levels, i.e.,
  $H_C = \sum_{\ell=0}^{d_C-1} \ell\delta\,\lvert {\ell}\rangle\hspace*{-0.25ex}\langle{\ell}\rvert _C$.  Define the state $\eta_C = \lvert {\eta}\rangle\hspace*{-0.25ex}\langle{\eta}\rvert _C$ by
  \begin{align}
    \lvert {\eta}\rangle _C = \frac1{\sqrt{d_C}} \sum_{\ell=0}^{d_C-1} \lvert {\ell}\rangle _C \ .
  \end{align}
  By $\mathcal{D}_{H'_S+H_C}$ we will denote the joint dephasing operation on
  $S$ and $C$, i.e., the dephasing in the common global energy eigenspaces of
  $H'_S+H_C$.

  In the first step of the protocol, starting from the trivial thermal state on
  $S\otimes C$, we prepare the state
  $\mathcal{D}_{H'_S+H_C}(\rho_S\otimes\eta_C)$ at a cost given by the
  max-divergence
  \begin{align}
    {S}_{\infty}(\mathcal{D}_{H'_S+H_C}(\rho_S\otimes\eta_C)\,\Vert\,{e}^{-\beta (H'_S+H_C)})\ .
    \label{eq:miuojkljkb}
  \end{align}
  We can bound this as follows.  The max-divergence can only decrease
  under the dephasing operation; we have
  ${e}^{-\beta(H'_S+H_C)}={e}^{-\beta H'_S}\otimes{e}^{-\beta H_C} \geqslant
  {e}^{-\beta d_C\delta}\,{e}^{-\beta H'_S}\otimes I_C$ because
  $H_C\leqslant d_C\delta I_C$ with $I_C$ being the identity operator of $C$; finally, the max-divergence is additive
  for tensor product states.  This gives us
  \begin{align}
    \text{\eqref{eq:miuojkljkb}}
    &\leqslant {S}_{\infty}(\rho_S\otimes\eta_C\,\Vert\,{e}^{-\beta (H'_S+H_C)})
      \nonumber\\
    &\leqslant {S}_{\infty}(\rho_S\otimes\eta_C\,\Vert\,{e}^{-\beta H'_S}\otimes I_C)
      + \beta d_C\delta
      \nonumber\\
    &= {S}_{\infty}(\rho_S\,\Vert\,{e}^{-\beta H'_S}) + \beta d_C\delta\ ,
  \end{align}
  noting that ${S}_{\infty}(\eta_C\,\Vert\,I_C)=0$ because $\eta_C$ is a pure state.
  Therefore: 
  \begin{align}
    \gamma'' \xrightarrow[\mathrm{TO}]{ \beta^{-1}{S}_{\infty}(\rho_S\,\Vert\,{e}^{-\beta H'_S}) + d_C\delta,\ 
    0,\ 0 } \mathcal{D}_{H'_S+H_C}(\rho_S\otimes\eta_C)\ .
    \tag{Formation protocol, Step~I}
  \end{align}
  
  The next steps are to ``consume'' $C$ in order to induce $\rho_S$ on the
  system $S$ (we need to externalize the reference frame).  This is done as
  follows.

  In preparation for the further steps, we first note that if we post-select the
  reference frame in being in the state $\lvert {\eta}\rangle _C$, then we induce the
  correct state on $S$, approximately.  This is shown as follows:
  \begin{align}
    \langle {\eta}\rvert _C \, \mathcal{D}_{H'_S+H_C}(\rho_S\otimes\eta_C) \, \lvert {\eta}\rangle _C
    &= \sum_{\omega,\omega'}\operatorname{tr}_C\Bigl\{
      \rho_S^{(\omega)}\otimes \bigl(\eta_C^{(-\omega)}\eta_C^{(\omega')} \bigr) \Bigr\}
      \nonumber\\
    &= \sum_{\omega}\operatorname{tr}\bigl(\eta_C^{(-\omega)}\eta_C^{(\omega)} \bigr) \,
      \rho_S^{(\omega)}
      \nonumber\\
    &= \sum_{k} \frac1{d_C}\Bigl(1 - \frac{\lvert {k}\rvert }{d_C}\Bigr) \rho_S^{(k\delta)}
      \nonumber\\
    &= \frac1{d_C}\left( \rho_S - \sum_{k}\frac{\lvert {k}\rvert }{d_C}\rho_S^{(k\delta)} \right)\ ,
      \label{eq:pfojbhas}
  \end{align}
  where we used the fact that $\operatorname{tr}(A^{(\omega)} B^{(\omega')}) = 0$ unless
  $\omega=-\omega'$, and that
  $\operatorname{tr}\bigl(\eta_C^{(-k\delta)}\eta_C^{(k\delta)}\bigr)=(d_C-\lvert {k}\rvert )/d_C^2$
  since $\eta_C^{(k\delta)}$ is the matrix of all zeros except for the $k$-th
  off-diagonal in which all entries are equal to $1/d_C$.  Then
  \begin{align}
    \frac12 \Bigl\lVert { \rho_S -
    d_C \langle {\eta}\rvert _C \, \mathcal{D}_{H'_S+H_C}(\rho_S\otimes\eta_C) \, \lvert {\eta}\rangle _C
    }\Bigr\rVert _1 
    &= 
      \frac12 \biggl\lVert { \sum_k \frac{\lvert {k}\rvert }{d_C} \rho_S^{(k\delta)} }\biggr\rVert _1
      \nonumber\\
    &\leqslant
      \frac12\biggl\lVert { \sum_{\lvert {k}\rvert <K} \frac{\lvert {k}\rvert }{d_C} \rho_S^{(k\delta)} }\biggr\rVert _1
      + 
      \frac12\sum_{\lvert {k}\rvert \geqslant K} \frac{\lvert {k}\rvert }{d_C} \bigl\lVert {  \rho_S^{(k\delta)} }\bigr\rVert _1
      \nonumber\\ 
    &\leqslant \frac12\biggl\lVert { \sum_{\lvert {k}\rvert <K} \frac{\lvert {k}\rvert }{d_C} \rho_S^{(k\delta)} }\biggr\rVert _1
      + \frac12 m^2\xi'\ ,
  \end{align}
  where in the last line we used~\eqref{eq:uygiojklbhjk}.  Let $M^{(K)}$ be the
  matrix in which the $k$-th off-diagonal is filled with the entries equal to
  $\lvert {k}\rvert $, up to the $(K-1)$-th off-diagonal, and the remaining matrix elements
  are zero.  Then we note that
  \begin{align}
    \sum_{\lvert {k}\rvert <K} \lvert {k}\rvert \, \rho_S^{(k\delta)} =  M^{(K)} * \rho_S\ ,
  \end{align}
  where $A*B$ denotes the Hadamard (entry-wise) product.  We note that
  $\lVert {A*B}\rVert _1\leqslant \lVert {A}\rVert _\infty \lVert {B}\rVert _1$, and that
  $\lVert {M^{(K)}}\rVert _\infty \leqslant K^2$ (Suppl.\@ Lemmas~3 and~4
  of~\cite{Aberg2014PRL_catalytic}, originally
  from~\cite{BookHornJohnsonMatrixAnalysis1985}).  Hence,
  $\lVert {M^{(K)} * \rho_S}\rVert _1 \leqslant K^2$ and we finally have
  \begin{align}
    \frac12 \Bigl\lVert { \rho_S -
    d_C \langle {\eta}\rvert _C \, \mathcal{D}_{H'_S+H_C}(\rho_S\otimes\eta_C) \, \lvert {\eta}\rangle _C
    }\Bigr\rVert _1 
    \leqslant \frac{K^2}{2d_C} + \frac12 m\xi'
    \leqslant \frac1{2q} + \frac12 m^2 \xi'\ .
    \label{eq:hgjkljhgfghjkkl}
  \end{align}

  We also note that $\lvert {\eta}\rangle _C$ passes through orthogonal states for each
  time steps $2\pi/(d_C\delta)$.  Actually, for $n=0,\ldots,d_C-1$, the set
  $\bigl\{ \lvert {n}\rangle _C \bigr\}_n$ forms an orthonormal basis of $C$, where
  $\lvert {n}\rangle _C = {e}^{-i\frac{2\pi n}{d_C\delta}H_C}\lvert {\eta}\rangle _C$.  Indeed,
  \begin{align}
    \langle {\eta}\hspace*{0.2ex}\vert\hspace*{0.2ex}{{e}^{i\frac{2\pi n}{d_C\delta}H_C}
    {e}^{-i\frac{2\pi n'}{d_C\delta}H_C}}\hspace*{0.2ex}\vert\hspace*{0.2ex}{\eta}\rangle _C
    &= \frac1{d_C} \sum_{\ell,\ell'=0}^{d_C-1}
      \langle {\ell'}\hspace*{0.2ex}\vert\hspace*{0.2ex}{{e}^{i\frac{2\pi (n-n')}{d_C\delta}H_C}}\hspace*{0.2ex}\vert\hspace*{0.2ex}{\ell}\rangle _C
      \nonumber\\
    &= \frac1{d_C} \sum_{\ell=0}^{d_C-1} {e}^{i\frac{2\pi (n-n')\ell}{d_C}}
      \quad = \delta_{n,n'}\ .
  \end{align}

  Step~2 of our protocol consists in flattening the Hamiltonian of $C$ so that
  we can perform nontrivial unitaries without worrying about coherences.  From
  the state $\mathcal{D}_{H'_S+H_C}(\rho_S\otimes\eta_C)$ with Hamiltonian
  $H'_S+H_C$, we ``flatten'' the Hamiltonian of the ancilla system $C$
  using~\cite[Lemma~8.1]{Faist2021CMP_impl} and consuming an additional
  ancilla $C'$ of dimension $d_{C}(q^{2}+2)$, with the Hamiltonian $H_{C''}$ being bounded as $\| H_{C'} \| \leqslant d_{C}(q^{2}+2) \delta$ and with the original state surviving
  up to precision $1/q$.  That is, we achieve the following Hamiltonian
  transformation
  \begin{align}
    \left(\mathcal{D}_{H'_S+H_C}(\rho_S\otimes\eta_C) \,;\; H'_S+H_C \right)
    \xrightarrow[\mathrm{TO}]{ 0,\; d_{C}(q^2+2)\delta,\; 1/q }
    \left(\mathcal{D}_{H'_S+H_C}(\rho_S\otimes\eta_C)\,;\;
    H'_S + (\Delta({H_C})/2)I_C \right).
    \tag{Formation protocol, Step~II}
  \end{align}

  Finally, in Step~3 we carry out the following energy-conserving unitary
  controlled on the system $C$:
  \begin{align}
    U_{SC} = \sum_{n=0}^{d_C-1} {e}^{i\frac{2\pi n}{d_C\delta} H'_S} \otimes \lvert {n}\rangle\hspace*{-0.25ex}\langle{n}\rvert _C\ ,
  \end{align}
  and we then use Landauer erasure to reset $C$ to a pure state and to trace it
  out.  Note that
  ${e}^{-iH'_S t} \, \mathcal{D}_{H'_S+H_C}(\rho_S\otimes\eta_C) \, {e}^{iH'_S
    t} = {e}^{-iH'_S t}\, {e}^{i(H'_S+H_C)t}\,
  \mathcal{D}_{H'_S+H_C}(\rho_S\otimes\eta_C)\, {e}^{-i(H'_S+H_C) t}\,
  {e}^{iH'_S t} = {e}^{iH_C t} \, \mathcal{D}_{H'_S+H_C}(\rho_S\otimes\eta_C) \,
  {e}^{-iH_C t}$ because the dephased state is invariant under time evolution.
  Then, the application of the unitary $U_{SC}$ to
  $\mathcal{D}_{H'_S+H_C}(\rho_S\otimes\eta_C)$, and tracing out $C$, yields
  \begin{align}
    \operatorname{tr}_C\bigl[U_{SC}\mathcal{D}_{H'_S+H_C}(\rho_S\otimes\eta_C) U_{SC}^\dagger\bigr]
    &= \sum_{n=0}^{d_C-1}
      \bigl({e}^{i\frac{2\pi n}{d_C\delta} H'_S}\otimes\langle {n}\rvert \bigr)
      \mathcal{D}_{H'_S+H_C}(\rho_S\otimes\eta_C) 
      \bigr({e}^{-i\frac{2\pi n}{d_C\delta} H'_S}\otimes\lvert {n}\rangle \bigr)
      \nonumber\\
    &= \sum_{n=0}^{d_C-1}
      \langle {n}\rvert _C\,{e}^{-i\frac{2\pi n}{d_C\delta}H_C} \,
      \mathcal{D}_{H'_S+H_C}(\rho_S\otimes\eta_C)  \,
      {e}^{i\frac{2\pi n}{d_C\delta}H_C} \lvert {n}\rangle _C
      \nonumber\\
    &= d_C \langle {\eta}\rvert _C\,
      \mathcal{D}_{H'_S+H_C}(\rho_S\otimes\eta_C)  \,
      \lvert {\eta}\rangle _C\ .
  \end{align}
  Recalling \eqref{eq:hgjkljhgfghjkkl}, we know that this state is close to the
  required $\rho_S$.  Noting that we need $\beta^{-1}\ln(d_C)$ work to reset
  $C$ to a pure state, we find:
  \begin{align}
    \left(\mathcal{D}_{H'_S+H_C}(\rho_S\otimes\eta_C)\ ;\ 
    H'_S + (\Delta({H_C})/2)I_C \right)
    \xrightarrow[\mathrm{TO}]{ \beta^{-1}\ln(d_C) ,\, 0 ,\, 1/(2q)+m\xi'/2 }
    \left( \rho_S \,;\; H'_S \right).
    \tag{Formation protocol, Step~III}
  \end{align}
  Note that the final uniform Hamiltonian on the system $C$ can be restored to
  the original Hamiltonian at no work or coherence cost, by keeping the state of
  $C$ at a pure state of constant energy and changing the other levels to match
  those of the original Hamiltonian $H_C$.

  Combining together these three steps, we see that
  \begin{align}
    \gamma'' \xrightarrow[\mathrm{TO}]{ {\beta^{-1}{S}_{\infty}(\rho_S\,\Vert\,{e}^{-\beta H'_S})+
    qK^2\delta+\beta^{-1}\ln(qK^2)},\ {qK^2(q^{2}+2)\delta},\
    {3/(2q)+m^2\xi'/2}} \rho_S\ .
  \end{align}
  Recalling $K=\lceil c/(\beta\delta)\rceil\leqslant 2c/(\beta\delta)$ while
  assuming $c\geqslant\beta\delta$, we
  obtain~\eqref{eq:prop-managing-coherence--state-formation}.
\end{proof}

\subsubsection{Collapse of the min and max divergences suppresses coherence}
\label{subsub:3}

Here we show that the difference between (alternative) min-divergence and the
max-divergence is a quantity that provides a characterization of how much
coherence there is in the state.  Namely, if the divergences do not differ by
more than $2\Delta'$, then the one-norm of off-diagonal energy blocks
$P_k\hat \rho P_{k'}$ is exponentially suppressed in $\lvert {E_k-E_{k'}}\rvert $ as long
as $\lvert {E_k-E_{k'}}\rvert  \gtrsim \Delta'$.

\begin{lemma}
  \label{lemma:large-coherences-suppressed-if-approximate-equipartition}
  Let $\rho$ be a quantum state. Suppose there are $S\in\mathbb{R}$ and
  $\Delta'>0$ such that
  \begin{align}
    {S}_{\infty}(\rho\,\Vert\,{e}^{-\beta H'}) \leqslant S + \Delta'\ ;
    \qquad\text{and}\qquad
    {S}_{1/2}(\rho\,\Vert\,{e}^{-\beta H'}) \geqslant S - \Delta'\ .
  \end{align}
  Then for any $k,k'$, we have
  \begin{align}
    \lVert {P_k\rho P_{k'}}\rVert _1 \leqslant {e}^{-\beta\lvert {E_k - E_{k'}}\rvert /2 + \Delta'}\ .
  \end{align}
\end{lemma}
\begin{proof}
  Using H\"older's inequality, we have
  \begin{align}
    \bigl\lVert {P_k\rho P_{k'}}\bigr\rVert _1 \leqslant
    \bigl\lVert {P_k\rho^{1/2}}\bigr\rVert _1 \bigl\lVert {P_{k'}\rho^{1/2}}\bigr\rVert _\infty\ .
  \end{align}
  By definition of the R\'enyi-1/2 divergence, we have for any $k$,
  \begin{align}
    {S}_{1/2}(\rho\,\Vert\,{e}^{-\beta H'})
    &= -2\ln \operatorname{tr}\sqrt{\rho^{1/2}{e}^{-\beta H'}\rho^{1/2}}
      \nonumber\\
    &\leqslant -2\ln\bigl[{e}^{-\beta E_k/2}\operatorname{tr}\sqrt{\rho^{1/2}P_k\rho^{1/2}}\bigr]
      \nonumber\\
    &= \beta E_k - 2\ln\,\bigl\lVert {P_k\rho^{1/2}}\bigr\rVert _1\ ,
  \end{align}
  and hence
  \begin{align}
    \bigl\lVert {P_k\rho^{1/2}}\bigr\rVert _1^2
    \leqslant \exp\bigl\{-{S}_{1/2}(\rho\,\Vert\,{e}^{-\beta H'}) + \beta E_k\bigr\}
    \leqslant \exp\bigl\{-S + \Delta' + \beta E_k\bigr\}\ .
  \end{align}
  On the other hand, we have
  \begin{align}
    \bigl\lVert {P_{k'}\rho^{1/2}}\bigr\rVert _\infty^2
    = \bigl\lVert {P_{k'}\rho P_{k'}}\bigr\rVert _\infty
    \leqslant
    {e}^{{S}_{\infty}(\rho\,\Vert\,{e}^{-\beta H'})} \bigl\lVert {P_{k'}{e}^{-\beta H'}P_{k'}}\bigr\rVert _\infty
    \leqslant \exp\bigl\{S+\Delta' - \beta E_{k'}\bigr\}\ ,
  \end{align}
  recalling that the square of the largest singular value of a matrix $A$ is the
  maximum eigenvalue of $A A^\dagger$.  Putting these together, and noting that
  the same argument holds if we exchange $k$ and $k'$, we obtain
  \begin{align}
    \bigl\lVert {P_k\rho P_{k'}}\bigr\rVert _1 \leqslant {e}^{-\beta\lvert {E_k - E_{k'}}\rvert /2 + \Delta'}\ ,
  \end{align}
  as claimed.
\end{proof}

\subsubsection{Proof of \cref{prop:approx-equipartition-implies-reversibility-under-TO}}
\label{subsub:4}

Finally, we can prove
\cref{prop:approx-equipartition-implies-reversibility-under-TO}.  If the smooth
min and max R\'enyi divergences coincide approximately, we use the above lemmas
to conclude that there exist protocols for work distillation and state formation
with approximately matching work costs.  The difficult part of the proof is to
show that there is a single state that is a good enough smoothing candidate
simultaneously in both~\eqref{eq:def-smooth-Dmax}
and~\eqref{eq:def-smooth-Dminz}.

\begin{proof}
  First, we need to connect the assumption on the smoothed entropy measures to a
  specific state which has a small gap between its non-smoothed min and
  max-divergences.
  Our specific goal below is to construct a state $\tilde \rho$ that satisfies the conditions of \cref{lemma:large-coherences-suppressed-if-approximate-equipartition} and is sufficiently close to $\rho$.

  Because $H'\leqslant H+\delta$ and $H\leqslant H'+\delta$, we have
  \begin{align}
    {S}_{1/2}^{\varepsilon}(\rho\,\Vert\,{e}^{-\beta H'})
    &\geqslant {S}_{1/2}^{\varepsilon}(\rho\,\Vert\,{e}^{-\beta H}) - \beta\delta
      \geqslant 
      {S}_{0}^{\varepsilon}(\rho\,\Vert\,{e}^{-\beta H}) - \beta\delta \geqslant S - \Delta - \beta\delta\ ;
      \nonumber\\
    {S}_{\infty}^{\varepsilon}(\rho\,\Vert\,{e}^{-\beta H'})
    &\leqslant {S}_{\infty}^{\varepsilon}(\rho\,\Vert\,{e}^{-\beta H}) + \beta\delta
      \leqslant S + \Delta + \beta\delta\ .
  \end{align}

  Both protocols, work extraction and state formation, start by shifting the
  Hamiltonian $H\to H'$, and at the end shifting the Hamiltonian back $H'\to H$.
  Thanks to \cref{prop:perturbing-Hamiltonian-level-by-level}, this can be done
  at a cost in the total coherence parameter of $(q^2+1)\delta$ and at a precision
  cost $1/q$ in each way.

  Let $\rho'$ be the optimal subnormalized quantum state for
  ${S}_{1/2}^{\varepsilon}(\rho\,\Vert\,{e}^{-\beta H'}) = {S}_{1/2}(\rho'\,\Vert\,{e}^{-\beta H'})$,
  satisfying $D(\rho,\rho')\leqslant \varepsilon$ and
  $\operatorname{tr}(\rho')\geqslant 1-\varepsilon$.
  
  Let $\gamma' = {e}^{-\beta H'}/\operatorname{tr}({e}^{-\beta H'})$ and write
  \begin{align}
    {S}_{\infty}^{2\varepsilon}(\rho'\,\Vert\,\gamma')
   \ \  = \ \ 
      \begin{array}[t]{rc}
        \min & \ln(\alpha) \\
        \text{s.t.}: & \rho''\leqslant \alpha\gamma' \\
             &D(\rho'',\rho')\leqslant2\varepsilon\ .
      \end{array}
 \ \  \geqslant  \ \  
      \begin{array}[t]{rc}
        \min & \ln(\alpha) \\
        \text{s.t.}: & \rho'\leqslant \alpha\gamma' + F \\
             &\operatorname{tr}(F)\leqslant2\varepsilon\ \text{and}\ F\geqslant 0\ .
      \end{array}
  \end{align}
  Let $\alpha$, $F$ denote optimal choices in the last optimization.  Let
  \begin{align}
    G =
    \gamma'^{1/2} \bigl(\gamma' + \alpha^{-1}\mathcal{D}_{H'}(F)\bigr)^{-1/2}\ ,
  \end{align}
  where $\mathcal{D}_{H'}(\cdot)$ denotes the dephasing operation in the
  eigenspaces of $H'$.  Then, using the pinching inequality, and because $G$
  commutes with time evolution,
  \begin{multline}
    G\rho' G^\dagger \leqslant G\bigl( \alpha \gamma' + F \bigr) G^\dagger
    \leqslant
    m\,\mathcal{D}_{H'}\bigl[ G\bigl( \alpha \gamma' + F \bigr) G^\dagger \bigr]
    =
    m\,G\,\mathcal{D}_{H'}\!\bigl[ \alpha \gamma' + F \bigr]\,G^\dagger
    \\
    =
    m\,G\,\bigl(\alpha \gamma' + \mathcal{D}_{H'}[F]\bigr)\,G^\dagger
    = m\alpha\,\gamma'\ ,
  \end{multline}
  and thus
  ${S}_{\infty}(G\rho'G^\dagger\,\Vert\,\gamma') \leqslant\ln(m)+\ln(\alpha)
  \leqslant\ln(m)+{S}_{\infty}^{2\varepsilon}(\rho'\,\Vert\,\gamma')$.  Shifting back the
  normalization of the second argument gives
  \begin{multline}
    {S}_{\infty}(G\rho'G^\dagger\,\Vert\,{e}^{-\beta H'})
    \leqslant \ln(m)+{S}_{\infty}^{2\varepsilon}(\rho'\,\Vert\,{e}^{-\beta H'})
    \leqslant \ln(m)+{S}_{\infty}^{\varepsilon}(\rho\,\Vert\,{e}^{-\beta H'})
    \\
    \leqslant S + \Delta + \beta\delta + \ln(m)\ ,
  \end{multline}
  because the optimal state in the last max-divergence is a candidate in
  the optimization for ${S}_{\infty}^{2\varepsilon}(\rho'\,\Vert\,{e}^{-\beta H'})$.  Also, taking
  the trace of the constraint $\rho'\leqslant\alpha\gamma'+F$ we obtain
  $\alpha\geqslant 1-4\varepsilon$, and then using~\cite[Lemma~A.4]{Dupuis2013_DH},
  we have
  $P(G\rho' G^\dagger/\operatorname{tr}(\rho'),\rho'/\operatorname{tr}(\rho')) \leqslant
  \sqrt{2\operatorname{tr}(\alpha^{-1}\mathcal{D}[F])/\operatorname{tr}(\rho')} \leqslant
  2\sqrt{\varepsilon/[(1-4\varepsilon)(1-\varepsilon)]} \leqslant 4\sqrt{\varepsilon}$
  (using $\varepsilon\leqslant 1/8$),
  where $P(\sigma,\sigma'):= \sqrt{1-F(\sigma,\sigma')^2} \geqslant D(\sigma,\sigma')$ is the purified distance for $\sigma, \sigma' \in {\mathcal{S}}(\mathscr{H})$.
  Hence,
  $D(G\rho' G^\dagger/\operatorname{tr}(\rho'),\rho'/\operatorname{tr}(\rho'))\leqslant 4\sqrt{\varepsilon}$ and
  thus
  $D(G\rho' G^\dagger,\rho')\leqslant 4\operatorname{tr}(\rho')\sqrt{\varepsilon}\leqslant
  4\sqrt{\varepsilon}$.

  On the other hand, we have
  \begin{align}
    F\bigl(G\rho'G^\dagger, \gamma'\bigr) =
    \operatorname{tr}\sqrt{\gamma'^{1/2} G\rho'G^\dagger \gamma'^{1/2}}
    =
    \bigl\lVert {\rho'^{1/2} G^\dagger \gamma'^{1/2}}\bigr\rVert _1
    \leqslant
    \bigl\lVert {\rho'^{1/2} \gamma'^{1/2}}\bigr\rVert _1
    \bigl\lVert {\gamma'^{-1/2} G^\dagger \gamma'^{1/2}}\bigr\rVert _\infty\ ,
  \end{align}
  using H\"older's inequality.  Conveniently, $[G,\gamma']=0$ by construction,
  and thus also $[G,\gamma'^{1/2}]=0$ and $[G,\gamma'^{-1/2}]=0$, and
  $\bigl\lVert {\gamma'^{-1/2} G^\dagger \gamma'^{1/2}}\bigr\rVert _\infty =
  \bigl\lVert {G}\bigr\rVert _\infty\leqslant 1$, since $G$ is a contraction (because
  $G^\dagger G\leqslant I$).  Hence
  \begin{align}
    {S}_{1/2}(G\rho'G^\dagger\,\Vert\,\gamma')
    = -\ln F^2\bigl(G\rho'G^\dagger, \gamma'\bigr)
    \geqslant -\ln F^2\bigl( \rho', \gamma' \bigr)
    = {S}_{1/2}(\rho'\,\Vert\,\gamma')\ ,
  \end{align}
  and thus
  \begin{align}
    {S}_{1/2}(G\rho'G^\dagger\,\Vert\,{e}^{-\beta H'})
    \geqslant
    {S}_{1/2}(\rho'\,\Vert\,{e}^{-\beta H'})
    = {S}_{1/2}^{\varepsilon}(\rho\,\Vert\,{e}^{-\beta H'})
    \geqslant S - \Delta - \beta\delta\ .
  \end{align}

  Finally, we define
  \begin{align}
    \tilde\rho = \frac{G\rho' G^\dagger}{\operatorname{tr}\bigl(G\rho' G^\dagger\bigr)}\ .
  \end{align}
  We have
  $\operatorname{tr}(G\rho'G^\dagger)\geqslant 1 - \varepsilon- 4\sqrt{\varepsilon} \geqslant 1 -
  5\sqrt{\varepsilon}$, and thus
  \begin{align}
    D\bigl(\tilde\rho, G\rho' G^\dagger\bigr)
    = 1 - \operatorname{tr}\bigl(G\rho' G^\dagger\bigr)
    \leqslant 1 - (1 - \varepsilon- 4\sqrt{\varepsilon}) = \varepsilon+4\sqrt{\varepsilon}\ ,
  \end{align}
  and by a chain of triangle inequalities
  \begin{align}
    D\bigl( \tilde\rho, \rho \bigr)
    \leqslant
    D\bigl(\tilde\rho, G\rho'G^\dagger\bigr)+
    D\bigl(G\rho'G^\dagger, \rho'\bigr)+
    D\bigl(\rho', \rho\bigr)
    \leqslant 2\varepsilon+8\sqrt{\varepsilon}\leqslant 10\sqrt{\varepsilon}\ .
  \end{align}

  We can define
  $\Delta' = \Delta + \beta\delta + \ln(m) - \ln(1-5\sqrt{\varepsilon})$, while noting
  that $-\ln(1-5\sqrt{\varepsilon})\leqslant\ln(2)$ as $\varepsilon<1/100$.  Then, the
  state $\tilde\rho$ satisfies
  \begin{align}
    {S}_{\infty}(\tilde\rho\,\Vert\,{e}^{-\beta H'})
    & \leqslant S + \Delta + \beta\delta + \ln(m) -\ln\operatorname{tr}(G\rho' G^\dagger)
      \leqslant S + \Delta'\ ;
    \\
      {S}_{1/2}(\tilde\rho\,\Vert\,{e}^{-\beta H'})
    &\geqslant S - \Delta - \beta\delta + \ln\operatorname{tr}(G\rho' G^\dagger)
      \geqslant S - \Delta'\ .
  \end{align}
  We then have $\Delta'\leqslant \Delta + \beta\delta + \ln(2m)$ and
  $\Delta'\geqslant \Delta+\beta\delta+\ln(m)$.

  Then, the conditions of
  \cref{lemma:large-coherences-suppressed-if-approximate-equipartition} are
  fulfilled, and for any $k,k'$, we have that
  \begin{align}
    \bigl\lVert {P_k \tilde \rho P_{k'}}\bigr\rVert _1
    \leqslant \exp\bigl( -\lvert {k - k'}\rvert \beta\delta+\Delta'\bigr)\ .
    \label{eq:lkjyfuygiuhoijpok}
  \end{align}

  Now, for any $r>1$ we set $c = r\Delta'$.  For any $k,k'$ with
  $\lvert {k - k'}\rvert \beta\delta\geqslant c$, \Cref{eq:lkjyfuygiuhoijpok} tells us
  that $\bigl\lVert {P_k \tilde \rho P_{k'}}\bigr\rVert _1 \leqslant {e}^{-(r-1)\Delta'} =: \xi'$.
  We set $r=2$ in the following for convenience.

  The conclusions of \cref{prop:managing-coherence} apply to the interconversion
  of $\tilde\rho$ to and from the thermal state.

  \emph{Distilling work from $\rho$.}  Work can be distilled, i.e., the transition $\rho\to\gamma''$ is possible, with
  the parameters (we have set $\varepsilon'=\sqrt{\varepsilon}$ in
  \cref{prop:managing-coherence})
  \begin{align}
    \left\{
    \begin{array}{rll}
      w &=& \beta^{-1} \left[ -S + \Delta \right] + \delta +
            \beta^{-1}\ln(2m^2(36/\varepsilon)^{3}) \\[1ex]
      \eta &=& 2(q^2+1)\delta \\[1ex]
      \varepsilon&=& 11\sqrt{\varepsilon}+ 2/q\ .
    \end{array}\right.
  \end{align}

  \emph{Preparing the state $\rho$.}  The state $\rho$ can be prepared, i.e., the transition
  $\gamma'' \to\rho$ is possible, with the parameters
  \begin{align}
    \left\{
    \begin{array}{rll}
      w &=& \beta^{-1}\left[ S + \Delta \right] + \delta + \beta^{-1}\ln(2qm^3) +
            16  q(\Delta+\beta\delta+\ln(2m))^2/(\beta^2\delta)\\[1ex]
      \eta &=& 16q(q^2+2) (\Delta+\beta\delta+\ln(2m))^2/(\beta^2\delta) \\[1ex]
      \varepsilon&=& 10\sqrt{\varepsilon}+ 7/(2q) + m^2{e}^{-(\Delta+\beta\delta+\ln(m))}\ .
    \end{array}\right.
  \end{align}

Finally, letting $q\geqslant2$, we obtain the slightly simplified parameters in  \cref{prop:approx-equipartition-implies-reversibility-under-TO}.

\end{proof}

\subsubsection{Proof of \cref{thm:asympt-equipartition-implies-asympt-TO-reversibility}}
\label{subsub:5}

We now present the proof of
\cref{thm:asympt-equipartition-implies-asympt-TO-reversibility}, the main
theorem of the first part of our main result.  The proof proceeds by applying
\cref{prop:approx-equipartition-implies-reversibility-under-TO} in the
thermodynamic limit.

\begin{proof}
  We use \cref{prop:approx-equipartition-implies-reversibility-under-TO} to show
  asymptotic convertibility of $\widehat{P}$ (relative to $\widehat{\Sigma}$) to and from
  the Gibbs state $  \gamma''$ on a trivial system at zero energy.  We write
  $\widehat{\Sigma}'' = \{  \gamma'' \}$ the trivial sequence of trivial Gibbs states.
  For $\varepsilon>0$, let
  \begin{align}
    S_{n,\varepsilon} &:= \frac12\left\{ {S}_{\infty}^{\varepsilon}( \rho_n\,\Vert\,{e}^{-\beta   H_n}) +
                     {S}_{0}^{\varepsilon}( \rho_n\,\Vert\,{e}^{-\beta   H_n}) \right\}\ ,
         \\
    \Delta_{n,\varepsilon} &:= \max\bigl\{ {S}_{\infty}^{\varepsilon}( \rho_n\,\Vert\,{e}^{-\beta   H_n}) -
                          {S}_{0}^{\varepsilon}( \rho_n\,\Vert\,{e}^{-\beta   H_n}) , \sqrt{n} \bigr\} \geqslant 0\ ;
  \end{align}
  and let
  $\Delta_{\infty, \varepsilon} := \limsup_{n\to\infty} \Delta_{n,\varepsilon}/n$.
  We have
  \begin{gather}
    \lim_{\varepsilon\to 0}\limsup_{n\to\infty}\frac1n S_{n,\varepsilon}
    =
    \lim_{\varepsilon\to 0}\limsup_{n\to\infty}\frac1{2n}
    \left\{ {S}_{\infty}^{\varepsilon}( \rho_n\,\Vert\,{e}^{-\beta   H_n}) +
          {S}_{0}^{\varepsilon}( \rho_n\,\Vert\,{e}^{-\beta   H_n}) \right\}
    =: \bar S\ ;
    \\
    \lim_{\varepsilon\to0} \limsup_{n\to\infty}
    \frac1n \Delta_{n,\varepsilon}
    \leqslant
    \max\left\{
    \lim_{\varepsilon\to0} \limsup_{n\to\infty}
    \frac1n\bigl[ {S}_{\infty}^{\varepsilon}( \rho_n\,\Vert\,{e}^{-\beta   H_n}) -
               {S}_{0}^{\varepsilon}( \rho_n\,\Vert\,{e}^{-\beta   H_n}) \bigr], 0 \right\}
    = 0\ .
  \end{gather}

  For $\varepsilon>0$ and for each $n$, we apply
  \cref{prop:approx-equipartition-implies-reversibility-under-TO} with the
  choices $S = S_{n,\delta}$, $\Delta=\Delta_{n,\delta}$,
  $\delta=\beta^{-1}\Delta_{n,\varepsilon}$ and
  $q=(\Delta_{\infty, \varepsilon})^{-1/4}$.  Then
  $m=O(\operatorname{poly}(n))/\Delta_{n,\varepsilon}$.  Observe that
  $\Delta_{n,\varepsilon} = O(n)$ and that $\Delta_{n,\varepsilon}$ increases at least
  as fast as $\sqrt{n}$ by definition;
  thus $m=O(\operatorname{poly}(n))$.  Let
  $w_{n,\varepsilon}, \eta_{n,\varepsilon}, \bar{\varepsilon}_{n,\varepsilon}$ be the
  parameters of the work extraction process given by
  \cref{prop:approx-equipartition-implies-reversibility-under-TO} for these
  choices.  Then
  \begin{align*}
    \lim_{\varepsilon\to0}\limsup_{n\to\infty} \frac{w_{n,\varepsilon}}{n}
    &= -\beta^{-1}\bar S\ ;
    &\lim_{\varepsilon\to0}\limsup_{n\to\infty} \frac{\eta_{n,\varepsilon}}{n}
    &= \lim_{\varepsilon\to0} \,3(\Delta_{\infty, \varepsilon})^{1/2} = 0 \ ;
    \\
    \lim_{\varepsilon\to0} \limsup_{n\to\infty}  \bar{\varepsilon}_{n,\varepsilon}
    &= 0 + \lim_{\varepsilon\to0} \,(\Delta_{\infty, \varepsilon})^{1/4} = 0\ ,
  \end{align*}
  and we can apply \cref{lemma:asympt-transform-fixedepsilon-ok} in \cref{appx:lemmas} to conclude
  that $\widehat{P}\xrightarrow[\mathrm{TO}]{-\beta^{-1}\bar S} \widehat{\Sigma}''$.
  
  For the work extraction process, we define $S'_{n,\varepsilon}$, $\bar S'$, $\Delta'_{n,\varepsilon}$, and $\Delta'_{\infty, \varepsilon}$ similarly. Then the parameters
  $w'_{n,\varepsilon},\, \eta'_{n,\varepsilon}, \bar{\varepsilon}'_{n,\varepsilon}$ given by
  \cref{prop:approx-equipartition-implies-reversibility-under-TO} satisfy
  \begin{align*}
    \lim_{\varepsilon\to0}\limsup_{n\to\infty} \frac{w'_{n,\varepsilon}}{n}
    &= \beta^{-1}\bar S'\ ;
    &\lim_{\varepsilon\to0}\limsup_{n\to\infty} \frac{\eta'_{n,\varepsilon}}{n}
    &= \lim_{\varepsilon\to0} \,
      32 \beta^{-2} (\Delta'_{\infty, \varepsilon})^{-3/4}\Delta'_{\infty, \varepsilon}  = 0\ ;
      \\
    \lim_{\varepsilon\to0}\limsup_{n\to\infty} \bar{\varepsilon}'_{n,\varepsilon}
    &= 0\ ,
  \end{align*}
  where we used the fact that sublinear terms are suppressed, that
  $\limsup_{n\to\infty} [\Delta'_{n,\varepsilon}+\beta\delta+\ln(am^b)]/n =
  2\Delta'_{\infty, \varepsilon}$ for any $a,b\geqslant0$, and that
  $\limsup_{n \to \infty} m^2 e^{-(\Delta+\delta+\ln(m))}/n = 0$ because
  $\Delta'_{n,\varepsilon}+\delta+\ln(m)$ grows at least as fast as $\sqrt{n}$ and
  the exponential takes over the polynomial.
  Thus from \cref{lemma:asympt-transform-fixedepsilon-ok} in \cref{appx:lemmas} we see that
  $\widehat{\Sigma}'' \xrightarrow[\mathrm{TO}]{\beta^{-1}\bar S} \widehat{P}$.
  Combining these two processes for different states immediately yields
  $\widehat{P}\xrightarrow[\mathrm{TO}]{\beta^{-1}[{S}(\widehat{P}'\,\Vert\,\widehat{\Sigma}') -
    {S}(\widehat{P}\,\Vert\,\widehat{\Sigma})]} \widehat{P}'$.

  It is clear that if
  ${S}(\widehat{P}\,\Vert\,\widehat{\Sigma}) \geqslant{S}(\widehat{P}'\,\Vert\,\widehat{\Sigma}')$, then
  $\widehat{P}\xrightarrow[\mathrm{TO}]{} \widehat{P}'$ from the above, using
  Property~\ref{propitem:nonoptimal-thermo-transformations}
  of~\cref{prop:elementary-thermo-transformation-properties} in \cref{appx:lemmas}.  Also, if
  $\widehat{P}\xrightarrow[\mathrm{TO}]{} \widehat{P}'$, then monotonicity of the spectral rates imply that
  ${S}(\widehat{P}\,\Vert\,\widehat{\Sigma}) \geqslant{S}(\widehat{P}'\,\Vert\,\widehat{\Sigma}')$.
\end{proof}

\section{Collapse of the min and max divergences for ergodic states relative
  to local Gibbs states}
\label{sec:main-result-for-ergodic-local-Gibbs}

In this section, we prove the second main theorem of our main result (\cref{quantum_theorem_main2t}): For any $\widehat{P}$ that
is translation-invariant ergodic and for any local translation-invariant Gibbs
state $\widehat{\Sigma}$, then we have
${\overline{S}}(\widehat{P}\,\Vert\,\widehat{\Sigma}) = {\underline{S}}(\widehat{P}\,\Vert\,\widehat{\Sigma}) =
{S}_{1}(\widehat{P}\,\Vert\,\widehat{\Sigma})$.  Combined with
\cref{thm:asympt-equipartition-implies-asympt-TO-reversibility}, this implies
that all such states can be reversibly converted into one another with thermal
operations and a negligible amount of coherence.

We prove this assertion in two steps.  First, we formulate a generalized version
of Stein's lemma~\cite{Ogawa2000IEEETIT_Stein,Bjelakovic2004CMP_ergodic,Nagaoka2007IEEETIT_hypothesis,Bjelakovic2003arXiv_revisted,Brandao2010CMP_Stein}.  We derive a sufficient condition for the min and max
divergence converge to the same value that is heavily inspired by these
references.  The condition is the existence of an operator obeying a simple set
of properties, that plays the role of a typical projector.  In a second step, we
prove that for ergodic states and local Gibbs translation-invariant states, this
condition is fulfilled.

\subsection{A sufficient condition for quantum Stein's lemma}
\label{sec:generalization-quantum-Stein-basics}

The quantum Stein's lemma relates to a hypothesis test between two states
$\hat\rho_n$ and $\hat\sigma_n$ using a single measurement. If we employ the
optimal strategy that correctly reports $\hat\rho_n$ with probability at least
$\eta$, then the probability of erroneously reporting $\hat\rho_n$ decreases
exponentially as $\exp(-n {S}_{1}(\hat\rho_n\,\Vert\,\hat\sigma_n))$, with the rate being
given by the KL divergence.  This statement holds in several known
cases, such as for i.i.d.\@ states, or if $\hat\rho_n$ is ergodic and
$\hat\sigma_n$ is i.i.d.~\cite{Bjelakovic2004CMP_ergodic}.

Quantum Stein's lemma can be formulated in terms of the hypothesis testing
divergence.  For sequences $\widehat{P}$, $\widehat{\Sigma}$, a quantum Stein's lemma
would state that for all $0<\eta<1$,
\begin{align}
  \label{eq:DHyp-converges-to-KL}
  {S}_{\mathrm{H}}^{\eta}(\widehat{P}\,\Vert\,\widehat{\Sigma}) = {S}_{1}(\widehat{P}\,\Vert\,\widehat{\Sigma})\ .
\end{align}

Because the hypothesis testing divergence is monotonic in $\eta$, and because it
interpolates between the min and max divergences [cf.\@
\cref{hypothesis_spectral_eq}], we see that the hypothesis testing divergence
converges to the KL divergence as per~\eqref{eq:DHyp-converges-to-KL}, if and
only if the min and max divergences converge to the KL divergence,
\begin{align}
  \label{eq:Dinf-Dsup-converge-to-KL}
  {\underline{S}}(\widehat{P}\,\Vert\,\widehat{\Sigma}) = {\overline{S}}(\widehat{P}\,\Vert\,\widehat{\Sigma}) = {S}_{1}(\widehat{P}\,\Vert\,\widehat{\Sigma})\ .
\end{align}
Therefore, to prove~\eqref{eq:Dinf-Dsup-converge-to-KL} for a class of states it
suffices to prove~\eqref{eq:DHyp-converges-to-KL}.

A simplest situation where the quantum Stein's lemma holds is the i.i.d.\@
setting, i.e., $\widehat{P}:= \{ \hat\rho^{\otimes n} \}$ and
$\widehat{\Sigma}:= \{ \hat\sigma^{\otimes n} \}$. In this situation, for any
$0 < \eta < 1$,
\begin{align}
  {S}_{\mathrm{H}}^{\eta}(\widehat{P}\,\Vert\,\widehat{\Sigma}) =  {S}_{1}(\hat\rho\,\Vert\,\hat\sigma)\ ,
\end{align}
and consequently,
\begin{equation}
  {\overline{S}}(\widehat{P}\,\Vert\,\widehat{\Sigma}) = {\underline{S}}(\widehat{P}\,\Vert\,\widehat{\Sigma}) = {S}_{1}(\hat\rho\,\Vert\,\hat\sigma)\ ,
\end{equation}
as was proved in~\cite[Theorem~2]{Nagaoka2007IEEETIT_hypothesis}.

We now derive a sufficient condition for the
convergence~\eqref{eq:DHyp-converges-to-KL}, providing a generalization of the
quantum Stein's lemma beyond i.i.d.\@ states.
\begin{lemma}
  \label{Philippe_prop}
  Let $\widehat{P}$ and $\widehat{\Sigma}$ be any sequences of states.  Suppose that there
  exists $c \in \mathbb R$ such that for any $\varepsilon> 0$, there exists a
  sequence of operators $\hat W_n^\varepsilon$ that satisfy, for sufficiently
  large $n$,
  \begin{subequations}
    \begin{align}
      \hat W_n^{\varepsilon\dagger}\hat W_n^\varepsilon&\leqslant\hat{I}\ ;
        \label{condition1}
      \\
      \operatorname{tr}\bigl[\hat W_n^\varepsilon\hat \sigma_n \hat W_n^{\varepsilon\dagger}\bigr]
      &\leqslant e^{-n (c - 2\varepsilon)}\ ;
        \label{condition2}
      \\
      \hat W_n^{\varepsilon\dagger} \hat \rho_n \hat W_n^\varepsilon&\leqslant e^{n(c+2\varepsilon)} \hat \sigma_n\ ;
        \label{condition3}
      \\
      \lim_{n \to \infty} \operatorname{Re}\bigl( \operatorname{tr}\bigl[\hat W_n^\varepsilon\hat \rho_n\bigr]  \bigr)
      &= 1\ .
        \label{condition4}
    \end{align}
  \end{subequations}
  Then, for any $0 < \eta < 1$, we have
  \begin{equation}
    {S}_{\mathrm{H}}^{\eta}(\widehat{P}\,\Vert\,\widehat{\Sigma}) = c\ .
    \label{Philippe_prop_eq}
  \end{equation}
\end{lemma}

Our proof is based on tools from semidefinite
programming~\cite{Watrous2009_sdps,Dupuis2013_DH,Faist2018PRX_workcost}, which imply that the
hypothesis testing divergence is equivalently expressed using two different
optimizations:
\begin{align}
  \label{eq:DHyp-SDP-primal-dual}
  {S}_{\mathrm{H}}^{\eta}(\hat\rho\,\Vert\,\hat\sigma)
  = -\ln \,
  \min_{\substack{0\leqslant\hat Q\leqslant\hat{I}\\ \operatorname{tr}[\hat Q \hat\rho]\geqslant\eta}}
  \left\{ \eta^{-1}\operatorname{tr}\bigl[\hat Q\hat\sigma\bigr] \right\}
  = -\ln \,
  \sup_{\substack{\mu\geqslant0,\ \hat X\geqslant0\\ \mu\hat\rho\leqslant\hat\sigma+\hat X}}
  \left\{ \mu - \frac{\operatorname{tr}[\hat  X]}{\eta } \right\}  \ .
\end{align}
The optimizations are called the \emph{primal problem} and \emph{dual problem}
respectively.  We note that our proof below only requires the so-called
\emph{weak duality} between the minimization and the maximization, which states
that the optimal value of the minimization problem is an upper bound to the
optimal value of the maximization problem. 

The reason that we have equality
in~\eqref{eq:DHyp-SDP-primal-dual} is that for the hypothesis testing divergence,
the stronger notion of \emph{strong duality} holds, which states that both
optimization problems have the same optimal value.
We note that the reason we
  write a supremum for the dual problem is that for $\eta=1$, even as strong
  duality holds, we are not guaranteed that the supremum is achieved by a
  specific choice of $\mu$ and $\hat X$.  In the primal problem the minimum is
  always achieved.  This can be seen using Slater's
  conditions~\protect\cite{Watrous2009_sdps}, noting that we can restrict the
  optimization to the support of $\hat\sigma$.

\begin{proof}[Proof of \cref{Philippe_prop}]
  Our proof proceeds by exhibiting explicit candidates in both optimizations
  in~\eqref{eq:DHyp-SDP-primal-dual}, yielding upper and lower bounds that both
  converge to $c$ as $n\to\infty$.

  Let $\hat Q_n^\varepsilon:= \hat W_n^\varepsilon{}^\dagger \hat W_n^\varepsilon$.
  From condition~\eqref{condition4} and \cref{Philippe_lemma7} (a) in
  \cref{appx:lemmas}, we have
  \begin{align}
    \lim_{n \to \infty} \operatorname{tr}\bigl[\hat Q_n^\varepsilon\hat\rho_n\bigr] = 1\ ,
  \end{align}
  which implies that for any $0 < \eta < 1$, we have
  $ \operatorname{tr}\bigl[\hat Q_n^\varepsilon\hat\rho_n\bigr] > \eta$ for sufficiently
  large $n$, and $\hat Q_n$ is a valid optimization candidate
  in~\eqref{eq:DHyp-SDP-primal-dual}.  Using~\eqref{condition2}, the value
  attained by this candidate is
  \begin{align}
    e^{-{S}_{\mathrm{H}}^{\eta}(\hat\rho_n\,\Vert\,\hat\sigma_n)}
    \leqslant\eta^{-1}\, e^{-n (c-2\varepsilon)}\ ,
  \end{align}
  and thus
  \begin{align}
    \frac1n {S}_{\mathrm{H}}^{\eta}(\hat\rho_n\,\Vert\,\hat\sigma_n)
    \geqslant c - 2\varepsilon+ \frac1n\ln(\eta)\ .
  \end{align}
  By taking $n \to \infty$ and then $\varepsilon\to +0$, we conclude that
  ${S}_{\mathrm{H}}^{\eta}(\widehat{P}\,\Vert\,\widehat{\Sigma}) \geqslant c$.

  Now we consider the second optimization in~\eqref{eq:DHyp-SDP-primal-dual}.
  First, we note that using a generalization of the Pinching inequality
  (Lemma~B.1 of Ref.~\cite{Faist2021CMP_impl}),
  \begin{align}
    \hat\rho_n
    \leqslant2 \bigl[ \hat W_n^{\varepsilon\dagger} \hat\rho_n \hat W_n^\varepsilon+ (\hat{I}- \hat W_n^{\varepsilon\dagger}) \hat\rho_n (\hat{I}- \hat W_n^\varepsilon)\bigr]\ .
    \label{p_ineq}
  \end{align}
  Let $\mu := e^{-n (c+2\varepsilon)} / 2 > 0$ and
  $\hat X := 2\mu(\hat{I}- \hat W_n^\varepsilon{}^\dagger ) \hat\rho_n (\hat{I}-
  \hat W_n^\varepsilon) \geqslant0$.  From inequality~\eqref{p_ineq} and
  condition~\eqref{condition3}, we have
  $\mu \hat\rho\leqslant\hat\sigma+ \hat X$, and hence $\mu,\hat X$ are valid
  optimization candidates in the maximization
  in~\eqref{eq:DHyp-SDP-primal-dual}.
  From \cref{Philippe_lemma7} (b) in \cref{appx:lemmas}, we have
  $\operatorname{tr}\bigl[{(\hat{I}- \hat W_n^{\varepsilon\dagger} ) \hat\rho_n (\hat{I}-
    \hat{W}_n^\varepsilon)}\bigr] \to 0$ as $n\to\infty$, and therefore, for
  sufficiently large $n$, we have
  $\operatorname{tr}\bigl[{(\hat{I}- \hat W_n^{\varepsilon\dagger} ) \hat\rho_n (\hat{I}-
    \hat{W}_n^\varepsilon)}\bigr] \geqslant\eta/4$.  Therefore, for sufficiently large
  $n$,
  \begin{align}
    \mu - \frac{\operatorname{tr}[\hat X]}{\eta}
    = \mu \left( 1 - \frac{2\operatorname{tr}\bigl[(\hat{I}- \hat W_n^{\varepsilon\dagger} )
    \hat\rho_n (\hat{I}- \hat W_n^\varepsilon)\bigr]}{\eta} \right)
    \geqslant\frac{\mu}{2}\ .
  \end{align}
  The value attained by the maximization is then
  \begin{align}
    {S}_{\mathrm{H}}^{\eta}(\hat\rho_n\,\Vert\,\hat\sigma_n)
    \leqslant-\ln\left\{ \mu - \frac{\operatorname{tr}[\hat X]}{\eta} \right\}
    \leqslant-\ln\left\{ \frac{1}{4} e^{-n(c+2\varepsilon)} \right\}\ .
  \end{align}
  Dividing by $n$, taking $n \to \infty$ and then $\varepsilon\to +0$, we deduce
  that ${S}_{\mathrm{H}}^{\eta}(\widehat{P}\,\Vert\,\widehat{\Sigma}) \leqslant c$.
\end{proof}

In fact, one can see that the product of two typical projectors constructed in
Ref.~\cite{Bjelakovic2003arXiv_revisted} for the i.i.d.\@ case satisfies the conditions
\eqref{condition1}--\eqref{condition4} above, with
$c={S}_{1}(\hat\rho\,\Vert\,\hat\sigma)$.

\subsection{Formulation of ergodic states and local Gibbs states}
\label{sec:ergodic-states}

In a second step of our main result, we consider ergodic states and local Gibbs
states.  Here we show that for these states, it is possible to construct an
operator that satisfies the conditions in \cref{Philippe_prop}, in turn proving
the collapse of the min and max divergences to the KL divergence.

The standard way to rigorously formulate ergodicity invokes infinite-dimensional $C^\ast$-algebras~\cite{BookBratteliRobinson_OpAlgQStatMech1,BookBratteliRobinson_OpAlgQStatMech2,BookRuelle_StatMechRigorous}.
Here, for the sake of broad readability, we introduce the relevant concepts directly
in an equivalent~--- albeit perhaps less elegant~--- formulation that does not
require the use of $C^\ast$ algebras.  For completeness, we provide the
construction based on $C^\ast$ algebras in \cref{appx:Cstar-algebras-construction}.

We consider a spatially $d$-dimensional system on the lattice $\mathbb Z^d$.  To
each site $i\in\mathbb{Z}^d$, we assign a copy $\mathscr{H}_i$ of a finite-dimensional
Hilbert space, such that the Hilbert spaces for all sites are isomorphic.
We denote the set of operators acting on $\mathscr{H}_i$ by $\mathcal A_i$.  For a
bounded region $\Lambda \subset \mathbb Z^d$, we define
$\mathscr{H}_\Lambda := \bigotimes_{i \in \Lambda} \mathscr{H}_i$ and
$\mathcal A_\Lambda := \bigotimes_{i \in \Lambda} \mathcal A_i$.  We note that
these are finite-dimensional spaces because $\Lambda$ is bounded.

For a bounded region $\Lambda \subset \mathbb Z^d$, we consider a density
operator $\hat\rho_\Lambda$ whose support is $\Lambda$, i.e.,
$\hat\rho_\Lambda \in {\mathcal{S}}(\mathscr{H}_\Lambda)$.  We assume that we are given a collection $\{ \hat\rho_\Lambda \}$ for all bounded subregions of the lattice, which furthermore obey
the \textit{consistency} condition, namely,
\begin{align}
  \hat\rho_\Lambda = \operatorname{tr}_{\Lambda' \setminus \Lambda} [\hat\rho_{\Lambda'}]\ .
  \label{consistency}
\end{align}
This condition is necessary to ensure that all $\hat\rho_\Lambda$ are obtained
from a common global state defined on the entire infinite lattice (see
\cref{appx:Cstar-algebras-construction}).

Consider now a sequence of bounded regions of the lattice defined as follows.
For any $\ell\in \mathbb{N}$, let $[-\ell,\ell] := \{ -\ell, -\ell+1, \cdots, \ell-1, \ell \} \subset \mathbb Z$ and
$\Lambda_\ell := [-\ell,\ell]^d \subset \mathbb Z^d$.  
We define the sequence of quantum states $\widehat{P}= \{ \hat\rho_n \}$ by
$\hat\rho_n := \hat\rho_{\Lambda_\ell}$, where we set
$n := (2\ell+1)^d = | \Lambda_\ell |$.  While $n=(2\ell+1)^d$ with
$\ell=1,2, \cdots$ does not run over all of the elements of $\mathbb N$, it does
not affect our following argument; indeed, it is straightforward to complete the
sequence with intermediate states for all $n\in\mathbb{N}$ such that the limits
that we derive are unaffected.

Before we can formulate ergodicity, we consider the shift superoperator.  The
shift superoperator $T_i$ is defined such that for any local operator $\hat A_j$
whose support is $j \in \mathbb Z^d$, it is mapped by $T_i$ to the same operator
at site $j+i \in \mathbb Z^d$, i.e., $T_i (\hat A_j) = \hat A_{j+i}$, where we
regard $i \in \mathbb Z^d$ as a $d$-dimensional vector with the standard
addition for such vectors.

\begin{definition}[Translation invariance]
  A sequence $\widehat{P}$ of the form above is translation invariant, if it
  satisfies the consistency condition~\eqref{consistency}, and for all
  $n = (2\ell+1)^d$, all $\hat A \in \mathcal A_\Lambda$ with $\Lambda$ being
  bounded, and all $i \in \mathbb Z^d$ satisfying
  $T_i (\hat A) \in \mathcal A_{\Lambda_l}$, we have
  \begin{align}
    \operatorname{tr}\bigl[ \hat\rho_n T_i ( \hat A) \bigr] = \operatorname{tr}\bigl[ \hat\rho_n \hat A \bigr]\ .
  \end{align}
\end{definition}

We note that ``translation invariant'' is often referred to as ``stationary'' in
the context of ergodic theory.  In our setup, we interpret $i \in \mathbb Z^d$
as a coordinate of the spatial potition instead of time, and therefore we prefer
the denomination ``translation invariant.''

Translation invariance is a central ingredient for the definition of ergodicity:

\begin{definition}[Ergodicity]
  \label{defn:ergodic}
  A sequence $\widehat{P}$ is translation-invariant and ergodic, if it is
  translation invariant, and for all self-adjoint
  $\hat A \in \mathcal A_\Lambda$ for a bounded region $\Lambda$ we have
  \begin{align}
    \lim_{m \to \infty} \operatorname{tr}\left[
      \hat\rho_n \left( \frac{1}{(2m+1)^d} \sum_{i \in \Lambda_{m}} T_i (\hat A) \right)^2
    \right]
    = \lim_{n \to \infty} \bigl( \operatorname{tr}\bigl[ \hat\rho_n \hat A \bigr]\bigr)^2\ ,
    \label{ergodic_def}
  \end{align}
  where $n = (2\ell+1)^d$ on the left-hand side is taken such that
  $T_i (\hat A) \in \mathcal A_{\Lambda_\ell}$ for all $i \in \Lambda_{m}$.
\end{definition}

The limit on the right-hand side of (\ref{ergodic_def}) is not actually necessary, because the consistency condition (\ref{consistency}) implies that $ \operatorname{tr}[ \hat\rho_n \hat A ]$ does not depend on $n$ for large $n=(2\ell+1)^d$ satisfying $\Lambda \subset \Lambda_\ell$.
The equivalence of this definition and the standard definition is proved in
\cref{appx:Cstar-algebras-construction}.

This definition implies that the variance of the shift average (i.e., the
spatial average) of any local observable vanishes in the thermodynamic limit.
We emphasize that an ergodic state can be out of equilibrium, because ergodicity
is defined with respect to the spatial shift instead of time evolution.

We now define the Hamiltonian of the system which determines the Gibbs state.
Let $\hat h_i$ be a local operator describing interaction, whose support is a
bounded region around site $i \in \mathbb Z^d$.  More precisely, we assume that
the support of $\hat h_i$ is in
$\{ j : | j_k -i_k | \leqslant r, \ \forall k \} \subset \mathbb Z^d$, where
$0 \leqslant r < \infty$ is an integer and $i_k$, $j_k$ describe the $k$-th
components of $i, j \in \mathbb Z^d$ ($k=1, \cdots, d$).  We note that $r$
represents the interaction length, where $r=0$ describes non-interacting cases.

Then, for a bounded region $\Lambda \subset \mathbb Z^d$, the truncated
Hamiltonian is given by
\begin{align}
  \hat H_{\Lambda} := \sum_{i \in \Lambda} \hat h_i\ .
\end{align}
A Hamiltonian of this form is referred to as a \textit{local} Hamiltonian.  The
Hamiltonian is \textit{translation invariant}, 
if it can be written in the form
\begin{align}
  \hat H_{\Lambda} = \sum_{i \in \Lambda} T_i(\hat h_{0})\ ,
\end{align}
for some fixed operator $h_0$.

Let $\beta > 0$ be the inverse temperature.
The \textit{truncated} Gibbs state on a bounded  region $\Lambda$ is given by the density operator
\begin{equation}
\hat\sigma^{\Box}_{\Lambda} := \exp (\beta (F_{\Lambda} - \hat H_{\Lambda})),
\label{tr_Hamiltonian}
\end{equation}
where $F_{\Lambda} := - \beta^{-1} \ln \operatorname{tr}[ \exp (-\beta \hat H_{\Lambda} )]$ is
the truncated free energy.  We note that $\hat\sigma^{\Box}_\Lambda$ does not
satisfy the consistency condition \eqref{consistency}, because of the effects on
the edges of the region $\Lambda$ where we have truncated the Hamiltonian.

We  consider a sequence of the truncated Gibbs states: We define $\widehat{\Sigma}^{\Box}:= \{ \hat\sigma^{\Box}_n \}$ with $\hat\sigma^{\Box}_n := \hat\sigma^{\Box}_{\Lambda_m}$, where $n:= (2\ell+1)^d$ and $m:= \ell - r$.
We note that, with this definition, the supports of $\hat\sigma^{\Box}_n$ and $\hat\rho_n$ are the same.
In the following we  use the shorthands $\hat H_n := \hat H_{\Lambda_{m}}$ and $F_n := F_{\Lambda_{m}}$.

\subsection{Generalized Stein's lemma for ergodic states relative to
  local Gibbs states}

We now consider a proof of a generalization of the quantum Stein's lemma for ergodic states relative to local Gibbs states.
We begin by proving that the limiting KL divergence is well defined:

\begin{lemma}
  \label{Gibbs_KL_rate0}
  Suppose that $\widehat{P}$ is translation invariant and $\widehat{\Sigma}^{\Box}$ is the
  truncated Gibbs state of a local and translation-invariant Hamiltonian in any
  dimensions.  Then $S_1 (\widehat{P}\| \widehat{\Sigma}^{\Box})$ exists.
\end{lemma}

\begin{proof}
  This follows from the following well-known facts.  From \cref{tr_Hamiltonian},
  \begin{align}
    \frac{1}{\lvert {\Lambda}\rvert } {S}_{1}(\hat\rho_\Lambda\,\Vert\,\hat\sigma^{\Box}_\Lambda)
    = - \frac{1}{\lvert {\Lambda}\rvert } {S}_{1}(\hat\rho_\Lambda)
    - \beta \frac{F_{\Lambda}}{\lvert {\Lambda}\rvert }
    + \beta \frac{\operatorname{tr}\bigl[\hat H_{\Lambda} \hat\rho_\Lambda\bigr]}{\lvert {\Lambda}\rvert }\ .
    \label{KL_energy}
  \end{align}
  The first term on the right-hand side converges to ${S}_{1}(\widehat{P})$ because
  $\widehat{P}$ is translation invariant (Proposition~6.2.38 of
  Ref.~\cite{BookBratteliRobinson_OpAlgQStatMech2}).  It is also known that the second term converges to
  the free energy density (Theorem~6.2.40 of Ref.~\cite{BookBratteliRobinson_OpAlgQStatMech2}).
  The third term also converges, because $\hat H_{\Lambda}$ is local and
  translation invariant, and $\widehat{P}$ is translation invariant.
\end{proof}

One important ingredient in the proof of our generalization of the quantum
Stein's lemma is the following typical projector for ergodic states (Theorem~2.1
of Ref.~\cite{Bjelakovic2004IM_lattice}; see also Theorem~5.1 of
Ref.~\cite{Bjelakovic2005QIP_compression} and Theorem~1.4 of Ref.~\cite{Ogata2013LMP_shannonmcmillan}).

\begin{proposition}[Quantum Shannon-McMillan Theorem]
  \label{ergodic_prop}
  Suppose that $\widehat{P}$ is ergodic.  Then for any $\varepsilon>0$ there exists a
  sequence of projectors $\hat \Pi_{\widehat{P},n}^\varepsilon$ (called typical
  projectors) that satisfy, for sufficiently large $n$,
  \begin{gather}
    e^{-n (s+ \varepsilon)} \hat \Pi_{\widehat{P},n}^\varepsilon\leqslant\hat \Pi_{\widehat{P},n}^\varepsilon\, \hat\rho_n \, \hat \Pi_{\widehat{P},n}^\varepsilon\leqslant e^{-n(s-\varepsilon)} \hat \Pi_{\widehat{P},n}^\varepsilon\ ;
    \label{rho_typical1}
    \\
    e^{n(s-\varepsilon)} \leqslant\operatorname{tr}\bigl[\hat \Pi_{\widehat{P},n}^\varepsilon\bigr] \leqslant e^{n(s+\varepsilon)}\ ,
    \label{rho_typical2}
    \\
    \lim_{n \to \infty} \operatorname{tr}\bigl[\hat \Pi_{\widehat{P},n}^\varepsilon\hat\rho_n\bigr] = 1\ ,
    \label{rho1_1}
  \end{gather}
  where $s:= {S}_{1}(\widehat{P})$.
\end{proposition}

We now consider our main theorem for ergodic states and for the truncated Gibbs
state.

\begin{theorem}[Collapse of the spectral rates for the truncated Gibbs state]
  \label{quantum_theorem_main2t}
  Consider a lattice $\mathbb{Z}^d$ of spatial dimension $d$ and suppose that
  $\widehat{P}$ is translation invariant and ergodic, as
  in~\cref{sec:ergodic-states}.  Let $\widehat{\Sigma}^{\Box}$ be the sequence of
  truncated Gibbs states of a local and translation invariant Hamiltonian on the
  lattice.  Then, for any $0 < \eta < 1$,
  \begin{align}
    {S}_{\mathrm{H}}^{\eta}(\widehat{P}\,\Vert\,\widehat{\Sigma}^{\Box}) =  {S}_{1}(\widehat{P}\,\Vert\,\widehat{\Sigma}^{\Box})\ ,
    \label{second_main_0t}
  \end{align}
  and as a consequence,
  \begin{align}
    {\underline{S}}(\widehat{P}\,\Vert\,\widehat{\Sigma}^{\Box})
    = {\overline{S}}(\widehat{P}\,\Vert\,\widehat{\Sigma}^{\Box})
    =  {S}_{1}(\widehat{P}\,\Vert\,\widehat{\Sigma}^{\Box})\ .
    \label{second_main_t}
  \end{align}
\end{theorem}

\begin{proof}
  From the proof of \cref{Gibbs_KL_rate0}, the following limit exists,
  \begin{align}
    m := \lim_{n \to \infty} \frac{1}{n} \left( - \operatorname{tr}[\hat\rho_n \ln \hat\sigma^{\Box}_n] \right)\ .
  \end{align}
  Let $s:= S_1 (\widehat{P})$.  We define relative typical projectors (as inspired
  by Refs.~\cite{Bjelakovic2004CMP_ergodic,Bjelakovic2003arXiv_revisted}) as
  \begin{align}
    \hat \Pi_{\widehat{P}|\widehat{\Sigma}^{\Box}, n}^\varepsilon:= \operatorname{Proj}\left\{ - \frac{1}{n} \ln \hat\sigma^{\Box}_n \in [m -\varepsilon, m+\varepsilon]\right\}\ ,
  \end{align}
  which satisfy by definition
  \begin{align}
    e^{-n (m+\varepsilon)}\, \hat\Pi_{\widehat{P}|\widehat{\Sigma}^{\Box},n}^\varepsilon\leqslant\hat\Pi_{\widehat{P}|\widehat{\Sigma}^{\Box},n}^\varepsilon\, \hat\sigma^{\Box}_n \,
    \hat\Pi_{\widehat{P}|\widehat{\Sigma}^{\Box},n}^\varepsilon\leqslant e^{-n(m-\varepsilon)} \, \hat\Pi_{\widehat{P}|\widehat{\Sigma}^{\Box},n}^\varepsilon\ .
    \label{sigma_typical}
  \end{align}
  We then define 
  \begin{equation}
    \hat W_n^\varepsilon:= \hat\Pi_{\widehat{P},n}^\varepsilon\,  \hat\Pi_{\widehat{P}|\widehat{\Sigma}^{\Box},n}^\varepsilon\ .
  \end{equation}
  The remainder of the proof is devoted to showing that the operator
  $\hat W_n^\varepsilon$ satisfies the four conditions
  \eqref{condition1}--\eqref{condition4} in \cref{Philippe_prop}
  with
  \begin{align}
    c := -s+m = {S}_{1}(\widehat{P}\,\Vert\,\widehat{\Sigma}^{\Box})\ .
  \end{align}
  These conditions then immediately imply \cref{second_main_0t}, as discussed in
  \cref{sec:generalization-quantum-Stein-basics}.
  
  The condition~\eqref{condition1} is clear by definition.
  Condition~\eqref{condition2} is obtained from
  inequalities~\eqref{rho_typical2} and~\eqref{sigma_typical} as
  \begin{align}
    \operatorname{tr}\bigl[
    \hat\Pi_{\widehat{P},n}^\varepsilon\hat\Pi_{\widehat{P}|\widehat{\Sigma}^{\Box},n}^\varepsilon\, \hat\sigma^{\Box}_n \,
    \hat\Pi_{\widehat{P}|\widehat{\Sigma}^{\Box},n}^\varepsilon\hat\Pi_{\widehat{P},n}^\varepsilon\bigr]
    &\leqslant e^{-n(m-\varepsilon) } \operatorname{tr}\bigl[ \hat\Pi_{\widehat{P},n}^\varepsilon\hat\Pi_{\widehat{P}|\widehat{\Sigma}^{\Box},n}^\varepsilon\hat\Pi_{\widehat{P},n}^\varepsilon\bigr]
      \nonumber\\
    &\leqslant e^{-n(m-\varepsilon) } \operatorname{tr}\bigl[ \hat\Pi_{\widehat{P},n}^\varepsilon\bigr]
    \nonumber\\
    &\leqslant e^{-n(m-s -2\varepsilon) }\ .
  \end{align}
  The third condition~\eqref{condition3} is obtained from
  inequalities~\eqref{rho_typical1} and~\eqref{sigma_typical} as
  \begin{align}
    \hat\Pi_{\widehat{P}|\widehat{\Sigma}^{\Box},n}^\varepsilon\hat\Pi_{\widehat{P},n}^\varepsilon\, \hat\rho_n \,
    \hat\Pi_{\widehat{P},n}^\varepsilon\hat\Pi_{\widehat{P}|\widehat{\Sigma}^{\Box},n}^\varepsilon&\leqslant e^{-n(s-\varepsilon)} \,  \hat\Pi_{\widehat{P}|\widehat{\Sigma}^{\Box},n}^\varepsilon\hat\Pi_{\widehat{P},n}^\varepsilon\hat\Pi_{\widehat{P}|\widehat{\Sigma}^{\Box},n}^\varepsilon\nonumber\\
    &\leqslant e^{-n(s-\varepsilon)}  \, \hat\Pi_{\widehat{P}|\widehat{\Sigma}^{\Box},n}^\varepsilon\nonumber\\
    &\leqslant e^{-n(s-\varepsilon)} e^{+n (m+\varepsilon)} \,
    \hat\Pi_{\widehat{P}|\widehat{\Sigma}^{\Box},n}^\varepsilon\, \hat\sigma^{\Box}_n \, 
    \hat\Pi_{\widehat{P}|\widehat{\Sigma}^{\Box},n}^\varepsilon\nonumber\\
    &\leqslant e^{+n (m-s+2\varepsilon)} \,  \hat\sigma^{\Box}_n\ . 
  \end{align}
  The final condition~\eqref{condition4} follows from \cref{Philippe_lemma6} in
  \cref{appx:lemmas}, \cref{rho1_1} in \cref{ergodic_prop}, and from
  \begin{align}
    \lim_{n \to \infty} \operatorname{tr}\bigl[
    \hat\Pi_{\widehat{P}|\widehat{\Sigma}^{\Box},n}^\varepsilon\, \hat\rho_n \bigr] = 1\ .
    \label{rho1_2}
  \end{align}
  To show \cref{rho1_2}, we use the assumption of ergodicity of $\widehat{P}$.
  Since the Hamiltonian is local and translation invariant, we have
  \begin{align}
    \hat H_n = \sum_{i \in \Lambda_\ell} T_i  (\hat  h_0)\ ,
    \label{Hamiltonian_shift}
  \end{align}
  where $T_i$ is the shift operator.
  Then, denoting by $\operatorname{Proj}\{ \cdots \}$ the projection operator onto a subspace
  satisfying the corresponding condition, we have
  \begin{align}
    \hat\Pi_{\widehat{P}|\widehat{\Sigma}^{\Box}, n}^\varepsilon=  \operatorname{Proj}\left\{ \frac{1}{n} \sum_{i \in \Lambda_\ell}  T_i ( \hat h_0) - \frac{ F_n}{n}
      \in [h-f-\varepsilon, h-f+\varepsilon]\right\}\ ,
  \end{align}
  where
  $h := \lim_{n \to \infty}\frac{1}{n} \operatorname{tr}\bigl[\hat\rho_n \hat H_n\bigr]$ and
  $f := \lim_{n \to \infty} \frac{F_n}{n}$.  For sufficiently large $n$, we have
  $\lvert {\frac{F_n}{n} - f}\rvert  < \frac{\varepsilon}{2}$, and therefore,
  \begin{align}
    \hat \Pi_{\widehat{P}|\widehat{\Sigma}^{\Box}, n}^\varepsilon\geqslant\operatorname{Proj}\left\{  \frac{1}{n} \sum_{i \in \Lambda_\ell}  T_i ( \hat h_0)
    \in [h-\varepsilon/ 2, h+\varepsilon/2 ]\right\}\ .
  \end{align}
  By definition of ergodicity, observables of the form~\eqref{Hamiltonian_shift}
  converge in probability; we have proven \cref{rho1_2}.
\end{proof}

The above proof reduces to the main theorem of Ref.~\cite{Bjelakovic2004CMP_ergodic} in the
special case where $\widehat{\Sigma}^{\Box}$ is i.i.d.\@, i.e.\@, if the system has a
strictly local Hamiltonian with no interaction terms ($r=0$).

Finally, we can ask whether the same theorem holds also for the sequence
$\widehat{\Sigma}$ of reduced states of the full Gibbs state on the infinite lattice.
We show that this is indeed the case.  Because this theorem requires a rigorous
formulation in terms of $C^*$-algebras, we defer the precise claim and proof to
\cref{quantum_theorem_main2} in \cref{appx:Cstar-algebras-construction} .

\subsection{Remarks on ergodicity, mixtures, and the
  KL divergence}

\subsubsection{The mixing property}
A local Gibbs state with a mixing (or clustering) property is ergodic.  However,
we emphasize that the converse is false; ergodicity does not necessarily imply
that the state can be written as a Gibbs state of a local Hamiltonian.

\begin{definition}[Mixing]
  Let $T_{(k)} := T_{(0, \cdots, 0, 1, 0, \cdots, 0)}$ be the shift operator
  corresponding to the one-step shift to the $k$-th direction
  ($k=1,2, \cdots, d$).  A sequence $\widehat{P}$ has the mixing (or clustering)
  property, if it satisfies the consistency condition \eqref{consistency}, and
  if for all $\hat A, \hat B \in \mathcal A_\Lambda$ with $\Lambda$ being
  bounded and if for all $k$, we have
  \begin{align}
    \lim_{m \to \infty} \operatorname{tr}\mathopen{}\left[ \hat\rho_n T_{(k)}^{m}( \hat A) \, \hat B \right]
    = \lim_{n \to \infty}\operatorname{tr}\mathopen{}\left[ \hat\rho_n \hat A \right]
    \operatorname{tr}\mathopen{}\left[ \hat\rho_n\hat B \right]\ ,
    \label{mixing}
  \end{align}
  where $n = (2\ell+1)$ on the left-hand side is taken such that the supports of
  $T_{(k)}^{m}( \hat A)$ and $\hat B$ are included in $\Lambda_\ell$.
\label{def_mixing}
\end{definition}

The equivalence of this definition and the standard definition is proven in
\cref{appx:Cstar-algebras-construction}.  It is well-known that mixing implies
ergodicity (cf.\@ Ref.~\cite{BookRuelle_StatMechRigorous}):

\begin{proposition}
  Any translation-invariant and mixing state is ergodic. 
\end{proposition}

For local operators and the Gibbs state of a local and translation-invariant
Hamiltonian, a stronger property called the exponential clustering property has
been proven for any $\beta > 0$ in one dimension~\cite{Araki1969CMP_Gibbs} and in higher
dimensions for sufficiently high temperature (see, for example,
Ref.~\cite{Tasaki2018JSP_equivalence} and references therein).  Therefore, the quantum Stein's
lemma is proved for two local Gibbs states $\widehat{P}$ and $\widehat{\Sigma}$ at least
for sufficiently high temperature.

\subsubsection{Mixtures of ergodic states}

Consider now the situation in which the state is a mixture of different ergodic
states.  In this setting, ergodicity is broken, and the existence of a
thermodynamic potential is no longer guaranteed.

Let $\widehat{P}^{(k)} := \{ \hat\rho_n^{(k)} \}$ be ergodic states
($k=1,2, \cdots, K < \infty$), and consider their mixture
$\widehat{P}:= \{ \hat\rho_n \}$ with
$\hat\rho_n := \sum_k r_k \hat\rho_n^{(k)}$, where $r_k > 0$ and
$\sum_k r_k = 1$.  We continue to suppose that $\widehat{\Sigma}$ is given by the Gibbs
state of a local and translation-invariant Hamiltonian.  In this setting, we can
show that the min and max divergences are given by the minimal and maximal value
of the KL divergence of the states in the mixture, respectively:

\begin{lemma}
  \label{non_ergodic_KL}
  The spectral divergence rates are split as
  \begin{align}
    {\underline{S}}(\widehat{P}\,\Vert\,\widehat{\Sigma})
    &= \min_k \{ {S}_{1}(\widehat{P}^{(k)}\,\Vert\,\widehat{\Sigma}) \}\ ;
    &
      {\overline{S}}(\widehat{P}\,\Vert\,\widehat{\Sigma})
    &= \max_k \{ {S}_{1}(\widehat{P}^{(k)}\,\Vert\,\widehat{\Sigma}) \}\ ,
    \label{non_ergodic_KLeq}
  \end{align}
  while the KL divergence rate is given by
  \begin{equation}
    {S}_{1}(\widehat{P}\,\Vert\,\widehat{\Sigma})
    = \sum_k r_k {S}_{1}(\widehat{P}^{(k)}\,\Vert\,\widehat{\Sigma})\ .
    \label{mix_KL}
  \end{equation}
\end{lemma}

\begin{proof}
  \Cref{non_ergodic_KLeq} immediately follows from \cref{non_ergodic_lemma} in
  \cref{appx:lemmas}.  To prove~\eqref{mix_KL}, we note that
  $-\operatorname{tr}\bigl[\hat\rho^{(k)}_n \ln \hat\sigma_n \bigr]$ is additive with
  respect to $k$, and thus we only need to show
  ${S}_{1}(\widehat{P}) = \sum_k r_k {S}_{1}(\widehat{P}^{(k)})$.  This in turn follows from
  the fact that the von Neumann entropy satisfies the following inequalities,
  $\sum_k r_k {S}_{1}(\hat\rho_n^{(k)}) \leqslant{S}_{1}(\hat\rho_n) \leqslant\sum_k r_k
  {S}_{1}(\hat\rho_n^{(k)}) + {S}_{1}(\{r_k\})$.
\end{proof}

\subsubsection{The role of the KL divergence for the thermodynamic
  potential}
Usually, we have that if the min and max divergences coincide, then the limiting
values coincide with the limiting value of the KL divergence.
This is because in usual cases, the asymptotic divergences obey
\begin{align}
  {\underline{S}}(\widehat{P}\,\Vert\,\widehat{\Sigma})
  \leqslant{S}_{1}(\widehat{P}\,\Vert\,\widehat{\Sigma})
  \leqslant{\overline{S}}(\widehat{P}\,\Vert\,\widehat{\Sigma})\ .
  \label{S_inequality}
\end{align}
Indeed, this inequality follows in usual cases from the fact that
${S}_{0}(\hat\rho\,\Vert\,\hat\sigma) \leqslant{S}_{1}(\hat\rho\,\Vert\,\hat\sigma) \leqslant{S}_{\infty}(\hat\rho\,\Vert\,\hat\sigma)$ combined with a continuity argument of the
KL divergence in $\hat\rho$ which ensures the inequality persists
after smoothing with $\varepsilon>0$.  Indeed, for
$D(\hat\rho', \hat\rho)\leqslant\varepsilon$, we have
$\lvert {{S}_{1}(\hat\rho\,\Vert\,\hat\sigma) - {S}_{1}(\hat\rho'\,\Vert\,\hat\sigma)}\rvert  \leqslant\lvert {{S}_{1}(\hat\rho) - {S}_{1}(\hat\rho')}\rvert  + 2\varepsilon\lVert {\ln\hat\sigma}\rVert _{\infty}$,
where the first term can be bounded using the Fannes-Audenaert
inequality~\cite{Fannes1973CMP_continuity,Audenaert2007JPA_sharp} and where the
second term behaves as $O(\varepsilon n)$ as long as $\lVert {\ln\hat\sigma}\rVert _{\infty}$ is
at most linear in $n$.  In this case,
$\frac1n{S}_{0}^{\varepsilon}(\hat\rho_n\,\Vert\,\hat\sigma_n) \leqslant\frac1n{S}_{1}(\hat\rho_n\,\Vert\,\hat\sigma_n) + O(\varepsilon)$ and
$\frac1n{S}_{1}(\hat\rho_n\,\Vert\,\hat\sigma_n) - O(\varepsilon) \leqslant\frac1n{S}_{\infty}^{\varepsilon}(\hat\rho_n\,\Vert\,\hat\sigma_n)$, which ensures
that~\eqref{S_inequality} holds.
Notably, while this is the case in most usual settings such as the one
considered in the present paper, this continuity argument does not hold in
general for arbitrary sequences of states and operators.

As a simple toy example, consider a two-level system with states $\lvert {0}\rangle ,\lvert {1}\rangle $,
fix an inverse temperature $\beta>0$, and let $\{\varepsilon_n\}$ be a sequence of
small positive nonzero reals with $\lim_{n\to\infty} \varepsilon_n = 0$.  We
consider the sequence of states $\widehat{P}$ with
$\hat\rho_n = \varepsilon_n \lvert {1}\rangle\hspace*{-0.25ex}\langle{1}\rvert  + (1-\varepsilon_n) \lvert {0}\rangle\hspace*{-0.25ex}\langle{0}\rvert $ and a sequence of
Hamiltonians $\widehat{\mathcal{H}}$ with $\hat{H}_n = (n/\varepsilon_n)\lvert {1}\rangle\hspace*{-0.25ex}\langle{1}\rvert $.  (The sequence is
defined on a single copy of the Hilbert space; it is straightforward to embed
these operators in $\mathscr{H}^{\otimes n}$, though perhaps not in a local and
translation-invariant way.)  The corresponding sequence $\widehat{\Sigma}$ of Gibbs
weights is
$\hat\sigma_n = e^{-\beta \hat{H}_n} = e^{-(\beta n/\varepsilon_n)}\lvert {1}\rangle\hspace*{-0.25ex}\langle{1}\rvert  +
(\hat{I}-\lvert {1}\rangle\hspace*{-0.25ex}\langle{1}\rvert )$.  We can calculate
\begin{align}
  {S}_{1}(\widehat{P}\,\Vert\,\widehat{\Sigma})
  = \lim_{n\to\infty}
  \frac1n{S}_{1}(\hat\rho_n\,\Vert\,\hat\sigma_n)
  = \lim_{n\to\infty} \left\{
  - \frac1n{S}_{1}(\hat\rho_n) + \frac{\beta}{\varepsilon_n}\, \operatorname{tr}\bigl[\hat\rho\,\lvert {1}\rangle\hspace*{-0.25ex}\langle{1}\rvert \bigr]
  \right\}
  =  \beta\ .
\end{align}
For the min divergence and for any $\varepsilon>0$ we have
\begin{align}
  \frac1n {S}_{0}^{\varepsilon}(\hat\rho_n\,\Vert\,\hat\sigma_n) \geqslant\frac1n {S}_{0}(\hat\rho_n\,\Vert\,\hat\sigma_n) =
  -\frac1n\ln\bigl(1+e^{-\beta n/\varepsilon_n}\bigr) \geqslant\frac1n\ln(2)
  \ \xrightarrow{n\to\infty}\ 0\ ,
\end{align}
and hence ${\underline{S}}(\widehat{P}\,\Vert\,\widehat{\Sigma}) \geqslant0$.
On the other hand, for any $\varepsilon>0$ we have that $\varepsilon_n\leqslant\varepsilon$
for $n$ large enough; then for $n$ large enough,
$D\bigl(\lvert {0}\rangle\hspace*{-0.25ex}\langle{0}\rvert ,\hat\rho_n\bigr)\leqslant\varepsilon$ and
\begin{align}
  \frac1n {S}_{\infty}^{\varepsilon}(\hat\rho_n\,\Vert\,\hat\sigma_n)
  \leqslant\frac1n \ln\,\bigl\lVert {\hat\sigma_n^{-1/2}\, \lvert {0}\rangle\hspace*{-0.25ex}\langle{0}\rvert \, \hat\sigma_n^{-1/2} }\bigr\rVert _{\infty} = 0\ ,
\end{align}
and hence ${\overline{S}}(\widehat{P}\,\Vert\,\widehat{\Sigma}) \leqslant0$.  Finally,
recalling~\eqref{eq:min-max-asympt-divergence-rates-ordered}, we find
\begin{align}
  \label{eq:1}
  {\underline{S}}(\widehat{P}\,\Vert\,\widehat{\Sigma}) = {\overline{S}}(\widehat{P}\,\Vert\,\widehat{\Sigma}) = 0\ .
\end{align}
Crucially, the operator $\hat\sigma_n$ has an eigenvalue that is at least exponentially
small in $n$, and $\lVert {\ln\hat\sigma}\rVert _{\infty}$ is superlinear in $n$.
This invalidates the usual continuity argument described above.  Having
$\lVert {\ln\hat\sigma}\rVert _{\infty}$ with such a behavior amounts to having a Hamiltonian
(such as $\hat{H}_n$ in our example) with an energy level that scales superlinearly
in $n$.  Physically, this means that the system does not have a sound
thermodynamic limit; in practice, for instance in the case of all-to-all
coupling, one prefers to normalize the full Hamiltonian to ensure a good
behavior in the thermodynamic limit.  Nevertheless, our toy example shows that
in full generality, the min- and max-divergences can collapse to a single value
and define a thermodynamic potential which does not coincide with the
KL divergence in the thermodynamic limit.

We emphasize that this issue does not appear in usual settings such as the one
considered in the present paper, where the energy is extensive.  Also, this
issue cannot appear with the spectral entropy rates (i.e., if
$\hat\sigma=\hat{I}$), because of the argument above, or alternatively, thanks
to Lemma~3 of Ref.~\cite{Bowen2006arXiv_arbitrary}.  
In those cases, the Kullback-Leibler divergence (or the von Neumann entropy
rate) is the relevant thermodynamic potential that emerges from the
reversibility of the resource theory.

\section{Discussion}
\label{sec:discussion}

Our results provide new insight on the role of ergodicity and typicality in
  many-body systems~\cite{Anshu2016NJP_concentration,Wilming2019PRL_entglergodic}.
Our two main theorems on one hand advance our understanding of the possible
interconversion of states with thermal operations and a limited source of
coherence, and on the other hand establish a generalized quantum Stein's lemma
for lattice systems with local and translation-invariant Hamiltonians.
Together, these theorems prove our main result, namely, that a thermodynamic
potential emerges in the resource theory of thermal operations for all ergodic
states in lattices with a translation-invariant local Hamiltonian.

\paragraph*{Thermal operations involving nonsemiclassical states.}  While
the possible state transformations under thermal operations are well understood
for semiclassical states thanks to the notion of
thermomajorization~\cite{Horodecki2013_ThermoMaj}, the picture becomes
significantly more involved if we consider states that present coherences
between energy
eigenspaces~\cite{Gour2018NatComm_entropic,Lostaglio2015PRX_coherence}.  The min-
and max-divergence no longer represent the distillable work and the work
cost of formation of a state, because in general one requires a suitable
reference frame to accurately carry out those
transformations~\cite{Bartlett2007_refframes,Korzekwa2016NJP_extraction,Gour2018NatComm_entropic,Popescu2018PTRSA_applications}.  Our
\cref{thm:asympt-equipartition-implies-asympt-TO-reversibility} shows, however,
that if the two divergences coincide approximately, then the coherences
that are present in the state are necessarily small in a suitable sense, such
that these transformations become approximately possible after all with only a
small reference frame.  In the thermodynamic limit, the size of the reference
frame becomes negligible.

Our theorem provides a conceptually clear characterization of which states can
be reversibly converted to the thermal state, and hence, for which class of
states the thermodynamic potential emerges.  Namely, approximately reversible
conversion to the thermal state is possible if and only if the min and max
divergences coincide approximately (although the error terms have to be
adjusted in each direction of the proof).

We resort to a crude metric for the amount of coherence that was used in a
process: We allow the use of an ancilla whose Hamiltonian is suitably bounded.
Recently, more refined methods of accounting for coherence have been introduced,
such as via coherent work~\cite{HindsMingo2019Qu_superpositions} or with a
more traditional resource-theoretic
approach~\cite{Marvian2020NC_distillation}.  Using an improved measure of
coherence would allow to clarify the amount of coherence used in the processes
of \cref{thm:asympt-equipartition-implies-asympt-TO-reversibility}.

One could ask for a characterization of which classes of states can be
reversibly converted into one another, without being necessarily reversibly
convertible to the thermal state.  Consider for instance two states with the
same spectrum that is not uniform, both living within a fixed energy subspace:
They can be related by an energy-conserving unitary, but they cannot be
reversibly converted to the thermal state.  In this paper, we have adopted the
convention that a thermodynamic potential should be well defined for the thermal
state itself.  Curiously however, it is also possible to define some kind of
``alternative thermodynamic potentials'' for such classes of states which cannot
include the thermal state.  It is not clear to us what the physical relevance of
such classes of states would be.

We also note that ergodic states have off-diagonal elements that vanish
exponentially, similarly to the behavior encountered in states obeying the eigenstate
thermalization hypothesis
(\cref{lemma:large-coherences-suppressed-if-approximate-equipartition} combined
with~\cref{quantum_theorem_main2t}).  It is then natural to ask whether there
are properties of states that obey the eigenstate thermalization hypothesis
(such as error-correcting properties~\cite{Brandao2019PRL_chainAQECC}) that
can be carried over to ergodic states.

\paragraph*{Asymptotic Equipartition, the Shannon-McMillan theorem, and Stein's lemma} 
The classical Shannon-McMillan
theorem along with its quantum
counterparts provide a collection of AEP statements that play 
an important role in information
theory, statistics, and statistical physics, where ergodic processes are
naturally encountered.  
Because of the stark formal differences between the quantum and the classical
definitions of Markovianity, the quantum versions of these AEP
theorems do not follow directly from their classical counterparts.  Building on
earlier proofs of the quantum Shannon-McMillan
theorem~\cite{Bjelakovic2004IM_lattice,Bjelakovic2005QIP_compression,Ogata2013LMP_shannonmcmillan} and a relative AEP
theorem with respect to product states~\cite{Bjelakovic2004CMP_ergodic}, we finally provide
the full quantum version of the classical relative AEP theorem mentioned above,
which applies to ergodic states relative to Gibbs states of a local Hamiltonian.

A main component of the proof of our main result is a generalized version of
Stein's lemma which is tightly related to the proof techniques of
Ref.~\cite{Bjelakovic2003arXiv_revisted}.  Namely, it suffices to find an
operator obeying a set of simple conditions to conclude that the min and max
divergences collapse, which can be seen partly thanks to ideas from semidefinite
programming~\cite{Dupuis2013_DH,Tomamichel2013_hierarchy}.
By constructing suitable typical projectors using the ergodicity property of the
state, our \cref{quantum_theorem_main2t} exploits this characterization and
provides a new version of Stein's lemma.  The latter applies to situations
beyond i.i.d.\@ states, since we may consider any ergodic state with respect to
any Gibbs state that arises from a local Hamiltonian.

Crucially, the states we consider are spatially ergodic, rather than ergodic
with respect to time evolution.  Spatially ergodic states can have a nontrivial
time evolution, even producing significant changes of macroscopic quantities in
time~\cite{Faist2019PRL_macroscopic}.
Importantly, this shows that one can define a thermodynamic potential that has a
operational interpretation even for certain states that are not in thermodynamic
equilibrium.

By endowing a new class of states with a rigorous, well-justified thermodynamic
potential, one may ask whether or not it is possible to find even larger classes
of states that can be reversibly converted into one another.  Thanks to \cref{non_ergodic_KL}, the thermodynamic potential also emerges for
all finite mixtures of ergodic states with the same thermodynamic potential.
Whether there are more translation-invariant states that have a well-defined
thermodynamic potential in the sense of the present paper is an open question.

One may ask whether or not our results could be generalized to systems that
violate translation-invariance.  It might be possible to treat a weak violation
by adapting the present argument with a suitable control of the relevant error
terms.  For systems that are fundamentally not translation-invariant, one could
instead ask whether ideas from entropy accumulation could be leveraged to prove
bounds on the min and max spectral rates in the thermodynamic limit, using local
properties of the state (or of the local process that generates the state)
rather than symmetry
considerations~\cite{Dupuis2020CMP_accumulation,Dupuis2019IEEETIT_improved}.
Conversely, insights gained from the behavior of the spectral rates in
statistical mechanical systems might provide new ways of proving more general
entropy accumulation theorems which might involve the divergence, the
mutual information, or a channel capacity.

A further natural extension of our work would be to lift our results from
transformations of quantum states to transformations of quantum channels, in
line of the results of Ref.~\cite{Faist2019PRL_thcapacity}.  Can non-i.i.d.\@
quantum channels that have a suitable ergodic property be reversibly converted
into one another?

The quantum Shannon-McMillan theorem moreover holds in a more general and
abstract operator algebra context~\cite{Ogata2013LMP_shannonmcmillan}.  We might
expect that additional AEP results in such settings can be shown using ideas put
forward in the present paper.

Finally, one could attempt to further characterize the min and max divergence
rates in natural situations where they do not coincide.  These quantities are
known to bound any extension of the thermodynamic potential outside of the set
of reversibly interconvertible states~\cite{Lieb2013_entropy_noneq}, and as
such, the interval $[{\underline{S}}(\widehat{P}\,\Vert\,\widehat{\Sigma}),{\overline{S}}(\widehat{P}\,\Vert\,\widehat{\Sigma})]$
provides a ``best possible characterization'' of the thermodynamic behavior of
such states that takes into account the fluctuations in thermodynamic quantities
that persist in the thermodynamic limit.  We expect this to be the case, for
instance, for many-body-localized states, or for states at critical points immediately before spontaneous symmetry breaking. 
The techniques put
forward in the present paper might help derive bounds in such cases, which,
while falling short of a collapse of the min and max divergences, would still
provide a useful characterization for a greater class of states that are far out
of equilibrium.

\begin{acknowledgments}
  The authors are grateful to
  Hiroyasu Tajima, Yoshiko Ogata
  and
  Matteo Lostaglio for valuable discussions.
  TS is supported by JSPS KAKENHI Grant Number JP16H02211 and JP19H05796.
  PhF is supported by the Institute for Quantum Information and Matter (IQIM) at
  Caltech which is a National Science Foundation (NSF) Physics Frontiers Center
  (NSF Grant~{PHY}-{1733907}), from the Department of Energy
  Award~{DE}-{SC0018407}, from the Swiss National Science Foundation (SNSF) via
  the NCCR QSIT and project No.~{200020\_165843}, and from the Deutsche
  Forschungsgemeinschaft (DFG) Research Unit {FOR}~{2724}.
  KK is supported by the Institute for Quantum Information and Matter
  (IQIM) at Caltech which is a National Science Foundation (NSF) Physics
  Frontiers Center (NSF Grant PHY-1733907).
  FB is is supported by the NSF.
\end{acknowledgments}

\appendix

\section{General technical lemmas}
\label{appx:lemmas}

The following \emph{gentle measurement lemma} states that a measurement effect
that is almost certain to appear does not disturb the state
much~\cite{Winter1999IEEETIT_coding,Ogawa2007IEEETIT_making}.

\begin{proposition}
  \label{gentle_lemma}
  For a state $\hat\rho$ and any operator with $0 \leqslant\hat Q \leqslant\hat{I}$, if
  $\operatorname{tr}[\hat\rho\hat Q] \geqslant1-\varepsilon$, then
  \begin{equation}
    \bigl\lVert { \hat\rho- \hat Q^{1/2}\, \hat\rho\, \hat Q^{1/2} }\bigr\rVert _{1}
    \leqslant2 \sqrt{2 \varepsilon}\ .
  \end{equation}
\end{proposition}

The following technical lemmas provide a few variations around the gentle
measurement lemma, dealing with operators that capture most of the weight of a
state.

\begin{lemma}
  \label{Philippe_lemma6}
  Let $\hat Q$ and $\hat Q'$ be projectors.  Suppose that a state $\hat\rho$
  satisfies $\operatorname{tr}\bigl[\hat Q \hat\rho\bigr] \geqslant1-\varepsilon$ and
  $\operatorname{tr}\bigl[\hat Q' \hat\rho\bigr] \geqslant1-\varepsilon'$ for $\varepsilon> 0$,
  $\varepsilon' > 0$.  Then,
  \begin{align}
    \operatorname{Re}\bigl(\operatorname{tr}\bigl[\hat Q \hat Q' \hat\rho\bigr]\bigr)
    \geqslant1 - \varepsilon- \sqrt{\varepsilon'}\ .
    \label{Philippe_lemma6_eq}
  \end{align}
\end{lemma}

\begin{proof}
  We first note that
  \begin{align}
    \operatorname{Re}\operatorname{tr}\bigl[\hat Q \hat Q' \hat\rho\bigr]
    = \operatorname{tr}\bigl[\hat Q \hat\rho\bigr] - \operatorname{Re}\operatorname{tr}\bigl[\hat Q (\hat{I}- \hat Q' )\hat\rho\bigr]
    \geqslant1 - \varepsilon-  \operatorname{Re}\operatorname{tr}\bigl[\hat Q (\hat{I}- \hat Q' )\hat\rho\bigr]\ .
  \end{align}
  From the Schwarz inequality, we have
  \begin{equation}
    \operatorname{Re}\operatorname{tr}\bigl[\hat Q (\hat{I}- \hat Q' )\hat\rho\bigr]
    \leqslant\sqrt{\operatorname{tr}[\hat Q \hat\rho] \operatorname{tr}[(\hat{I}- \hat Q') \hat\rho]} \leqslant\sqrt{\varepsilon'}\ ,
  \end{equation}
  where we used that $\hat Q$ and $\hat{I}- \hat Q'$ are projectors.  Therefore,
  we obtain \cref{Philippe_lemma6_eq}.
\end{proof}

\begin{lemma}
  \label{Philippe_lemma7}
  Let $\hat W$ be an operator with $\lVert { \hat W }\rVert _{\infty} \leqslant1$.  Suppose that a
  subnormalized state $\hat\rho\in{\mathcal{S}_\leq}(\mathscr{H})$ satisfies
  $\operatorname{Re}\bigl(\operatorname{tr}\bigl[\hat W \hat\rho\bigr]\bigr) \geqslant1-\varepsilon$ with
  $\varepsilon> 0$.  Then, both following statements are true:
  \begin{enumerate}[label=(\alph*)]
  \item $\operatorname{tr}\bigl[\hat W^\dagger \hat W \hat\rho\bigr] \geqslant1-2\varepsilon$ and
    $\operatorname{tr}\bigl[\hat W \hat W^\dagger \hat\rho\bigr] \geqslant1-2\varepsilon$\,;
  \item
    $\operatorname{tr}\bigl[(\hat{I}- \hat W ) (\hat{I}- \hat W^\dagger) \hat\rho\bigr] \leqslant2\varepsilon$\,.
  \end{enumerate}
\end{lemma}

\begin{proof}
  \begin{enumerate}[label=(\alph*)]
  \item From the Cauchy-Schwarz inequality,
    \begin{align}
      \operatorname{tr}\bigl[\hat W^\dagger \hat W \hat\rho\bigr]
      \geqslant\operatorname{tr}[\hat\rho]\cdot \operatorname{tr}\bigl[\hat W^\dagger \hat W \hat\rho\bigr]
      \geqslant\left( \operatorname{Re}\operatorname{tr}[\hat W \hat\rho] \right)^2
      \geqslant1 - 2 \varepsilon\ .
    \end{align}
    We can show the second inequality in the same manner.
  \item This follows from
    \begin{align}
      \operatorname{tr}\bigl[(\hat{I}- \hat W ) (\hat{I}- \hat W^\dagger) \hat\rho\bigr]
      = 1 - 2 \operatorname{Re}\operatorname{tr}[\hat W \hat\rho] + \operatorname{tr}[\hat W \hat W^\dagger \hat\rho]
      \leqslant1- 2(1-\varepsilon) +1 = 2 \varepsilon\ ,
    \end{align}
    where we used $\lVert { \hat W }\rVert _{\infty} \leqslant1$.
    \hfill\qed\end{enumerate}
\end{proof}

Next we show that for a mixture of states, the min and max spectral rates are
given by the smallest or largest spectral rate in the mixture, respectively.

\begin{proposition}\label{non_ergodic_lemma}
  Consider a sequence of states $\widehat{P}:= \{ \hat\rho_n \}$ where each state
  is given by a mixture $\hat\rho_n = \sum_{k=1}^K r_k \hat\rho_n^{(k)}$ for a
  given probability distribution $\{ r_k \}_{k=1}^{K}$ independent of $n$, and
  consider the individual sequences $\widehat{P}^{(k)} = \{ \hat\rho_n^{(k)} \}$.
  Then, the lower and the upper divergence rates of $\widehat{P}$ relative to a
  sequence of positive operators $\widehat{\Sigma}$ are given by
  \begin{align}
    {\underline{S}}(\widehat{P}\,\Vert\,\widehat{\Sigma})
    &= \min_k \, \bigl\{ {\underline{S}}(\widehat{P}^{(k)}\,\Vert\,\widehat{\Sigma}^{(k)}) \bigr\}\ ;
    &
    {\overline{S}}(\widehat{P}\,\Vert\,\widehat{\Sigma})
    &= \max_k \, \bigl\{ {\overline{S}}(\widehat{P}^{(k)}\,\Vert\,\widehat{\Sigma}^{(k)}) \bigr\}\ .
  \end{align}
\end{proposition}

This proposition immediately follows from the following three lemmas.

\begin{lemma}
  \label{convex_lemma}
  Consider a mixture of states $\hat\rho= \sum r_k \hat\rho^{(k)}$ with a
  probability distribution $\{ r_k \}$.  Let $\hat\tau$ be a quantum state such
  that $F^2(\hat\rho,\hat\tau) \geqslant 1 - \varepsilon^2$.  Then there exists a
  probability distribution $\{ r_k' \}$ and a collection of states
  $\hat\tau^{(k)}{}'$ such that
  \begin{align}
    \hat\tau&= \sum_k r_k' \hat\tau^{(k)}\ ;
    &
    D\bigl(\{ r_k \}, \{ r'_k \}\bigr) &\leqslant \varepsilon\ ;
    &
    D\bigl(\hat\rho^{(k)}, \hat\tau^{(k)}\bigr)
    &\leqslant \frac{2\varepsilon}{r_k}\ .
  \end{align}
\end{lemma}

\begin{proof}
  Call our system of interest $A$, and consider a copy $B\simeq A$.  Let
  $\{ \lvert {j}\rangle _{A} \}, \{ \lvert {j}\rangle _B \}$ be orthonormal bases of $A$ and $B$,
  respectively, and let $\lvert {\Phi }\rangle := \sum_j \lvert {j}\rangle _{A} \lvert {j}\rangle _{B}$ be the
  reference unnormalized maximally entangled state.  Consider the following
  purification of $\hat\rho^{(k)}$,
  \begin{align}
    \lvert {\hat\rho^{(k)}}\rangle _{AB}
    = \bigl(\bigl(\hat\rho_{A}^{(k)}\bigr)^{1/2}\otimes\hat{I}_B\bigr) \, \lvert {\Phi}\rangle _{AB}\ .
  \end{align}
  Let $C$ be a register with an orthonormal basis $\{ \lvert {k}\rangle _C \}$ and consider
  the following purification of $\hat\rho_C$,
  \begin{align}
    \lvert {\hat\rho}\rangle _{AB}
    = \sum_k \sqrt{r_k} \, \lvert {\rho^{(k)}}\rangle  \lvert {k}\rangle _{C}\ .
  \end{align}
  From Uhlmann's theorem, there exists a purification $\lvert {\hat\tau}\rangle _{ABC}$ of
  $\hat\tau_A$ such that
  \begin{align}
    F(\lvert {\hat\rho}\rangle , \lvert {\hat\tau}\rangle ) = F(\hat\rho, \hat\tau)
    \geqslant\sqrt{1 - \varepsilon^2}\ .
  \end{align}
  Invoking the Fuchs-van de Graaf relations between the fidelity and the trace
  distance~\cite{Fuchs1999IEEETIT_distinguishability,BookNielsenChuang2000},
  $1 - F(\cdot,\cdot) \leqslant D(\cdot,\cdot) \leqslant \sqrt{1 -
    F^2(\cdot,\cdot)}$, we find that
  $D(\lvert {\hat\rho}\rangle , \lvert {\hat\tau}\rangle ) \leqslant \varepsilon$.
  Now, define
  \begin{align}
    r'_{k} &:= \operatorname{tr}\bigl( \lvert {k}\rangle\hspace*{-0.25ex}\langle{k}\rvert _C \, \lvert {\hat\tau}\rangle\hspace*{-0.25ex}\langle{\hat\tau}\rvert _{ABC} \bigr)\ ;
    &
      \hat\tau^{(k)} &:= \frac1{r'_k}\,\operatorname{tr}_{BC}\bigl( \lvert {k}\rangle\hspace*{-0.25ex}\langle{k}\rvert _C\,\lvert {\hat\tau}\rangle\hspace*{-0.25ex}\langle{\hat\tau}\rvert _{ABC} \bigr)\ .
  \end{align}
  From the monotonicity of the trace norm under CPTP maps, we have $D(\{ r_k \}, \{ r'_k \} ) \leqslant \varepsilon$,
  where here the trace distance is calculated between the two classical
  probability distributions, which is known as the total variational distance.
  Furthermore, the trace norm cannot increase under any CP and
  trace-nonincreasing maps, and hence,
  \begin{align}
    \frac12\bigl\lVert { r_k\hat\rho^{(k)} - r'_k \hat\tau^{(k)} }\bigr\rVert _{1}
    \leqslant
    \frac12\bigl\lVert { \lvert {\hat\rho}\rangle\hspace*{-0.25ex}\langle{\hat\rho}\rvert - \lvert {\hat\tau}\rangle\hspace*{-0.25ex}\langle{\hat\tau}\rvert }\bigr\rVert _{1}
    = D(\lvert {\hat\rho}\rangle , \lvert {\hat\tau}\rangle ) \leqslant \varepsilon\ .
  \end{align}
  This implies
  \begin{multline}
    D\bigl( \hat\rho^{(k)}, \hat\tau^{(k)} \bigr)
    = \frac1{2r_k}\bigl\lVert {r_k\hat\rho^{(k)} - r_k\hat\tau^{(k)}}\bigr\rVert _{1}
    \\
    \leqslant\frac1{r_k} \biggl( \frac12\bigl\lVert {r_k\hat\rho^{(k)} - r'_k\hat\tau^{(k)}}\bigr\rVert _{1}
      + \frac{\lvert {r'_k - r_k}\rvert }{2}\,\bigl\lVert {\hat\tau^{(k)}}\bigr\rVert _{1} \biggr)
      \leqslant \frac{2\varepsilon}{r_k}\ ,
  \end{multline}
  which completes the proof.
\end{proof}

\begin{lemma}
  \label{mix_max_lemma}
  Consider a mixture of states $\hat\rho= \sum_{k=1}^{K} r_k \hat\rho^{(k)}$
  with a probability distribution $\{ r_k \}$.  Let $\varepsilon>0$ be such that
  $2 \sqrt{\varepsilon} < r_k$ for all $k$. Then
  \begin{align}
    \label{mix_max1}
    {S}_{\infty}^{\varepsilon}(\hat\rho\,\Vert\,\hat\sigma)
    &\leqslant\max_k {S}_{\infty}^{\varepsilon}(\hat\rho^{(k)}\,\Vert\,\hat\sigma)\ ;
    \\
    \label{mix_max2}
    {S}_{\infty}^{\varepsilon}(\hat\rho\,\Vert\,\hat\sigma)
    &\geqslant\max_k \left\{ {S}_{\infty}^{{2 \sqrt{2\varepsilon} / r_k}}(\hat\rho^{(k)}\,\Vert\,\hat\sigma)
      + \ln(r_k - 2\sqrt{\varepsilon} ) \right\}\ .
  \end{align}
\end{lemma}

\begin{proof}
  We first show inequality~\eqref{mix_max1}.  For each $k$, there exists
  $\hat \tau^{(k)} \in B^\varepsilon(\hat\rho^{(k)} )$ such that
  ${S}_{\infty}^{\varepsilon}(\hat\rho^{(k)}\,\Vert\,\hat\sigma) =
  {S}_{\infty}(\hat\tau^{(k)}\,\Vert\,\hat\sigma)$.  Let
  $\hat\tau:= \sum_k r_k \hat\tau^{(k)}$, which is a candidate for
  minimization in ${S}_{\infty}^{\varepsilon}(\hat\rho\,\Vert\,\hat\sigma)$, because
  $D(\hat\tau, \hat\rho) \leqslant \sum_k r_k D(\hat\tau^{(k)},
  \hat\rho^{(k)}) \leqslant\varepsilon$, using the joint convexity of the trace
  distance.  Then,
  \begin{multline}
    {S}_{\infty}^{\varepsilon}(\hat\rho\,\Vert\,\hat\sigma)
    \leqslant{S}_{\infty}(\hat\tau\,\Vert\,\hat\sigma)
    \leqslant\ln \sum_k r_k \lVert {\hat\sigma^{-1/2} \hat\tau^{(k)} \hat\sigma^{-1/2} }\rVert _{\infty}
    \\
    \leqslant\max_k \ln \lVert { \hat\sigma^{-1/2} \hat\tau^{(k)} \hat\sigma^{-1/2} }\rVert _{\infty}
    = \max_k {S}_{\infty}^{\varepsilon}(\hat\rho^{(k)}\,\Vert\,\hat\sigma)\ .
  \end{multline}

  We next show inequality~\eqref{mix_max2}.  There exists
  $\hat\tau\in B^\varepsilon(\hat\rho)$ such that
  ${S}_{\infty}^{\varepsilon}(\hat\rho\,\Vert\,\hat\sigma) = {S}_{\infty}(\hat\tau\,\Vert\,\hat\sigma)$.  By
  the Fuchs-van de Graaf
  inequalities~\cite{Fuchs1999IEEETIT_distinguishability,BookNielsenChuang2000},
  we have $F(\hat\rho,\hat\tau) \geqslant 1 - D(\hat\rho,\hat\tau)$ and thus
  $F^2(\hat\rho,\hat\tau) \geqslant 1 - 2\varepsilon$.  Let
  $\{ \hat\tau^{(k)} \}$ be quantum states and $\{ r'_k \}$ be a probability
  distribution that are given by \cref{convex_lemma}, such that
  $D(\{ r_k \}, \{ r'_k \}) \leqslant \sqrt{2\varepsilon}$ and
  $D(\hat\rho^{(k)}, \hat\tau^{(k)}) \leqslant 2\sqrt{2\varepsilon}/r_k$.  Noting
  that $r_k' \hat \tau^{(k)} \leqslant\hat \tau$, we have
  \begin{multline}
    {S}_{\infty}(\hat\tau\,\Vert\,\hat\sigma)
    \geqslant\ln\, \bigl\lVert { \hat\sigma^{-1/2} r_k' \hat\tau^{(k)} \hat\sigma^{-1/2} }\bigr\rVert _{\infty}
    = {S}_{\infty}(\hat\tau^{(k)}\,\Vert\,\hat\sigma) + \ln r_k' 
      \\
      \geqslant{S}_{\infty}^{{2 \sqrt{2\varepsilon}/r_k}}(\hat\rho^{(k)}\,\Vert\,\hat\sigma)
      + \ln (r_k - 2\sqrt{\varepsilon} )\ ,
  \end{multline}
  which implies inequality~\eqref{mix_max2}.
\end{proof}

\begin{lemma}
  \label{mix_min_lemma}
  Consider a mixture of states $\hat\rho= \sum_{k=1}^{K} r_k \hat\rho^{(k)}$
  with a probability distribution $\{ r_k \}$, and let $\varepsilon>0$. Then
  \begin{align}
    {S}_{0}^{\varepsilon}(\hat\rho\,\Vert\,\hat\sigma)
    \geqslant\min_k \left\{
    {S}_{0}^{\varepsilon}(\hat\rho^{(k)}\,\Vert\,\hat\sigma)  - \ln K \right\}\ .
    \label{mix_min1}
    \\
    {S}_{0}^{\varepsilon}(\hat\rho\,\Vert\,\hat\sigma)
    \leqslant\min_k {S}_{0}^{{2\sqrt{2\varepsilon}/r_k}}(\hat\rho^{(k)}\,\Vert\,\hat\sigma)\ .
    \label{mix_min2} 
  \end{align}
\end{lemma}

\begin{proof}
  We first show inequality~\eqref{mix_min1}.  For each $k$, there exists
  $\hat\tau^{(k)} \in B^\varepsilon(\hat\rho^{(k)} )$ such that
  ${S}_{0}^{\varepsilon}(\hat\rho^{(k)} \,\Vert\, \hat\sigma) =
  {S}_{0}(\hat\tau^{(k)}\,\Vert\,\hat\sigma)$.  Let
  $\hat\tau:= \sum_k r_k \hat\tau^{(k)}$, which is a candidate for
  maximization in ${S}_{0}^{\varepsilon}(\hat\rho\,\Vert\,\hat\sigma)$.  We note that
  $\hat{P}_{\hat\tau} \leqslant \sum_k \hat{P}_{\hat\tau^{(k)}}$, 
  because the kernel of $\hat\tau$ is larger than the intersection of the kernels of $\hat\tau^{(k)}$'s.
  Therefore,
  \begin{multline}
    {S}_{0}^{\varepsilon}(\hat\rho\,\Vert\,\hat\sigma)
    \geqslant{S}_{0}(\hat\tau\,\Vert\,\hat\sigma)
    = - \ln \operatorname{tr}[\hat P_{\hat\tau} \hat\sigma]
    \geqslant-\ln \sum_k \operatorname{tr}[\hat P_{\hat\tau^{(k)}} \hat\sigma]
      \\
    \geqslant-\ln \left( K \max_k \operatorname{tr}[\hat P_{\hat\tau^{(k)}} \hat\sigma] \right)
    = \min_k \left\{ {S}_{0}^{\varepsilon}(\hat\rho^{(k)}\,\Vert\,\hat\sigma)  - \ln K \right\}\ .
  \end{multline}
  We next show inequality~\eqref{mix_min2}.  There exists
  $\hat\tau\in B^\varepsilon(\hat\rho)$ such that
  ${S}_{0}^{\varepsilon}(\hat\rho\,\Vert\,\hat\sigma) = {S}_{0}( \hat\tau\,\Vert\,\hat\sigma)$.
  By the Fuchs-van de Graaf inequalities, we have
  $F^2(\hat\rho,\hat\tau) \geqslant 1 - 2\varepsilon$ as above.  Let
  $\{ \hat\tau^{(k)} \}$ be states and $\{ r'_k \}$ be a probability
  distribution given by \cref{convex_lemma}.  For all $k$,
  \begin{equation}
    \operatorname{tr}[\hat P_{\hat\tau} \hat\sigma] \geqslant\operatorname{tr}[\hat P_{\hat\tau^{(k)}} \hat\sigma]\ ,
  \end{equation}
  and therefore
  \begin{equation}
    {S}_{0}( \hat\tau\,\Vert\, \hat\sigma)
    \leqslant{S}_{0}(\hat\tau^{(k)}\,\Vert\,\hat\sigma)
    \leqslant{S}_{0}^{{2\sqrt{2\varepsilon}/r_k}}(\hat\rho^{(k)}\,\Vert\, \hat\sigma)\ ,
  \end{equation}
  which implies  inequality~\eqref{mix_min2}.
\end{proof}

\section{Properties of our thermodynamic framework and convertibility proof for Gibbs-preserving maps}
\label{appx:thermodynamic-lemmas}

In this section we derive a collection of useful properties of thermodynamic
transformations that were introduced in \cref{sec:thermodynamic-operations}, and
provide a simplified version of
\cref{prop:approx-equipartition-implies-reversibility-under-TO} that is
specialized to Gibbs-preserving maps.  

The partial isometry in the definition of a thermal operation commutes with the
system-and-bath Hamiltonian in the following sense.
\begin{proposition}
  \label{prop:energypres-partial-isom-V-commutes-with-H}
  Let $K,L$ be systems with Hamiltonians $\hat{H}_K,\hat{H}_L$ and let
  $\hat V_{K\to L}$ be a partial isometry such that
  $\hat V_{K\to L}\hat{H}_K = \hat{H}_L \hat V_{K\to L}$.  Then
 \begin{align}
   [\hat V^\dagger_{K\leftarrow L} \hat V_{K\to L}, \hat{H}_K] &= 0\ ;
   &
   [\hat V_{K\to L} \hat V^\dagger_{K\leftarrow L}, \hat{H}_L] &= 0\ .
 \end{align}
 In consequence, $\hat V_{K\to L}$ is a mapping of a subset of initial energy
 eigenstates on $K$ to some final energy eigenstates on $L$.
\end{proposition}
\begin{proof}
  We compute directly $[\hat V^\dagger_{K\leftarrow L} \hat V_{K\to L}, \hat{H}_K] = \hat V^\dagger_{K\leftarrow L} (\hat{V}_{K\to{}L}\hat{H}_K -
  \hat{H}_L\hat{V}_{K\to{}L}) + (\hat{V}^\dagger_{K\leftarrow{}L}\hat{H}_L -
  \hat{H}_K\hat{V}^\dagger_{K\leftarrow{}L})\hat{V}_{K\to L} = 0$ and similarly
  $[\hat V_{K\to L} \hat V^\dagger_{K\leftarrow L}, \hat{H}_L]=0$.
\end{proof}

Now we show that any partial isometry that is compatible with the system
Hamiltonian (i.e., one that maps the input Hamiltonian to the output Hamiltonian
on the range of the partial isometry) can be dilated into a full
energy-conserving unitary on a larger system from which the partial isometry is
recovered by preparing an ancilla in a pure state and post-selecting on a
specific measurement outcome of an ancilla on the output of the unitary.  The
present proof is partly adapted
from~\cite[Proposition~C.2]{Faist2021CMP_impl}.

\begin{proposition}[Dilation of a partial energy-conserving isometry]
  \label{lemma:dilation-energy-conserving-partial-isometry}
  Consider systems $K,L$ with Hamiltonians $\hat{H}_K,\hat{H}_{L}$.  Let
  $\hat V_{K\to L}$ be a partial isometry such that
  $\hat V_{K\to L}\hat{H}_K = \hat{H}_L \hat V_{K\to L}$.
  Let $M$ be a system with Hamiltonian $\hat{H}_M$, and suppose that there exist nontrivial
  systems $\bar{K}$ and $\bar{L}$
  with respective Hamiltonians $\hat{H}_{\bar{K}}$, $\hat{H}_{\bar{L}}$ along with unitaries
  $\hat{U}'_{K\bar{K}\to M}$, $\hat{U}'_{L\bar{L}\to M}$ such that
  $\hat{U}'_{K\bar{K}\to M}(\hat{H}_K + \hat{H}_{\bar K}) = \hat{H}_M \hat{U}'_{K\bar{K}\to M}$ and
  $\hat{U}''_{L\bar{L}\to M}(\hat{H}_L + \hat{H}_{\bar L}) = \hat{H}_M \hat{U}''_{L\bar{L}\to M}$.
  Let
  $\lvert {\mathrm{i}}\rangle _{\bar{K}}, \lvert {\mathrm{f}}\rangle _{\bar{L}}$
  be two eigenstates of $\hat{H}_{\bar{K}}$ and $\hat{H}_{\bar{L}}$ of the
  same energy $e_{\mathrm{i}}=\langle {\mathrm{i}}\hspace*{0.2ex}\vert\hspace*{0.2ex}{\hat{H}_{\bar{K}}}\hspace*{0.2ex}\vert\hspace*{0.2ex}{\mathrm{i}}\rangle _K = e_{\mathrm{f}} = \langle {\mathrm{f}}\hspace*{0.2ex}\vert\hspace*{0.2ex}{\hat{H}_{\bar{L}}}\hspace*{0.2ex}\vert\hspace*{0.2ex}{\mathrm{f}}\rangle _{\bar{L}}$.
  Then there exists a unitary $\hat{U}_{M}$ such that
  $[\hat{U}_{M}, \hat{H}_M] = 0$ and
  \begin{align}
    (\hat{I}_L\otimes \langle{\mathrm{f}}\rvert_{\bar{L}})\;
        \hat{U}_{L\bar{L}\leftarrow M}''^\dagger \,
        \hat{U}_{M} \,
        \hat{U}'_{K\bar{K}\to M} \;
    (\hat{I}_K \otimes \lvert{\mathrm{i}}\rangle_{\bar{K}})\;
    \hat{\Pi}_K
    =
    \hat{V}_{K\to L}\ ,
    \label{eq:dilation-partial-energy-conserving-isometry-complicated}
  \end{align}
  where $\hat{\Pi}_K = \hat{V}^\dagger \hat{V}$ is the projector onto
  the support of $\hat{V}$.
  
  Furthermore, we can remove $\hat{\Pi}_K$ from~\eqref{eq:dilation-partial-energy-conserving-isometry-complicated}
  under the following additional assumption.
  Let $\{ (\alpha_j, \mu_j) \}_{j=1}^{J}$ be the energy
  eigenvalues
  with the corresponding multiplicities of all energy eigenstates of $K$ that are
  in the kernel of $\hat{V}$.
  Let $\{ (\beta_\ell, \nu_\ell) \}_{\ell}$ be the energy eigenvalues
  with the corresponding multiplicities of all energy eigenstates $\lvert{\beta}\rangle_{L\bar{L}}$
  of $\hat{H}_L + \hat{H}_{\bar{L}}$ that have no overlap with $\hat{I}_L\otimes\lvert{\mathrm{f}}\rangle\hspace*{-0.25ex}\langle{\mathrm{f}}\rvert_{\bar{L}}$,
  i.e., for which $
  (\hat{I}_L\otimes\lvert{\mathrm{f}}\rangle\hspace*{-0.25ex}\langle{\mathrm{f}}\rvert_{\bar{L}}) \,
  \lvert{\beta}\rangle_{L\bar{L}} = 0$.
  Suppose that for each $(\alpha_j,\mu_j)$ (for $j=1,\ldots,J$),
  there exists a corresponding $\ell$ with
  $\beta_\ell=e_{\mathrm{i}}+\alpha_j$ and $\nu_\ell\geqslant\mu_j$.
  Then there exists a unitary operator $\hat{U}_{M}$ with $[\hat{H}_M, \hat{U}_M] = 0$
  and such that
  \begin{align}
    (\hat{I}_L\otimes \langle{\mathrm{f}}\rvert_{\bar{L}})\;
        \hat{U}_{L\bar{L}\leftarrow M}''^\dagger \,
        \hat{U}_{M} \,
        \hat{U}'_{K\bar{K}\to M} \;
    (\hat{I}_K \otimes \lvert{\mathrm{i}}\rangle_{\bar{K}})
    =
    \hat{V}_{K\to L}\ .
    \label{eq:dilation-partial-energy-conserving-isometry-complicated-noPi}
  \end{align}
\end{proposition}

Before delving into the proof of \cref{lemma:dilation-energy-conserving-partial-isometry} we
issue a few remarks to provide a better picture of the consequences of this general
proposition and to identify a few interesting special cases. 

\begin{enumerate}[label=(\alph*)]
\item The operator $\hat\Pi_K$
  in~\eqref{eq:dilation-partial-energy-conserving-isometry-complicated}
  can be replaced by an operator $\hat\Pi'_{L}$ acting \emph{after} the unitaries,
  where $\hat\Pi'_{L} = \hat{V}\hat{V}^\dagger$ is the projector onto the
  range of $\hat{V}$.

\item If $K=L$ and $H_K=H_L$, we can choose $M=K=L$ with trivial systems
  $\bar{K}, \bar{L}$.  With this choice of $M$,
  the projector in~\eqref{eq:dilation-partial-energy-conserving-isometry-complicated}
  is always necessary unless $\hat{V}_{K\to L}$ is already unitary.
  (The projector can be removed by choosing a larger system $M$, see below.)

\item
  The additional assumption in the second part of the proposition amounts
  to requiring that, for the given $\lvert{\mathrm{i}}\rangle _{\bar{K}},
  \lvert{\mathrm{f}}\rangle_{\bar{L}}$, it is possible to map
  the support of
  $(\hat{I}_K-\hat{\Pi}_{K})\otimes
  \lvert{\mathrm{i}}\rangle\hspace*{-0.25ex}\langle{\mathrm{i}}\rvert _{\bar{K}}$
  (i.e., the space spanned by all eigenstates outside of the support of $\hat{V}_{K\to L}$
  and tensored with $\lvert{\mathrm{i}}\rangle _{\bar{K}}$), into a space
  of the global output system such that the mapping is energy conserving
  and such that the resulting space has no overlap with
  $\lvert{\mathrm{f}}\rangle_{\bar{L}}$.  As long as the input state on $\bar{K}$
  is initialized in the state $\lvert{\mathrm{i}}\rangle _{\bar{K}}$, then
  projecting the output onto $\lvert{\mathrm{f}}\rangle_{\bar{L}}$ automatically
  ensures that the input state already lies within the projector $\hat{\Pi}'_L$.
  (Equivalently, the projector $\hat{\Pi}_K$ on the input becomes redundant.)

\item \label{item:dilation-partial-energy-conserving-isometry-remarks--add-Q-i-f-states}
  For any $K,L$ and for a general choice of $M$, $\bar{K}$, $\bar{L}$
  with corresponding Hamiltonians along with energy-preserving embedding
  unitaries $\hat{U}'_{K\bar{K}\to M}$, $\hat{U}''_{L\bar{L}\to M}$,
  there always exists a qubit system $Q$ with some Hamiltonian $H_Q$
  such that there exist $\lvert{\mathrm{i}}\rangle_{\bar{K}Q}$ and
  $\lvert{\mathrm{f}}\rangle_{\bar{L}Q}$ with the same eigenenergy.
  
  This statement is shown as follows.  We first pick any two energy eigenstates
  $\lvert{\mathrm{i}_0}\rangle_{\bar{K}}$ and
  $\lvert{\mathrm{f}_0}\rangle_{\bar{L}}$ of respective energies $e_{\mathrm{i}}$
  and $e_{\mathrm{f}}$.
  We then introduce a qubit $Q$ with the Hamiltonian
  $H_Q = q_{0} \lvert{0}\rangle\hspace*{-0.25ex}\langle{0}\rvert + 
  q_{1} \lvert{1}\rangle\hspace*{-0.25ex}\langle{1}\rvert$,
  with $q_0 = c-e_{\mathrm{i}}$ and $q_1 = c- e_{\mathrm{f}}$ for
  any chosen constant $c$.
  Define
  $M' = M\otimes Q$, $\bar{K}'=\bar{K}\otimes Q$, $\bar{L}'=\bar{L}\otimes Q$, etc.,
  along with $\lvert{\mathrm{i}}\rangle _{\bar{K}'} =
  \lvert{\mathrm{i}_0}\rangle _{\bar{K}}\otimes \lvert{0}\rangle _{Q}$
  and $\lvert{\mathrm{f}}\rangle _{\bar{L}'} =
  \lvert{\mathrm{f}_0}\rangle _{\bar{L}}\otimes \lvert{1}\rangle _{Q}$,
  observing that $\lvert{\mathrm{i}}\rangle _{\bar{K}'}$ and
  $\lvert{\mathrm{f}}\rangle _{\bar{L}'}$ are both energy eigenstates
  with energy $c$.
  
\item \label{item:dilation-partial-energy-conserving-isometry-remarks--add-Qp-for-noPi}
  For any $M,K,L,\bar{K},\bar{L}, \lvert{\mathrm{i}}\rangle _{\bar{K}} ,
  \lvert{\mathrm{f}}\rangle _{\bar{L}}$ satisfying the first part of the proposition,
  we can always introduce a qubit system $Q'$ with a degenerate Hamiltonian
  $H_{Q'} = c'$ for some arbitrary constant $c'$, and define
  $M'' = M\otimes Q'$, $\bar{K}''=\bar{K}\otimes Q'$, $\bar{L}''=\bar{L}\otimes Q'$, etc.,
  along with $\lvert{\mathrm{i}'}\rangle _{\bar{K}''} =
  \lvert{\mathrm{i}}\rangle _{\bar{K}}\otimes \lvert{0}\rangle _{Q'}$
  and $\lvert{\mathrm{f}'}\rangle _{\bar{L}''} =
  \lvert{\mathrm{f}}\rangle _{\bar{L}}\otimes \lvert{1}\rangle _{Q'}$,
  such that the additional condition of the second part of the proposition is satisfied.
  Indeed,
  from the unitary $\hat{U}_M$ given by the proposition without the extra qubit,
  we can define $\hat{\tilde{U}}_{MQ'}
  = (\hat{U}_M\otimes\hat{I}_{Q'})
  ((\hat{\Pi}_K\otimes\hat{I}_{\bar{K}}\otimes
  (\lvert{0}\rangle\hspace*{-0.25ex}\langle{1}\rvert
    +\lvert{0}\rangle\hspace*{-0.25ex}\langle{1}\rvert)_{Q'}
  + ((\hat{I}_K-\hat{\Pi}_K)\otimes\hat{I}_{\bar{K} Q'})
  $, i.e., $\hat{\tilde{U}}_{MQ'}$ conditionally flips the bit $Q'$ if the input on $K$
  is in the support of $\hat{V}_{K\to L}$, before applying $\hat{U}_M$.
  The effect of $\tilde{\hat{U}}_{M'}$ is to map all states of the form
  $\lvert{\psi'}\rangle_K \otimes \lvert{\mathrm{i}}\rangle_{\bar{K}}\otimes\lvert{0}\rangle_{Q'}$
  onto states with the $Q'$ system remaining in the state
  $\lvert{0}\rangle_{Q'}$, ensuring that there
  is no overlap with $\lvert{\mathrm{f}'}\rangle _{\bar{L}'}$.
  
\item The qubits introduced in
  Points~\ref{item:dilation-partial-energy-conserving-isometry-remarks--add-Q-i-f-states}
  and~\ref{item:dilation-partial-energy-conserving-isometry-remarks--add-Qp-for-noPi}
  may evidently be chosen to be larger systems that contain such qubits
  as subspaces.

\item For any $K,L$ and for a general choice of $M$, $\bar{K}$, $\bar{L}$
  with corresponding Hamiltonians along with energy-preserving embedding
  unitaries $\hat{U}'_{K\bar{K}\to M}$, $\hat{U}''_{L\bar{L}\to M}$,
  there might not always exist $\lvert{\mathrm{i}}\rangle_{\bar{K}}$ and $\lvert{\mathrm{f}}\rangle_{\bar{L}}$
  with the same eigenenergy, even if $\hat{V}_{K\to L} \neq 0$.
  As a counterexample, consider systems $K,L,\bar{K},\bar{L}$
  where the system $K$ has energy levels $\{ 0, 1 \}$, the system
  $\bar{K}$ has levels $\{0, 1\}$, the system
  $L$ has levels $\{ -2, -1, -1, 0 \}$, and the system $\bar{L}$
  is trivial with the single level $\{ 2 \}$.  In both cases,
  the joint energy levels are $\{0, 1, 1, 2\}$, and $\hat{V}_{K\to L}$
  can be nonzero by mapping the $0$ level of $K$ to the $0$ level of $L$.
  Yet, $\bar{K}$ and $\bar{L}$ do not share an energy level of same energy.

\item For arbitrary $K,L$, a
  simple choice for the system $M$ is $M = K\otimes L$ with $\bar{K}=L$,
  $\hat{H}_{\bar{K}}=\hat{H}_L$, $\bar{L}=K$, $\hat{H}_{\bar{L}} = \hat{H}_K$, along
  with the trivial identity embedding maps $\hat{U}'_{K\bar{K}\to M} = \hat{I}_{KL\to M}$,
  $\hat{U}'_{L\bar{L}\to M} = \hat{I}_{KL\to M}$.  There always
  exist $\lvert{\mathrm{i}}\rangle_{\bar{K}}$ and $\lvert{\mathrm{f}}\rangle_{\bar{L}}$
  with the same eigenenergy (as long as $\hat{V}_{K\to L}\neq 0$), by picking an
  eigenstate in the support of $\hat{V}_{K\to L}$ along with its associated image
  under $\hat{V}_{K\to L}$.
    
  Furthermore, with this choice it is always possible to satisfy our additional condition
  leading to~\eqref{eq:dilation-partial-energy-conserving-isometry-complicated-noPi}.
  This can be seen as follows.  Let $m=\operatorname{rank}(\hat{V})$.
  We choose energy eigenbases
  $\{ \lvert{u_j}\rangle_K \}_{j=1}^{d_K}$ of $K$ and
  $\{ \lvert{v_{j'}}\rangle_L \}_{j'=1}^{d_L}$ of $L$, with
  $\{ \lvert{u_j}\rangle_K \}_{j=1}^{m}$ spanning the support
  of $\hat{V}$ and with $\lvert{v_{j}}\rangle_L = \hat{V}_{K\to L}\,\lvert{u_j}\rangle_K$
  for those $j=1,\ldots, m$.  Then we choose
  $\lvert{\mathrm{i}}\rangle_L = \lvert{v_{1}}\rangle_L$ and
  $\lvert{\mathrm{f}}\rangle_K = \lvert{u_{1}}\rangle_L$ (assuming $\hat{V}\neq 0$),
  noting that they must have the same energy.  We see that all states of the form
  $\lvert{u_j}\rangle_K \otimes \lvert{\mathrm{i}}\rangle_L$ for $j>m$ can be mapped
  onto themselves, with clearly
  $(\langle{\mathrm{f}}\rvert_K \otimes \hat{I}_L) 
  (\lvert{u_j}\rangle_K \otimes \lvert{\mathrm{i}}\rangle_L)
  = 0$ because $\lvert{\mathrm{f}}\rangle_K = \lvert{u_1}\rangle_K \perp \lvert{u_j}\rangle_K$.
    
\item In the case of the generalized thermal operation depicted in \cref{fig}, we
  have $K = SB$ and $L=S'B$, with a given energy-conserving
  partial isometry $\hat{V}_{SB\to S'B}$.  In this case, we may choose $M = SS'B$,
  with $\bar{K}=S'$ and $\bar{L}=S$.  If necessary, we can enlarge
  $\bar{K}$ and $\bar{L}$ to include qubit systems
  $Q$ and/or $Q'$ as
  per Points~\ref{item:dilation-partial-energy-conserving-isometry-remarks--add-Q-i-f-states}
  and~\ref{item:dilation-partial-energy-conserving-isometry-remarks--add-Qp-for-noPi}
  to ensure that all the conditions of
  \cref{lemma:dilation-energy-conserving-partial-isometry} are satisfied.  Then
  there exists $\lvert{\mathrm{i}}\rangle_{\bar{K}}$,
  $\lvert{\mathrm{f}}\rangle_{\bar{L}}$, along with
  an energy-conserving unitary $\hat{U}_{SB\bar{K}\to S'B\bar{L}}$, such that
  \begin{align}
    \hat{V}_{SB\to S'B}
    = (\hat{I}_{S'B}\otimes \langle{\mathrm{f}}\rvert_{\bar{L}})\;
        \hat{U}_{SB\bar{K}\to S'B\bar{L}} \;
    (\hat{I}_{SB} \otimes \lvert{\mathrm{i}}\rangle_{\bar{K}})
    \ .      
  \end{align}
\end{enumerate}

We now turn to the proof of the proposition.

\begin{proof}[Proof of \cref{lemma:dilation-energy-conserving-partial-isometry}.]
First we compute as in \cref{prop:energypres-partial-isom-V-commutes-with-H}
the commutators
$[\hat{V}^\dagger \hat{V}, \hat{H}_K] =
\hat{V}^\dagger \hat{V}\hat{H}_K - \hat{H}_K \hat{V}^\dagger \hat{V}
= \hat{V}^\dagger (\hat{H}_{L} - \hat{H}_{L}) \hat{V} = 0$
and $[\hat{V}\hat{V}^\dagger, \hat{H}_{L}]
= \hat{V} \hat{V}^\dagger \hat{H}_{L} - \hat{H}_{L} \hat{V} \hat{V}^\dagger
= \hat{V} (\hat{H}_{K} - \hat{H}_{K}) \hat{V}^\dagger = 0$,
as well as $[\hat{V}^\dagger \hat{H}_{L} \hat{V}, \Hat{H}_K] =
\hat{V}^\dagger \hat{H}_{L} \hat{V}\hat{H}_K - \hat{H}_K \hat{V}^\dagger \hat{H}_{L} \hat{V}
= \hat{V}^\dagger (\hat{H}_{L}^2 - \hat{H}_{L}^2) \hat{V} = 0$
and $[\hat{V}^\dagger \hat{H}_L \hat{V}, \hat{V}^\dagger\hat{V}] 
= \hat{V}^\dagger \hat{H}_L\hat{V}\hat{V}^\dagger \hat{V}
  - \hat{V}^\dagger \hat{V}\hat{V}^\dagger \hat{H}_L \hat{V}
= \hat{V}^\dagger [\hat{H}_L, \hat{V}\hat{V}^\dagger] \hat{V} = 0$.

Because $\hat{U}'_{K\bar{K}\to M}$ and $\hat{U}''_{L\bar{L}\to M}$ are unitary
we must have $d_{K}d_{\bar{K}} = d_M = d_{L}d_{\bar{L}}$.
Also, the operator $\hat{U}_{L\bar{L}\leftarrow M}''^\dagger \hat{U}'_{K\bar{K}\to M}$
is an energy-conserving unitary operator from $K\bar{K}$ to $L\bar{L}$; therefore,
the Hamiltonians $H_K + H_{\bar{K}}$ and $H_L + H_{\bar{L}}$ must have the same
eigenvalues and with the same multiplicity.

Let $\hat{W}_{M} = \hat{U}''_{L\bar{L}\to M}\,
(\hat{V}_{K\to L}\otimes\lvert{\mathrm{f}}\rangle_{\bar{L}}\langle{\mathrm{i}}\rvert_{\bar{K}})
\,\hat{U}_{K\bar{K}\leftarrow M}'^\dagger$, noting that $\hat{W}_M$ is a partial isometry.  Furthermore,
$\hat{W}_{M} \hat{H}_M = \hat{U}''_{L\bar{L}\to M} \,
(\hat{V}_{K\to L}\otimes\lvert{\mathrm{f}}\rangle_{\bar{L}}\langle{\mathrm{i}}\rvert_{\bar{K}})
\,(\hat{H}_K+\hat{H}_{\bar{K}})\,\hat{U}_{K\bar{K}\leftarrow M}'^\dagger
= \ldots = \hat{H}_M\hat{W}_M$, recalling that $\lvert{\mathrm{i}}\rangle_{\bar{K}}$ and
$\lvert{\mathrm{f}}\rangle_{\bar{L}}$ have the same eigenvalue with respect to
$\hat{H}_K$ and $\hat{H}_{\bar{L}}$, respectively; therefore
$[\hat{W}_M, H_M] = 0$.

We can complete $\hat{W}_M$ into a fully energy-conserving unitary $\hat{U}_M$
by assigning to each input energy eigenstate an energy eigenstate of same energy at
the output; this association is possible since the eigenvalues of the input and output
systems coincide including with multiplicity.
Then~\eqref{eq:dilation-partial-energy-conserving-isometry-complicated} is satisfied
by construction, as can be checked by verifying the action of both sides of the
equation on an energy eigenbasis spanning the support of $\hat{V}_{K\to L}$.

Now we assume that the additional condition stated in the claim holds, in order
to prove~\eqref{eq:dilation-partial-energy-conserving-isometry-complicated-noPi}.

Let $\{ \lvert{\phi_j}\rangle_M \}_{j=1}^{d_M}$
be a basis of $M$ that is a simultaneous eigenbasis of
$\hat{H}_M$,
$\hat{U}'_{K\bar{K}\to M}\,(\hat{V}^\dagger\hat{V}\otimes
\lvert{\mathrm{i}}\rangle\hspace*{-0.25ex}\langle{\mathrm{i}}\rvert_{\bar{K}} )
\,\hat{U}_{K\bar{K}\leftarrow M}'^\dagger$, and
$\hat{U}'_{K\bar{K}\to M}\,(\hat{V}^\dagger \hat{H}_L \hat{V}\otimes
\lvert{\mathrm{i}}\rangle\hspace*{-0.25ex}\langle{\mathrm{i}}\rvert_{\bar{K}} )
\,\hat{U}_{K\bar{K}\leftarrow M}'^\dagger$,
and furthermore chosen such that \emph{(i)}~the
states $\{ \lvert{\phi_j}\rangle_M \}_{j=1}^{\operatorname{rank}(\hat{V})}$
span the support of $\hat{U}'_{K\bar{K}\to M}\,(\hat{V}_K\otimes
\lvert{\mathrm{i}}\rangle\hspace*{-0.25ex}\langle{\mathrm{i}}\rvert_{\bar{K}}
)\, \hat{U}_{K\bar{K}\leftarrow M}'^\dagger$, and \emph{(ii)}~the set
$\{ \lvert{\phi_j}\rangle_M \}_{j=\operatorname{rank}(\hat{V})+1}^{d_K}$
spans the subspace supported by $\hat{U}'_{K\bar{K}\to M}\,\bigl(
(\hat{I}_K-\hat{\Pi}_K)\otimes
\lvert{\mathrm{i}}\rangle\hspace*{-0.25ex}\langle{\mathrm{i}}\rvert_{\bar{K}}
\bigr)\, \hat{U}_{K\bar{K}\leftarrow M}'^\dagger$.

Let $\{ \lvert{\chi_{j'}}\rangle_M \}_{j'=1}^{d_M}$ be another basis of $M$
that is a simultaneous eigenbasis of $\hat{H}_M$ and
$\hat{U}''_{L\bar{L}\to M}\,(\hat{V}\hat{V}^\dagger\otimes
\lvert{\mathrm{f}}\rangle\hspace*{-0.25ex}\langle{\mathrm{f}}\rvert_{\bar{L}} )
\,\hat{U}_{L\bar{L}\leftarrow M}''^\dagger$, and
furthermore chosen such that \emph{(i)}~we have
$\lvert{\chi_{j'}}\rangle_M = \hat{U}''_{L\bar{L}\to M}
(\hat{V}_{K\to L} \lvert{\phi_{j'}}\rangle_M \otimes
\lvert{\mathrm{f}}\rangle_{\bar{L}}   )$
for all $j'=1, \ldots, \operatorname{rank}(\hat{V})$, \emph{(ii)}~we
have that the set
$\{ \lvert{\chi_{j'}}\rangle_M \}_{j'=\operatorname{rank}(\hat{V})+1}^{d_K}$
is orthogonal to
$
\hat{U}''_{L\bar{L}\to M} \,
(\hat{I}_L\otimes\lvert{\mathrm{f}}\rangle\hspace*{-0.25ex}\langle{\mathrm{f}}\rvert_{\bar{L}})
\, \hat{U}_{L\bar{L}\to M}''^\dagger
$ and \emph{(iii)}~we
also have that $\lvert{\chi_{j}}\rangle_M $
for $j=\operatorname{rank}(\hat{V})+1, \ldots, d_K$ is an energy eigenstate
with the same energy as $\lvert{\phi_j}\rangle_M$, which we can ensure
thanks to our additional assumption stated in the claim.

Then we define $\hat{U}_M$ as
\begin{align}
    \hat{U}_M
    &= \sum_j \lvert{\chi_j}\rangle\hspace*{-0.25ex}\langle{\phi_j}\rvert_{M}\ .
\end{align}
The operator  $\hat{U}_M$  is unitary and commutes with $\hat{H}_M$, since it maps
an energy eigenbasis onto an energy eigenbasis.  If the state $\lvert{\psi}\rangle_K$
is in the support of $\hat{V}$, we have that
$\hat{U}'_{K\bar{K}\to M} (\lvert{\psi}\rangle_K \otimes \lvert{\mathrm{i}}\rangle_{\bar{K}})
=
\sum_{j=1}^{\operatorname{rank}(\hat{V})} c_j \lvert{\phi_j}\rangle_{M}
$
for suitable complex coefficients $c_j$. Then
\begin{align}
(\hat{I}_L\otimes\langle{\mathrm{f}}\rvert_L)\,
\hat{U}_{L\bar{L}\leftarrow M}''^\dagger \hat{U}_M \hat{U}'_{K\bar{K}\to M} \,
(\lvert{\psi}\rangle_K \otimes \lvert{\mathrm{i}}\rangle_{\bar{K}})
&= \sum_{j=1}^{\operatorname{rank}(\hat{V})} c_j\;
(\hat{I}_L\otimes\langle{\mathrm{f}}\rvert_L)\,
\hat{U}_{L\bar{L}\leftarrow M}''^\dagger \; \lvert{\chi_j}\rangle_{M}
\nonumber\\
&=
\hat{V}_{K\to L} \lvert{\psi}\rangle_M \otimes
\lvert{\mathrm{f}}\rangle_{\bar{L}}\ .
\end{align}
If the state $\lvert{\psi'}\rangle_K$
lies outside the support of $\hat{V}$, we have that
$\hat{U}'_{K\bar{K}\to M} (\lvert{\psi'}\rangle_K \otimes \lvert{\mathrm{i}}\rangle_{\bar{K}})
=
\sum_{j=\operatorname{rank}(\hat{V})+1}^{d_K} c'_j \lvert{\phi_j}\rangle_{M}
$
for suitable complex coefficients $c_j'$, and
\begin{align}
(\hat{I}_L\otimes\langle{\mathrm{f}}\rvert_L)\,
\hat{U}_{L\bar{L}\leftarrow M}''^\dagger \hat{U}_M \hat{U}'_{K\bar{K}\to M} \,
(\lvert{\psi'}\rangle_K \otimes \lvert{\mathrm{i}}\rangle_{\bar{K}})
  &= \sum_{j=\operatorname{rank}(\hat{V})+1}^{d_K} c'_j \;
  (\hat{I}_L\otimes\langle{\mathrm{f}}\rvert_L)\,
\hat{U}_{L\bar{L}\leftarrow M}''^\dagger
\lvert{\chi_j}\rangle_M
  = 0\ .
\end{align}
We have therefore proven~\eqref{eq:dilation-partial-energy-conserving-isometry-complicated-noPi}.
\end{proof}

Now we present some general properties of the thermodynamic operations
introduced in \cref{sec:thermodynamic-operations}.

\begin{proposition}[Elementary properties of thermodynamic operations]
  \label{prop:elementary-thermo-transformation-properties}
  Consider systems $S,S'$ with corresponding Hamiltonians $\hat{H}_S,\hat{H}'_{S'}$.
  Let $*$ denote either TO or GPM.  The following hold:
  \begin{enumerate}[label=(\alph*)]
  \item \label{propitem:identity-process} If $S \simeq S'$
    and  $H'_{S'} = H_S + c$ for some $c\in\mathbb{R}$, the identity process
    is a $(c,0)$-work/coherence-assisted process in either model TO or GPM;
  \item \label{propitem:transform-different-energy-states} For two energy
    eigenstates $\lvert {E}\rangle _S, \lvert {E'}\rangle _{S'}$, we have
    $\lvert {E}\rangle _S \xrightarrow[*]{w,0,0} \lvert {E'}\rangle _{S'}$ if and only if $w \geqslant E'-E$;
  \item \label{propitem:wit-transformation} For any
    $w\in\mathbb{R},\eta>0,\varepsilon>0$, we have
    $\hat\rho_S\otimes\lvert {E}\rangle\hspace*{-0.25ex}\langle{E}\rvert _{A} \xrightarrow[*]{w,\eta,\varepsilon}
    \hat\rho'_{S'}\otimes\lvert {E'}\rangle\hspace*{-0.25ex}\langle{E'}\rvert _{A'}$ for energy eigenstates on ancillas
    $A,A'$ if and only if
    $\hat\rho_S \xrightarrow[*]{E+w-E',\eta,\varepsilon} \hat\rho'_{S'}$;
  \item \label{propitem:transform-different-thermal-states}
    We have
    $\hat\gamma\xrightarrow[*]{F'-F,\,0,\,0} \hat\gamma'$, where $\hat\gamma=e^{\beta(F-\hat{H})}$,
    $\hat\gamma'=e^{\beta(F'-\hat{H}')}$ with $F = - \beta^{-1}\ln\operatorname{tr}(e^{-\beta\hat{H}})$,
    $F' = - \beta^{-1}\ln\operatorname{tr}(e^{-\beta\hat{H}'})$;
  \item \label{propitem:nonoptimal-thermo-transformations}
    $\hat\rho\xrightarrow[*]{w,\,\eta,\,\varepsilon} \hat\rho'$ implies
    $\hat\rho\xrightarrow[*]{w',\,\eta',\,\varepsilon'} \hat\rho'$ for any $w'\geqslant w$,
    $\eta'\geqslant \eta$ and $\varepsilon'\geqslant\varepsilon$;
  \item \label{propitem:combine-thermo-transformations} If
    $\hat\rho\xrightarrow[*]{w,\,\eta,\,\varepsilon} \hat\rho'$ and
    $\hat\rho' \xrightarrow[*]{w',\,\eta',\,\varepsilon'} \hat\rho''$, then
    $\hat\rho\xrightarrow[*]{w+w',\, \eta+\eta',\, \varepsilon+\varepsilon'} \hat\rho''$.
  \end{enumerate}
\end{proposition}
\begin{proof}
  Property~\ref{propitem:identity-process} for $H'=H$ is obvious because the
  identity process is itself both a thermal operation and a Gibbs preserving map.  For
  $H'_{S'} = H_S + c\hat{I}$ with $c\neq 0$ we use a two-level battery $W$ with
  energy eigenstates $\lvert {0}\rangle _W,\lvert {c}\rangle _W$ and $H_W=c\,\lvert {c}\rangle\hspace*{-0.25ex}\langle{c}\rvert _W$; then
  $\hat{I}_{S\to S'}\otimes\lvert {0}\rangle\hspace*{-0.25ex}\langle{c}\rvert $ is an energy-conserving partial
  isometry, and thus a thermal operation, on the system $S$ and the battery $W$ with
  $c$ work expended.  The statement in the GPM model follows from
  \cref{lemma:TO-implies-GPM}.
  Property~\ref{propitem:transform-different-energy-states} is clear; the only
  nontrivial aspect is that we may have strict inequality.  That a thermal
  operation can perform this transformation can be seen using
  thermo-majorization~\cite{Horodecki2013_ThermoMaj}.  The statement for GPM
  follows because a thermal operation is also Gibbs-preserving.
  Property~\ref{propitem:wit-transformation} holds by definition of a
  $(w,\eta)$-work/coherence-assisted process; the systems $A,A'$ may be combined
  together with the battery system $W$ in the transformation.
  Property~\ref{propitem:transform-different-thermal-states} holds because the
  thermo-majorization curve of the thermal state is the line connecting $(0,0)$
  to $(e^{\beta F},1)$~\cite{Horodecki2013_ThermoMaj}.
  Property~\ref{propitem:nonoptimal-thermo-transformations} follows
  from~\ref{propitem:transform-different-energy-states}.
  To show Property~\ref{propitem:combine-thermo-transformations}, let $\Phi$
  (respectively $\Phi'$) be a work/coherence-assisted-process with parameters
  $(w,\eta)$ (respectively $(w',\eta')$).  Then $\Phi'\circ\Phi$ is a
  $(w+w', \eta+\eta')$-work/coherence-assisted process, and we have
  $D(\Phi'(\Phi(\hat\rho)), \hat\rho'') \leqslant D(\Phi'(\Phi(\hat\rho)),
  \Phi'(\hat\rho')) + D(\Phi'(\hat\rho'), \hat\rho'') \leqslant
  D(\Phi(\hat\rho), \hat\rho') + D(\Phi'(\hat\rho'), \hat\rho'') \leqslant
  \varepsilon+\varepsilon'$.
\end{proof}

Now we present the proofs of \cref{prop:relative_monotone,prop:quasi-monotonicity-dhyp-assisted-thermal-transformation} stated in
\cref{sec:thermodynamic-operations} regarding the monotonicity of the various
divergences under thermodynamic operations.

\begin{proof}[Proof of \cref{prop:relative_monotone}]
  We have $\hat\rho_S \xrightarrow[\mathrm{GPM}]{} \hat\rho_{S'}'$ (invoking
  \cref{lemma:TO-implies-GPM} if necessary); let $\Phi^{[\mathrm{GPM}]}$ be the
  corresponding Gibbs-sub-preserving map.
  The monotonicity of the hypothesis testing divergence follows directly from
  the properties~\eqref{eq:DHyp-dpi}
  and~\eqref{eq:DHyp-semidef-ordering-second-argument}.

  The monotonicity of the R\'enyi divergences is trickier to prove because the
  corresponding data processing inequality only holds for trace-preserving
  mappings.  Using~\cite[Proposition~2]{Faist2018PRX_workcost}, there exists a
  qubit system $Q$ with a basis $\{ \lvert {\mathrm{i}}\rangle _Q,\lvert {\mathrm{f}}\rangle _Q \}$
  and with a Hamiltonian
  $H_Q = q_{\mathrm i}\lvert {\mathrm{i}}\rangle\hspace*{-0.25ex}\langle{\mathrm{i}}\rvert +q_{\mathrm f}\lvert {\mathrm{f}}\rangle\hspace*{-0.25ex}\langle{\mathrm{f}}\rvert _Q$, as
  well as eigenstates $\lvert {\mathrm{i}}\rangle _{S'},\lvert {\mathrm{f}}\rangle _{S}$ of
  $\hat{H}_S,\hat{H}'_{S'}$, and a trace-preserving map
  $\mathcal{K}^{[\mathrm{GPM}]}_{SS'Q\to SS'Q}$ such that
  \begin{subequations}
    \begin{gather}
      \Phi^{[\mathrm{GPM}]}_{S\to S'}(\cdot)
      = \bigl\langle {\mathrm{f},\mathrm{f}}\bigr\rvert _{SQ} \; \mathcal{K}^{[\mathrm{GPM}]}\bigl(
      (\cdot)\otimes \lvert {\mathrm{i},\mathrm{i}}\rangle\hspace*{-0.25ex}\langle{\mathrm{i},\mathrm{i}}\rvert _{S'Q} \bigr) \,
      \bigl\lvert {\mathrm{f},\mathrm{f}}\bigr\rangle _{SQ}\ ;
      \\
      \mathcal{K}^{[\mathrm{GPM}]}_{SS'Q\to SS'Q}(e^{-\beta(\hat{H}_S + \hat{H}'_{S'} + \hat{H}_Q)})
      = e^{-\beta(\hat{H}_S + \hat{H}'_{S'} + \hat{H}_Q)}\ ;\text{ and}
      \\
      q_{\mathrm{i}} + \langle {\mathrm{i}}\hspace*{0.2ex}\vert\hspace*{0.2ex}{\hat{H}'_{S'}}\hspace*{0.2ex}\vert\hspace*{0.2ex}{\mathrm{i}}\rangle 
      = q_{\mathrm{f}} + \langle {\mathrm{f}}\hspace*{0.2ex}\vert\hspace*{0.2ex}{\hat{H}_{S}}\hspace*{0.2ex}\vert\hspace*{0.2ex}{\mathrm{f}}\rangle \ .
      \label{eq:hklhugiyfuhkj}
    \end{gather}
  \end{subequations}
  Since $\operatorname{tr}(\Phi^{[\mathrm{GPM}]}(\hat\rho_S)) = \operatorname{tr}(\hat\rho'_{S'}) = 1$, we
  can invoke~\cite[Corollary~3(b)]{Faist2018PRX_workcost} to see that
  \begin{align}
    \mathcal{K}^{[\mathrm{GPM}]}\bigl( \hat\rho_S \otimes
    \lvert {\mathrm{i},\mathrm{i}}\rangle\hspace*{-0.25ex}\langle{\mathrm{i},\mathrm{i}}\rvert _{S'Q} \bigr)
   = \Phi^{[\mathrm{GPM}]}_{S\to S'}(\hat\rho_S)
    \otimes\lvert {\mathrm{f},\mathrm{f}}\rangle\hspace*{-0.25ex}\langle{\mathrm{f},\mathrm{f}}\rvert _{SQ}\ .
  \end{align}
  Also, using~\cite[Proposition~17]{Faist2018PRX_workcost}
  and~\eqref{eq:hklhugiyfuhkj}, we have that
  \begin{align}
    {S}_{\alpha}(\lvert {\mathrm{i},\mathrm{i}}\rangle\hspace*{-0.25ex}\langle{\mathrm{i},\mathrm{i}}\rvert _{S'Q}\,\Vert\,{e}^{-\beta(\hat{H}'_{S'}+\hat{H}_Q)})
    =
    {S}_{\alpha}(\lvert {\mathrm{f},\mathrm{f}}\rangle\hspace*{-0.25ex}\langle{\mathrm{f},\mathrm{f}}\rvert _{SQ}\,\Vert\,{e}^{-\beta(\hat{H}_{S}+\hat{H}_Q)})
    =: C\ .
  \end{align}
  Then using the property~\eqref{eq:Dalpha-additive-under-tensor-products} of
  the R\'enyi $\alpha$-entropies and the above identities, we have
  \begin{align}
    {S}_{\alpha}(\hat\rho'_{S'}\,\Vert\,e^{-\beta\hat{H}'_{S'}})
    &= {S}_{\alpha}(\Phi^{[\mathrm{GPM}]}(\hat\rho_S)\otimes
      \lvert {\mathrm{f},\mathrm{f}}\rangle\hspace*{-0.25ex}\langle{\mathrm{f},\mathrm{f}}\rvert _{SQ} \,\Vert\, e^{-\beta(\hat{H}_S+\hat{H}'_{S'}+\hat{H}_Q)} )
      - C
      \nonumber\\
    &= {S}_{\alpha}(\mathcal{K}^{[\mathrm{GPM}]}(\hat\rho_S\otimes
      \lvert {\mathrm{i},\mathrm{i}}\rangle\hspace*{-0.25ex}\langle{\mathrm{i},\mathrm{i}}\rvert _{S'Q}) \,\Vert\,\mathcal{K}^{[\mathrm{GPM}]}(e^{-\beta(\hat{H}_S+\hat{H}'_{S'}+\hat{H}_Q)}) )
      - C
      \nonumber\\
    &\leqslant{S}_{\alpha}(\hat\rho_S\otimes \lvert {\mathrm{i},\mathrm{i}}\rangle\hspace*{-0.25ex}\langle{\mathrm{i},\mathrm{i}}\rvert _{S'Q}\,\Vert\,e^{-\beta(\hat{H}_S+\hat{H}'_{S'}+\hat{H}_Q)} )
      - C
      \nonumber\\
    &= {S}_{\alpha}(\hat\rho_S\,\Vert\,e^{-\beta\hat{H}_S})\ ,
  \end{align}
  where the inequality holds by the data processing
  inequality~\eqref{eq:Dalpha-dpi}.
\end{proof}

\begin{proof}[Proof
  of~\cref{prop:quasi-monotonicity-dhyp-assisted-thermal-transformation}]
  We prove the statement for the GPM model, invoking \cref{lemma:TO-implies-GPM}
  if necessary.  Let $C,C',W,W'$ be systems with Hamiltonians
  $\hat{H}_C,\hat{H}_{C'},\hat{H}_W,\hat{H}_{W'}$ from \cref{def:workcoherence} and let
  $\tilde\Phi^{[\mathrm{GPM}]}_{SCW\to S'C'W'}$ be the GPM operation
  in~\eqref{eq:workcoherence-process}.  Let $\hat{\tilde\rho}_{S'C'W'} = \tilde\Phi^{[\mathrm{GPM}]}_{SCW\to{}S'C'W'}(
  \hat\rho_S\otimes\lvert {E}\rangle\hspace*{-0.25ex}\langle{E}\rvert _W\otimes\lvert {\zeta}\rangle\hspace*{-0.25ex}\langle{\zeta}\rvert _C)$, with
  $D\bigl(\langle {E',\zeta'}\rvert _{W'C'} \hat{\tilde\rho}_{S'C'W'}
  \lvert {E',\zeta'}\rangle _{W'C'}\,,\, \hat\rho'_{S'}\bigr)\leqslant\varepsilon$.  Using
  property~\eqref{eq:DHyp-perturbation-state} we have
  \begin{align}
    {S}_{\mathrm{H}}^{\xi+\varepsilon}(\hat\rho'_{S'}\,\Vert\,e^{-\beta\hat{H}_{S'}})
    - \ln\Bigl(\frac{\xi+\varepsilon}{\xi}\Bigr)
    \leqslant{S}_{\mathrm{H}}^{\xi}(\langle {E',\zeta'}\rvert _{W'C'}\,\hat{\tilde\rho}_{S'C'W'}\,\lvert {E',\zeta'}\rangle _{W'C'}\,\Vert\,
      e^{-\beta\hat{H}_{S'}})\ .
    \label{eq:sdfghjklasdfghjk}
  \end{align}
  Now compute
  \begin{multline}
    \operatorname{tr}_{C'W'}\bigl[\lvert {E,\zeta'}\rangle\hspace*{-0.25ex}\langle{E,\zeta'}\rvert _{W'C'}\,
    e^{-\beta(\hat{H}_{S'}+\hat{H}_{C'}+\hat{H}_{W'})}\bigr]
    =
    e^{-\beta\hat{H}_{S'}}\,e^{-\beta E'}\,\langle {\zeta'}\hspace*{0.2ex}\vert\hspace*{0.2ex}{e^{-\beta H_C}}\hspace*{0.2ex}\vert\hspace*{0.2ex}{\zeta'}\rangle 
    \\
    \geqslant e^{-\beta(E'+\eta)}\,e^{-\beta\hat{H}_{S'}}\ ,
    \label{eq:lkghijobhvkliopik}
  \end{multline}
  because
  $\langle {\zeta'}\hspace*{0.2ex}\vert\hspace*{0.2ex}{e^{-\beta H_C}}\hspace*{0.2ex}\vert\hspace*{0.2ex}{\zeta'}\rangle  \geqslant\lambda_{\mathrm{min}}(e^{-\beta
    H_C}) \geqslant e^{-\beta\lVert {H_C}\rVert _{\infty}} \geqslant e^{-\beta\eta}$ where
  $\lambda_{\mathrm{min}}(\cdot)$ denotes the smallest eigenvalue of its
  argument.
  Observe that the operation
  $\operatorname{tr}_{C'W'}\bigl[\lvert {E',\zeta'}\rangle\hspace*{-0.25ex}\langle{E',\zeta'}\rvert _{W'C'}\,(\cdot)\bigr]$ is a completely
  positive, trace-nonincreasing map.  Then thanks to~\eqref{eq:DHyp-dpi}
  and~\eqref{eq:lkghijobhvkliopik} along with the scaling
  property~\eqref{eq:DHyp-scaling},
  \begin{align}
    \text{\eqref{eq:sdfghjklasdfghjk}}
    &\leqslant{S}_{\mathrm{H}}^{\xi}(\hat{\tilde\rho}_{S'C'W'}\,\Vert\,e^{\beta(E'+\eta)} e^{-\beta(\hat{H}_{S'}+\hat{H}_{W'}+\hat{H}_{C'})})
      \nonumber\\
    &= {S}_{\mathrm{H}}^{\xi}(\hat{\tilde\rho}_{S'C'W'}\,\Vert\,e^{-\beta(\hat{H}_{S'}+\hat{H}_{W'}+\hat{H}_{C'})})
    - \beta(E'+\eta)
      \nonumber\\
    &\leqslant{S}_{\mathrm{H}}^{\xi}(
      \tilde\Phi^{[\mathrm{GPM}]}\bigl(\hat\rho_S\otimes\lvert {E,\zeta}\rangle\hspace*{-0.25ex}\langle{E,\zeta}\rvert _{WC}\bigr)
      \,\Vert\,
      \tilde\Phi^{[\mathrm{GPM}]}\bigl(e^{-\beta(\hat{H}_{S}+\hat{H}_{W}+\hat{H}_{C})}\bigr)
      )
      - \beta(E'+\eta)
      \nonumber\\
    &\leqslant{S}_{\mathrm{H}}^{\xi}(
      \hat\rho_S\otimes\lvert {E,\zeta}\rangle\hspace*{-0.25ex}\langle{E,\zeta}\rvert _{WC}
      \,\Vert\,
      e^{-\beta(\hat{H}_{S}+\hat{H}_{W}+\hat{H}_{C})}
      )
      - \beta(E'+\eta)\ ,
      \label{eq:tdryfguhjkldrlijk}
  \end{align}
  where the two last inequalities hold using
  respectively~\eqref{eq:DHyp-semidef-ordering-second-argument} noting that
  $\tilde\Phi^{[\mathrm{GPM}]}_{SCW\to S'C'W'}$ is Gibbs-sub-preserving, and the
  data processing inequality~\eqref{eq:DHyp-dpi}.
  
  Let $\hat Q_{SCW}$ with $0\leqslant\hat Q_{SCW}\leqslant\hat{I}$ be an optimal choice
  for the last divergence term in~\eqref{eq:tdryfguhjkldrlijk}, such that
  ${S}_{\mathrm{H}}^{\xi}(\hat\rho_{S}\otimes\lvert {E,\zeta}\rangle\hspace*{-0.25ex}\langle{E,\zeta}\rvert _{WC}\,\Vert\,e^{-\beta(\hat{H}_{S}+\hat{H}_{W}+\hat{H}_{C})}) = -\ln\operatorname{tr}(\hat Q_{SCW}
  e^{-\beta(\hat{H}_{S}+\hat{H}_{W}+\hat{H}_{C})})$.  Let
  $\hat Q'_{S} = \langle {E,\zeta}\rvert _{WC}\,\hat Q_{SCW}\,\lvert {E,\zeta}\rangle _{WC}$, noting
  that $0\leqslant\hat Q'_{S}\leqslant\hat{I}_{S}$.  Then we have
  $\operatorname{tr}(\hat Q'_{S}\hat\rho_{S}) =
  \operatorname{tr}\bigl(\hat{Q}_{SCW}\,(\hat\rho_{S}\otimes\lvert {E,\zeta}\rangle\hspace*{-0.25ex}\langle{E,\zeta}\rvert _{WC})\bigr) \geqslant\xi$, and thus
  \begin{align}
    {S}_{\mathrm{H}}^{\xi}(\hat\rho_{S}\,\Vert\,e^{-\beta\hat{H}_S})
    &\geqslant-\ln\biggl(\frac1\xi \operatorname{tr}\bigl(\hat Q'_{S}\, e^{-\beta\hat{H}_S}\bigr) \biggr)
      \nonumber\\
    &= -\ln\biggl(\frac1\xi
      \operatorname{tr}\bigl(\hat Q_{SCW}\, \bigl(e^{-\beta\hat{H}_S}\otimes\lvert {E,\zeta}\rangle\hspace*{-0.25ex}\langle{E,\zeta}\rvert _{WC}\bigr)\bigr) \biggr)
      \nonumber\\
    &\geqslant-\ln\biggl(e^{\beta(E+\eta)}\,\frac1\xi\,
      \operatorname{tr}\bigl(\hat Q_{SCW}\, e^{-\beta(\hat{H}_S+\hat{H}_C+\hat{H}_W)}\bigr) \biggr)
      \nonumber\\
    &= {S}_{\mathrm{H}}^{\xi}(\hat\rho_{S}\otimes\lvert {E,\zeta}\rangle\hspace*{-0.25ex}\langle{E,\zeta}\rvert _{WC}\,\Vert\,e^{-\beta(\hat{H}_{S}+\hat{H}_{W}+\hat{H}_{C})})
      -\beta(E+\eta)\ ,
      \label{eq:rtyjlihgjjklhgj}
  \end{align}
  where in the last inequality we used
  $e^{-\beta\hat{H}_C}\geqslant\lambda_{\mathrm{min}}(e^{-\beta \hat{H}_C})\lvert {\zeta}\rangle\hspace*{-0.25ex}\langle{\zeta}\rvert _C
  \geqslant e^{-\beta\lVert {\hat{H}_C}\rVert _{\infty}}\lvert {\zeta}\rangle\hspace*{-0.25ex}\langle{\zeta}\rvert _C \geqslant e^{-\beta\eta}\lvert {\zeta}\rangle\hspace*{-0.25ex}\langle{\zeta}\rvert _C$ and
  $e^{-\beta\hat{H}_W}\geqslant e^{-\beta E}\lvert {E}\rangle\hspace*{-0.25ex}\langle{E}\rvert _W$ which imply together that
  $\lvert {E,\zeta}\rangle\hspace*{-0.25ex}\langle{E,\zeta}\rvert _{WC}\leqslant e^{\beta(E+\eta)}\,e^{-\beta(\hat{H}_C+\hat{H}_W)}$.
  Rewriting~\eqref{eq:rtyjlihgjjklhgj}, we have
  \begin{align}
    {S}_{\mathrm{H}}^{\xi}(\hat\rho_{S}\otimes\lvert {E,\zeta}\rangle\hspace*{-0.25ex}\langle{E,\zeta}\rvert _{WC}\,\Vert\,e^{-\beta(\hat{H}_{S}+\hat{H}_{W}+\hat{H}_{C})})
    \leqslant{S}_{\mathrm{H}}^{\xi}(\hat\rho_{S}\,\Vert\,e^{-\beta\hat{H}_S}) + \beta(E+\eta)\ ,
  \end{align}
  and finally,
  \begin{align}
    \text{\eqref{eq:tdryfguhjkldrlijk}}
    &\leqslant{S}_{\mathrm{H}}^{\xi}(\hat\rho\,\Vert\,e^{-\beta\hat{H}_S}) + \beta(E+\eta) - \beta(E'+\eta)
    \leqslant{S}_{\mathrm{H}}^{\xi}(\hat\rho\,\Vert\,e^{-\beta\hat{H}_S}) + \beta(w+2\eta)\ .
  \end{align}
  Following the chain of inequalities proves the claim.
\end{proof}

We present a convenient lemma that can ensure asymptotic convertibility if
good enough asymptotic convertibility can be achieved for any fixed
$\varepsilon>0$.
We first note that, thanks to Property~\ref{propitem:nonoptimal-thermo-transformations}
of \cref{prop:elementary-thermo-transformation-properties},
we may equivalently replace all limits
``$\lim_{n\to\infty}$'' in \cref{def:asymptotic_transformation} by ``$\limsup_{n\to\infty}$''.

\begin{lemma}
  \label{lemma:asympt-transform-fixedepsilon-ok}
  For sequences of states $\widehat{P}= \{ \hat\rho_n \}$,
  $\widehat{P}' = \{ \hat\rho'_n \}$ and sequences of Hamiltonians
  $\widehat{\mathcal{H}}= \{ \hat{H}_n \}$, $\widehat{\mathcal{H}}' = \{ \hat{H}'_n \}$, suppose that for all
  $\varepsilon>0$ there exists
  $w_{n,\varepsilon}, \eta_{n,\varepsilon}, \bar{\varepsilon}_{n,\varepsilon}$ such that
  $\hat\rho_n
  \xrightarrow[*]{w_{n,\varepsilon},\,\eta_{n,\varepsilon},\,\bar{\varepsilon}_{n,\varepsilon}}
  \hat\rho'_n$ for all $n$, where $*$ denotes TO or GPM.  If $r\in\mathbb{R}$
  is such that
  \begin{align}
    \label{eq:lemma-asympt-transform-fixedepsilon-ok-convconditions}
    \lim_{\varepsilon\to0} \limsup_{n\to\infty} \frac{w_{n,\varepsilon}}{n} &=r\ ;
    &
      \lim_{\varepsilon\to0} \limsup_{n\to\infty} \frac{\eta_{n,\varepsilon}}{n} &= 0\ ;
    &
      \lim_{\varepsilon\to0} \limsup_{n\to\infty} \bar{\varepsilon}_{n,\varepsilon} &= 0\ ,
  \end{align}
  then $\widehat{P}\xrightarrow[*]{r} \widehat{P}'$.
\end{lemma}

\begin{proof}
Let $w_\varepsilon:= \limsup_{n\to\infty} w_{n,\varepsilon}/n$,
  $\eta_\varepsilon:= \limsup_{n\to\infty} \eta_{n,\varepsilon}/n$, and
  $\bar{\varepsilon}_\varepsilon:= \limsup_{n\to\infty} \bar{\varepsilon}_{n,\varepsilon}$.
  Define
  \begin{align}
    N(\varepsilon) &:= \min \left\{ N: \forall n\geqslant N,\ 
                    \frac{w_{n,\varepsilon}}{n}\leqslant w_\varepsilon+\varepsilon\text{ and }
                    \frac{\eta_{n,\varepsilon}}{n} \leqslant \eta_\varepsilon+\varepsilon\text{ and }
                  \bar{\varepsilon}_{n,\varepsilon} \leqslant \bar{\varepsilon}_\varepsilon+ \varepsilon\right\}\ .
  \end{align}
  Now let $\varepsilon(n) := \inf \{ \varepsilon: N(\varepsilon) \leqslant n \}$ and
  observe that $\lim_{n\to\infty}\varepsilon(n) = 0$ because $N(\varepsilon)$ is
  finite for any small $\varepsilon>0$ thanks to the existence of the 
  limit superior
  defining $w_\varepsilon$, $\eta_\varepsilon$ and $\bar{\varepsilon}_\varepsilon$.  Then let
  $w_n := w_{n,\varepsilon(n)}$, $\eta_n := \eta_{n,\varepsilon(n)}$, and
  $\bar{\varepsilon}_n := \bar{\varepsilon}_{n,\varepsilon(n)}$, such that
  $\hat\rho_n \xrightarrow[*]{w_n,\,\eta_n,\,\bar{\varepsilon}_n} \hat\rho'_n$ for all $n$.
  We have $w_n/n = w_{n,\varepsilon(n)}/n \leqslant w_{\varepsilon(n)} + \varepsilon(n)$
  by definition of $\varepsilon(n)$ and hence
  $\limsup_{n\to\infty} w_n/n \leqslant\limsup_{n\to\infty} [w_{\varepsilon(n)} +
  \varepsilon(n)] = r$.  Similarly,
  $\eta_n/n = \eta_{n,\varepsilon(n)}/n \leqslant \eta_\varepsilon+ \varepsilon$ and
  thus $\lim_{n\to\infty} \eta_n/n = 0$.  Also,
  $\bar{\varepsilon}_n = \bar{\varepsilon}_{n,\varepsilon(n)} \leqslant \bar{\varepsilon}_\varepsilon+ \varepsilon$ and thus $\lim_{n\to\infty} \bar{\varepsilon}_n = 0$.
\end{proof}

An important known result is the fact that the min and max divergences quantify
the amount of work that is necessary to convert a semiclassical state
$\hat\rho$ to and from the thermal state.

\begin{proposition}[Work distillation and state formation for semiclassical
  states~\cite{Aberg2013_worklike,Horodecki2013_ThermoMaj}]
  \label{thm:work-distillation-state-formation-TO-semiclassical}
  Let $\hat\rho$ be a quantum state on a system with Hamiltonian $\hat{H}$, and
  suppose that $[\hat\rho,\hat{H}_S]=0$.  Let $\gamma''=1$ denote the trivial
  thermal state on the trivial system $\mathbb{C}$ with the trivial Hamiltonian
  $\hat{H}''=0$.  Then
  \begin{align}
    \hat\rho&\xrightarrow[\mathrm{TO}]{-\beta^{-1}{S}_{0}^{\varepsilon}(\hat\rho\,\Vert\,e^{-\beta\hat{H}}),\,0,\,\varepsilon}
      \gamma''\ ;
    &&\text{and}
    &
    \gamma''
    &\xrightarrow[\mathrm{TO}]{\beta^{-1}{S}_{\infty}^{\varepsilon}(\hat\rho\,\Vert\,e^{-\beta\hat{H}}),\,0,\,\varepsilon}
      \hat\rho\ .
  \end{align}
\end{proposition}

We now present a central proposition of this appendix, namely a simplified form
of \cref{prop:approx-equipartition-implies-reversibility-under-TO} that is
specific to Gibbs-preserving maps.  The error terms as well as the proof itself
are significantly simpler than the full result for thermal operations. 

\begin{proposition}[Work distillation and state
  formation~\cite{Horodecki2013_ThermoMaj,Faist2018PRX_workcost}]
  \label{prop:work-distillation-state-formation}
  Let $\hat\rho$ be a quantum state on a system with a Hamiltonian $\hat{H}$.  Let
  $\gamma''=1$ denote the trivial thermal state on the trivial system
  $\mathbb{C}$ with the trivial Hamiltonian $\hat{H}''=0$.  Then for any
  $\varepsilon\geqslant 0$ we have
  \begin{align}
    \label{eq:work-distillation-state-formation}
    \hat\rho&~\xrightarrow[\mathrm{GPM}]{-\beta^{-1}{S}_{0}^{\varepsilon}(\hat\rho\,\Vert\,e^{-\beta \hat{H}}),\, 0,\, \varepsilon}~
      \hat\gamma''\ ; &&\text{and}
    &
      \hat\gamma''
      &~\xrightarrow[\mathrm{GPM}]{\beta^{-1}{S}_{\infty}^{\varepsilon}(\hat\rho\,\Vert\,e^{-\beta \hat{H}}),\, 0,\, \varepsilon}~
      \hat\rho\ .
  \end{align}
  Consequently, for any $\hat\rho, \hat\rho'$, and for any Hamiltonians
  $\hat{H}, \hat{H}'$,
  \begin{align}
    \label{eq:state-transformation-min-max-rel-entropy}
    \hat\rho~\xrightarrow[\mathrm{GPM}]{ \beta^{-1} [ {S}_{\infty}^{\varepsilon}(\hat\rho'\,\Vert\,e^{-\beta \hat{H}'}) -
    {S}_{0}^{\varepsilon}(\hat\rho\,\Vert\,e^{-\beta \hat{H}}) ] ,\, 0,\, 2\varepsilon}~
    \hat\rho'\ .
  \end{align}
  For asymptotic sequences of states $\widehat{P}= \{ \hat\rho_n \}$,
  $\widehat{P}' = \{ \hat\rho'_n \}$ and sequences of Hamiltonians
  $\widehat{\mathcal{H}}= \{ \hat{H}_n \}$, $\widehat{\mathcal{H}}' = \{ \hat{H}'_n \}$, we have
  \begin{align}
    \label{eq:asymptotic-state-transformation-if-Dinf-Dsup}
    \widehat{P}\xrightarrow[\mathrm{GPM}]{\beta^{-1}[ {\overline{S}}(\widehat{P}'\,\Vert\,\widehat{\Sigma}')
    - {\underline{S}}(\widehat{P}\,\Vert\,\widehat{\Sigma})]} \widehat{P}'\ ,
  \end{align}
  where we denote by $\widehat{\Sigma}$ (respectively $\widehat{\Sigma}'$) the sequence
  $\{ e^{-\beta\hat{H}_n} \}$ (respectively $\{ e^{-\beta\hat{H}'_n} \}$).

\end{proposition}

\begin{proof}
  The statements~\eqref{eq:work-distillation-state-formation} are proven in
  Ref.~\cite{Faist2018PRX_workcost}.  The result for semiclassical states and
  thermal operations was shown in the earlier
  Ref.~\cite{Horodecki2013_ThermoMaj}.  The
  statement~\eqref{eq:state-transformation-min-max-rel-entropy} follows directly
  by combining the processes in~\eqref{eq:work-distillation-state-formation}.
  To prove~\eqref{eq:asymptotic-state-transformation-if-Dinf-Dsup}, observe that
  for any $\varepsilon>0$, we have for sufficiently large $n$ that
  $[{S}_{\infty}^{\varepsilon}(\hat\rho_n'\,\Vert\,\hat\sigma_n') -
  {S}_{0}^{\varepsilon}(\hat\rho_n\,\Vert\,\hat\sigma_n)]/n \leqslant{\overline{S}}(\widehat{P}'\,\Vert\,\widehat{\Sigma}') - {\underline{S}}(\widehat{P}\,\Vert\,\widehat{\Sigma}) + g(\varepsilon)$ where
  $g(\varepsilon)$ is some function of $\varepsilon$ with $g(\varepsilon)\to0$ as
  $\varepsilon\to0$.  Then~\eqref{eq:asymptotic-state-transformation-if-Dinf-Dsup}
  follows from~\eqref{eq:state-transformation-min-max-rel-entropy} and
  \cref{lemma:asympt-transform-fixedepsilon-ok}.\end{proof}

For completeness, we prove~\eqref{eq:state-transformation-min-max-rel-entropy}
directly with an explicit transformation (see also Theorem 6.3 of~\cite{SagawaBook}).

\begin{proof}[Alternative direct proof
  of~\eqref{eq:state-transformation-min-max-rel-entropy}]
  We prove the following equivalent statement: Assuming that
  ${S}_{0}^{\varepsilon}(\hat\rho\,\Vert\,e^{-\beta \hat{H}}) \geqslant{S}_{\infty}^{\varepsilon}(\hat\rho'\,\Vert\,e^{-\beta \hat{H}'})$, we explicitly construct a
  Gibbs-preserving operation that performs the given transformation using a
  hypothesis test.  The equivalence
  with~\eqref{eq:state-transformation-min-max-rel-entropy} follows from
  \cref{prop:elementary-thermo-transformation-properties}~\ref{propitem:wit-transformation}, the scaling
  property~\eqref{eq:Dalpha-scaling} of the divergences, and their additivity
  under tensor products~\eqref{eq:Dalpha-additive-under-tensor-products}.
  Without loss of generality we may assume that
  $\operatorname{tr}(e^{-\beta\hat{H}}) = \operatorname{tr}(e^{-\beta\hat{H}'}) = 1$; otherwise, shift the
  Hamiltonians by suitable constants and apply
  \cref{prop:elementary-thermo-transformation-properties}~\ref{propitem:identity-process} whose cost cancels the
  shift~\eqref{eq:Dalpha-scaling}.  Let $\hat\sigma= e^{-\beta\hat{H}}$ and
  $\hat\sigma'=e^{-\beta\hat{H}'}$, which are now quantum states.

  First, consider the case of $\varepsilon=0$.  We explicitly construct a CPTP map $E$ that maps $( \hat\rho, \hat\sigma)$
  to $(\hat\rho', \hat\sigma')$, by using a ``measure-and-prepare'' method.
  Let $c := e^{-{S}_{0}(\hat\rho\,\Vert\,\hat\sigma)}$, and let $\hat P_\rho$ be the
  projection onto the support of $\hat\rho$.  
  If $c=1$, the situation becomes trivial, because $\hat\rho= \hat\sigma$ and $\hat\rho' = \hat\sigma'$.
  If $c \neq 1$, we can construct the desired CPTP map $E$ as
  \begin{equation}
  E(\bullet ) := {\rm tr}[\hat P_{\rho} \bullet]  \hat\rho' + \left( 1- {\rm tr}[\hat P_{\rho} \bullet] \right) \frac{ \hat\sigma' - c \hat\rho'}{1-c},
  \end{equation}
  where the condition ${S}_{0}(\hat\rho\,\Vert\,\hat\sigma) \geqslant{S}_{\infty}(\hat\rho'\,\Vert\,\hat\sigma')$
  is used to guarantee that $\hat\sigma' - c \hat\rho' \geqslant0$.

  We next consider the case of $\varepsilon>0$.  By definition of the smooth entropies,
  there exist $\hat\tau,\hat\tau'$ such that
  ${S}_{\infty}^{\varepsilon}(\hat\rho'\,\Vert\,\hat\sigma') = {S}_{\infty}(\hat\tau'\,\Vert\,\hat\sigma')$
  and ${S}_{\infty}^{\varepsilon}(\hat\rho\,\Vert\,\hat\sigma) = {S}_{\infty}(\hat\tau\,\Vert\,\hat\sigma)$,
  with $D(\hat\tau,\hat\rho)\leqslant\varepsilon$,
  $D(\hat\tau',\hat\rho')\leqslant\varepsilon$.  From the case $\varepsilon=0$ we
  have that $\hat\tau\xrightarrow[\mathrm{GPM}]{}\hat\tau'$ with respect to the thermal states
  $\hat\sigma,\hat\sigma'$.  By triangle inequality and because quantum
  operations can only decrease the trace distance, we have that
  $D(E(\hat\rho),\hat\rho') \leqslant D(E(\hat\rho), E(\hat\tau)) +
  D(E(\hat\tau), \hat\rho') \leqslant D(\hat\rho,\hat\tau) +
  D(\hat\tau',\hat\rho') \leqslant 2\varepsilon$.  Hence
  $\hat\rho\xrightarrow[\mathrm{GPM}]{0,0,2\varepsilon}\hat\rho'$.
\end{proof}

As an immediate consequence, any state that satisfies
${S}_{0}(\hat\rho\,\Vert\,e^{-\beta\hat{H}})={S}_{\infty}(\hat\rho\,\Vert\,e^{-\beta\hat{H}})$ can be
reversibly converted to and from the thermal state
$e^{-\beta\hat{H}}/\operatorname{tr}(e^{-\beta\hat{H}})$ with Gibbs-preserving operations.  The same
holds for thermal operations if the state is semiclassical.  Consequently, the
common value of the divergences, which we can denote as $S(\hat\rho)$,
is the thermodynamic potential: It characterizes exactly which state
transformations are possible within this class of states.

\section{$C^\ast$-algebra formulation}
\label{appx:Cstar-algebras-construction}

In this appendix, we provide an overview of the standard formulation of
ergodicity with $C^\ast$-algebras~\cite{BookBratteliRobinson_OpAlgQStatMech1,BookBratteliRobinson_OpAlgQStatMech2,BookRuelle_StatMechRigorous},
and prove that it is equivalent to our formulation in
\cref{sec:main-result-for-ergodic-local-Gibbs}.  Furthermore, we prove 
\cref{quantum_theorem_main2t} in the alternative setting where we consider a
sequence of reduced states of the infinite Gibbs state, rather than a sequence
of finite Gibbs states corresponding to Hamiltonians truncated to finite
regions.  In the following, we use the notation of
\cref{sec:main-result-for-ergodic-local-Gibbs}.

The set of local operators is given by
$\mathcal A_{\mathrm{loc}} := \cup_\Lambda \mathcal A_\Lambda$ for a bounded
lattice region $\Lambda \subset \mathbb Z^d$.  Then, the $C^\ast$-algebra
$\mathcal A$ is defined as the $C^\ast$-inductive limit of
$\mathcal A_{\mathrm{loc}}$, which is often written as
$\mathcal A = \overline{\bigotimes_{i \in \mathbb Z^d} \mathcal A_i}$.

We consider a (normal) state $\Psi : \mathcal A \to \mathbb C$, where
$\Psi (\hat A ) \in \mathbb C$ is interpreted as the expectation value of
observable $\hat A$.  We consider a reduced state to a bounded region
$\Lambda \subset \mathbb Z^d$.  By definition, the reduced density operator on
this region, written as $\hat\rho_\Lambda$, satisfies
\begin{equation}
  \Psi (\hat A) = \operatorname{tr}[\hat\rho_{\Lambda} \hat A]\ ,
\end{equation}
for all $\hat A \in \mathcal A_{\Lambda}$.  We note that the consistency
condition~\eqref{consistency} is automatically satisfied for this
$\{ \hat\rho_\Lambda \}$.

By using the shift superoparator $T_i$ introduced in \cref{sec:ergodic-states},
we first define translation invariance.

\begin{definition}[Translation invariance]
  A state $\Psi$ is translation invariant, if for all $\hat A \in \mathcal A$
  and for all $i \in \mathbb Z^d$,
  \begin{align}
    \Psi (T_i ( \hat A)) = \Psi (\hat A)\ .
    \label{translation_invariance_C}
  \end{align}
\end{definition}

The above definition of translation invariance is equivalent to the definition
in \cref{sec:ergodic-states}; this is guaranteed by the following lemma, which
states that it is sufficient to take $\hat A$ above to be local.

\begin{lemma}
  If \cref{translation_invariance_C} is satisfied for all
  $\hat A \in \mathcal A_{\mathrm{loc}}$ and all $i \in \mathbb Z^d$, then $\Psi$ is
  translation invariant.
  \label{lemma_trans}
\end{lemma}
\begin{proof}
  Suppose that \cref{translation_invariance_C} is satisfied for all
  $\hat A \in \mathcal A_{\mathrm{loc}}$.  For any $\hat A \in \mathcal A$, there
  exists a sequence
  $\{ \hat A_m \}_{m \in \mathbb N} \subset \mathcal A_{\mathrm{loc}}$ such that
  $\hat A_m \in \mathcal A_{\Lambda_m}$ and
  $\lim_{m \to \infty} \lVert { \hat A - \hat A_m }\rVert _{\infty} = 0$.  Let
  $\hat \Delta_m := \hat A - \hat A_m$.  Then we have
  \begin{align}
    \bigl\lvert { \Psi (T_i ( \hat A)) - \Psi (\hat A) }\bigr\rvert 
    \leqslant\bigl\lvert { \Psi (T_i ( \hat A_m)) - \Psi (\hat A_m) }\bigr\rvert 
    + \bigl\lvert { \Psi (T_i ( \hat \Delta_m)) - \Psi (\hat \Delta_m) }\bigr\rvert \ .
  \end{align}
  The first term on the right-hand side vanishes.  The second term is bounded as
  \begin{equation}
    \bigl\lvert { \Psi (T_i ( \hat \Delta_m)) - \Psi (\hat \Delta_m) }\bigr\rvert 
    \leqslant2\, \lVert { \hat \Delta_m }\rVert _{\infty}\ ,
  \end{equation}
  which goes to zero as $m \to \infty$.
\end{proof}

We now define ergodicity in a more standard and mathematically elegant
way~\cite{BookBratteliRobinson_OpAlgQStatMech1,BookRuelle_StatMechRigorous} (see also
Refs.~\cite{Bjelakovic2004IM_lattice,Bjelakovic2005QIP_compression}).

\begin{definition}[Ergodicity]
  \label{ergodic_def_C1}
  A state $\Psi$ is translation-invariant and ergodic, if it is an extremal
  point of the set of translation-invariant states.
\end{definition}

Physically, an ergodic state corresponds to a ``pure thermodynamic phase''
without phase mixture, which is consistent with this mathematical definition.

The following theorem establishes the equivalence of the definition above with
the definition presented in~\cref{sec:ergodic-states}.  This is a reformulation
of Theorem~6.3.3, Proposition~6.3.5, and Lemma~6.5.1 of Ref.~\cite{BookRuelle_StatMechRigorous};
see also Ref.~\cite{Bjelakovic2005QIP_compression}.

\begin{lemma}
  \label{prop_ergodic}
  Using the notation of \cref{sec:main-result-for-ergodic-local-Gibbs}, the
  following are equivalent for any translation-invariant state $\Psi$:
  \begin{enumerate}[label=(\alph*)]
  \item $\Psi$ is ergodic;
  \item \label{propitem:ergodicity-selfaveraging}
      For all self-adjoint $\hat A \in \mathcal A$,
    \begin{equation}
      \lim_{\ell \to \infty}
      \Psi\mathopen{}\left( \Biggl(
        \frac{1}{(2\ell+1)^d} \sum_{i \in \Lambda_\ell} T_i (\hat A) \Biggr)^2 \right)
      = \bigl( \Psi ( \hat A ) \bigr)^2\ ;
      \label{ergodic_def_C2}
    \end{equation}
  \item For all  $\hat A, \hat B \in \mathcal A$,
    \begin{equation}
      \lim_{\ell \to \infty} \frac{1}{(2\ell+1)^d}
      \sum_{i\in \Lambda_\ell} \Psi\bigl(T_i(\hat A) \hat B\bigr)
      = \Psi(\hat A)\, \Psi (\hat B )\ ;
      \label{ergodic_def_C3}
    \end{equation}
  \item \label{propitem:ergodicity-selfavg-suffices-check-local-operators}
    \Cref{ergodic_def_C2} is satisfied for all self-adjoint
    $\hat A \in \mathcal A_{\mathrm{loc}}$.
  \end{enumerate}
\end{lemma}

For completeness, we prove the equivalence
of~\ref{propitem:ergodicity-selfavg-suffices-check-local-operators} with the
other points.

\begin{proof}
  It suffices to check that
  \ref{propitem:ergodicity-selfavg-suffices-check-local-operators}$\,\Rightarrow\,$\ref{propitem:ergodicity-selfaveraging}.  The proof is
  similar to that of \cref{lemma_trans}, and we use the same notation: For any
  $\hat A \in \mathcal A$, there exists a sequence
  $\{ \hat A_m \}_{m \in \mathbb N} \subset \mathcal A_{\mathrm{loc}}$ such that
  $\hat A_m \in \mathcal A_{\Lambda_m}$ and
  $\lim_{m \to \infty} \lVert { \hat A - \hat A_m }\rVert _{\infty} = 0$; let
  $\hat \Delta_m := \hat A - \hat A_m$.  Now suppose that \cref{ergodic_def_C2}
  is satisfied for all self-adjoint $\hat A \in \mathcal A_{\mathrm{loc}}$.  We first note
  that
  \begin{multline}
    \Biggl( \sum_{i \in \Lambda_\ell} T_i (\hat A) \Biggr)^2
    = \Biggl(\sum_{i \in \Lambda_\ell} T_i (\hat A_m) \Biggr)^2
    + \sum_{i,j \in \Lambda_\ell} \left( T_i (\hat A_m) \, T_j (\hat \Delta_m)
      + T_i (\hat \Delta_m) \, T_j (\hat A_m) \right)
    \\
    + \Biggl(\sum_{i \in \Lambda_\ell} T_i (\hat \Delta_m) \Biggr)^2\ .
  \end{multline}
  We then have
  \begin{align}
    \hspace*{1em}&\hspace*{-1em}
    \Biggl\lvert { \Psi \Biggl( \Biggl( \frac{1}{(2\ell+1)^d} \sum_{i \in \Lambda_\ell} T_i (\hat A) \Biggr)^2 \Biggr) - \Psi ( \hat A ) }\Biggr\rvert 
      \nonumber\\
    &\leqslant\Biggl\lvert { \Psi \Biggl( \Biggl(
      \frac{1}{(2\ell+1)^d} \sum_{i \in \Lambda_\ell} T_i (\hat A_m) \Biggr)^2 \Biggr)
      - \Psi ( \hat A_m ) }\Biggr\rvert   + 4 \lVert { \hat A_m }\rVert _{\infty} \lVert { \hat \Delta_m }\rVert _{\infty}
      +  2 \lVert { \hat \Delta_m }\rVert _{\infty}^2\ .
  \end{align}
  From \cref{ergodic_def_C2} for $\hat A_m \in \mathcal A_{\mathrm{loc}}$, we have,
  for a fixed $m$,
  \begin{equation}
    \lim_{\ell \to \infty}  \Biggl\lvert { \Psi \Biggl( \Biggl( \frac{1}{(2\ell+1)^d}
      \sum_{i \in \Lambda_\ell} T_i (\hat A) \Biggr)^2 \Biggr) - \Psi ( \hat A ) }\Biggr\rvert 
    \leqslant4 \lVert { \hat A_m }\rVert _{\infty} \lVert { \hat \Delta_m }\rVert _{\infty} +  2 \lVert { \hat \Delta_m }\rVert _{\infty}^2\ .
  \end{equation}
  Since $m$ can be taken arbitrarily large, the right-hand side above can be
  arbitrarily small.  Therefore, \cref{ergodic_def_C2} is satisfied for all
  $\hat A \in \mathcal A$.
\end{proof}

We now provide a definition of mixing that is suited to the formalism in this
section.

\begin{definition}[Mixing]
  \label{mixing_def_C}
  Let $T_{(k)}$ be the shift operator in \cref{def_mixing} in
  \cref{sec:main-result-for-ergodic-local-Gibbs}.  A state $\Psi$ has the mixing
  property, if for all $\hat A, \hat B \in \mathcal A$ and all $k$,
  \begin{align}
    \lim_{\ell \to \infty}  \Psi\bigl(T_{k}^\ell(\hat A) \hat B\bigr)
    = \Psi (\hat A) \Psi (\hat B )\ .
  \end{align}
\end{definition}

\begin{definition}[Weak mixing]
  A state $\Psi$ has the weak mixing property, if for all
  $\hat A, \hat B \in \mathcal A$,
  \begin{equation}
    \lim_{\ell \to \infty} \frac{1}{(2\ell+1)^d} \sum_{i\in \Lambda_\ell}
    \bigl\lvert { \Psi (T_i(\hat A) \hat B )  -  \Psi (\hat A) \Psi (\hat B ) }\bigr\rvert 
    = 0\ .
    \label{weak_mixing}
  \end{equation}
\end{definition}

Mixing implies weak mixing, and weak mixing implies ergodicity.  However, the
converses of them are not true.  In particular, the weak mixing in the above
sense should not be confused with \cref{ergodic_def_C3}.

The following lemma guarantees that the above definition of mixing is equivalent
to \cref{def_mixing} in \cref{sec:main-result-for-ergodic-local-Gibbs}.

\begin{lemma}
  In the definitions of mixing and weak mixing above, it is sufficient to take
  $\hat A, \hat B \in \mathcal A_{\mathrm{loc}}$.
\end{lemma}

\begin{proof}
  The proof of \ref{propitem:ergodicity-selfavg-suffices-check-local-operators}$\Rightarrow$\ref{propitem:ergodicity-selfaveraging} in \cref{prop_ergodic}
  provided above can be straightforwardly adapted to prove this lemma.
\end{proof}

We next consider the concept of local Gibbs states for the infinite-dimensional
setup~\cite{BookBratteliRobinson_OpAlgQStatMech2}.  We here assume that the Kubo-Martin-Schwinger (KMS)
state is unique at $\beta$, which physically implies no phase coexistence.
This is provable for any $\beta > 0$ in one dimension~\cite{Araki1975CMP_KMS}, but is
true at a sufficiently high temperature in higher
dimensions~\cite{BookBratteliRobinson_OpAlgQStatMech2}.

Let $\varphi^{\Box}_{\Lambda} : \mathcal A \to \mathbb C$ be the Gibbs state
corresponding to the truncated Hamiltonian associated with the region $\Lambda$,
and represented by the density operator $\hat\sigma^{\Box}_{\Lambda}$ in
\cref{tr_Hamiltonian} of \cref{sec:ergodic-states}.
Then, it is known that a state
\begin{equation}
  \Phi := \lim_{\ell \to \infty} \varphi^{\Box}_{\Lambda_\ell}
\end{equation}
exists, where the limit is given by the weak-$\ast$ (or ultraweak) topology of
the dual of $\mathcal A$ (cf.\@ Proposition~6.2.15 of Ref.~\cite{BookBratteliRobinson_OpAlgQStatMech2}).
We can then define the global Gibbs state on the entire lattice by $\Phi$.  This
global state satisfies the following condition for any
$\hat A \in \mathcal A_{\mathrm{loc}}$,
\begin{equation}
  \Phi (\hat A) = \lim_{\ell \to \infty}\varphi^{\Box}_{\Lambda_\ell} (\hat A)\ .
\end{equation}

Then, we define the reduced state of $\Phi$ on a bounded region $\Lambda$, which
is written as $\varphi_{\Lambda}$.  Let $\hat\sigma_{\Lambda}$ be the
corresponding reduced density operator.  For any observable
$\hat A \in \mathcal A_\Lambda$, we have
\begin{equation}
  \Phi (\hat A) = \varphi_{\Lambda} (\hat A) = \operatorname{tr}[\hat\sigma_{\Lambda} \hat A].
\end{equation}
In the following,  let  $\widehat{\Sigma}:= \{ \hat\sigma_n \}$  be  the sequence of the \textit{reduced} Gibbs states, where $ \hat\sigma_n := \hat\sigma_{\Lambda_\ell}$ and $n=(2\ell+1)^d$.
We note that the reduced state $\hat\sigma_{\Lambda}$ and the truncated state $\hat\sigma^{\Box}_{\Lambda}$ are different in general, where only  $\hat\sigma_{\Lambda}$ satisfies the consistency condition \eqref{consistency}.

We now prove another version of  \cref{quantum_theorem_main2t}
in \cref{sec:main-result-for-ergodic-local-Gibbs}, where $\widehat{\Sigma}$ is the
sequence of reduced states of the full Gibbs state on the infinite lattice,
instead of the sequence $\widehat{\Sigma}^{\Box}$ of Gibbs states corresponding to
truncated Hamiltonians associated with a sequence of finite regions.

Our proof strategy is to show that the asymptotic min divergence rate, the max
divergence rate and the KL divergence rate remain unchanged if we
substitute $\widehat{\Sigma}^{\Box}$ by $\widehat{\Sigma}$.  For this, we invoke the following
result, given as Theorem~3.11 in Ref.~\cite{Lenci2005JSP_onephase} (see in particular the
second proof provided in that reference, which holds for observables that are
not necessarily positive and proves the uniformity of the convergence).

\begin{proposition}[{Lenci and Rey-Bellet~\cite[Theorem~3.11]{Lenci2005JSP_onephase}}]
  \label{prop_Lenci}
  Suppose that the KMS state is unique.
  For any observable $\hat A_\Lambda \in \mathcal A_{\Lambda}$ for a bounded
  region $\Lambda \subset \mathbb Z^d$, we have
  \begin{align}
    \lim_{\Lambda \to \mathbb Z^d} \frac{1}{\lvert {\Lambda}\rvert }
    \Bigl\lvert { \ln \operatorname{tr}\bigl[\hat A_\Lambda \hat\sigma_\Lambda\bigr]
      - \ln \operatorname{tr}\bigl[\hat A_\Lambda \hat\sigma^{\Box}_\Lambda\bigr] }\Bigr\rvert  = 0\ ,
    \label{Lenci1}
  \end{align}
  where the convergence is uniform in $\hat A_\Lambda$. 
\end{proposition}

The above result allows us to prove that the KL divergence rate
does not change if we replace the Gibbs state of the truncated Hamiltonian by
the reduced state of the infinite Gibbs state.

\begin{lemma}
  \label{Gibbs_KL_rate}
  Suppose that the KMS state is unique and that ${S}_{1}(\widehat{P}\,\Vert\,\widehat{\Sigma}^{\Box})$
  exists.  Then ${S}_{1}(\widehat{P}\,\Vert\,\widehat{\Sigma})$ exists and equals
  ${S}_{1}(\widehat{P}\,\Vert\,\widehat{\Sigma}^{\Box})$.
\end{lemma}
\begin{proof}
  \Cref{prop_Lenci} implies that
  \begin{equation}
    \lim_{\Lambda \to \mathbb Z^d} \frac{1}{\lvert {\Lambda}\rvert }
    \Bigl\lvert { \operatorname{tr}[\hat \rho_\Lambda \ln \hat\sigma_\Lambda]
      - \operatorname{tr}[\hat \rho_\Lambda \ln \hat\sigma^{\Box}_\Lambda] }\Bigr\rvert  = 0\ ,
  \end{equation}
  which implies ${S}_{1}(\widehat{P}\,\Vert\,\widehat{\Sigma}) = {S}_{1}(\widehat{P}\,\Vert\,\widehat{\Sigma}^{\Box})$.
\end{proof}

Similarly, we may use \cref{prop_Lenci} to show that the min and max divergence
rates (via the hypothesis testing divergence rate) remain unchanged if we
replace $\widehat{\Sigma}^{\Box}$ by $\widehat{\Sigma}$.

\begin{lemma}
  \label{Gibbs_DHyp_rate}
  Suppose that the KMS state is unique and that
  ${S}_{\mathrm{H}}^{\eta}(\widehat{P}\,\Vert\,\widehat{\Sigma}^{\Box})$ exists for any $0<\eta<1$.  Then, for
  any $0 < \eta < 1$, the rate ${S}_{\mathrm{H}}^{\eta}(\widehat{P}\,\Vert\,\widehat{\Sigma})$ exists and
  equals ${S}_{\mathrm{H}}^{\eta}(\widehat{P}\,\Vert\,\widehat{\Sigma}^{\Box})$.
\end{lemma}
\begin{proof}
  From \cref{Lenci1} in \cref{prop_Lenci}, there exists $\delta_n > 0$
  satisfying $\lim_{n \to \infty}\frac{\delta_n}{n} = 0$ such that for any
  $\hat A_n \in \mathcal A_{\Lambda_\ell}$,
  \begin{equation}
    e^{-\delta_n} \operatorname{tr}[\hat A_n \hat\sigma^{\Box}_n]
    \leqslant\operatorname{tr}[\hat A_n \hat\sigma_n]
    \leqslant e^{+\delta_n}\operatorname{tr}[\hat A_n \hat\sigma^{\Box}_n]\ .
  \end{equation}
  Combined with \cref{def_SH}, this implies that
  \begin{equation}
    {S}_{\mathrm{H}}^{\eta}(\hat\rho_n\,\Vert\,\hat\sigma^{\Box}_n) - \delta_n
    \leqslant{S}_{\mathrm{H}}^{\eta}(\hat\rho_n\,\Vert\,\hat\sigma_n)
    \leqslant{S}_{\mathrm{H}}^{\eta}(\hat\rho_n\,\Vert\,\hat\sigma^{\Box}_n) + \delta_n\ .
  \end{equation}
  The claim follows
  by dividing this equation by $n$ and  taking the limit $n\to\infty$.
\end{proof}

It is now straightforward to combine \cref{Gibbs_KL_rate,Gibbs_DHyp_rate} to
prove another version of \cref{quantum_theorem_main2t} for the infinite Gibbs
state, rather than the limit of Gibbs states of the truncated Hamiltonian of
increasingly large finite regions.

\begin{theorem}[Collapse of the spectral rates for the reduced Gibbs state]
  \label{quantum_theorem_main2}
  Suppose that $\widehat{P}$ is translation invariant and ergodic, and $\widehat{\Sigma}$
  is the reduced Gibbs state of a local and translation invariant Hamiltonian in
  any dimensions, where the KMS state is unique.  Then, for any $0 < \eta < 1$,
  \begin{align}
    {S}_{\mathrm{H}}^{\eta}(\widehat{P}\,\Vert\,\widehat{\Sigma}) =  {S}_{1}(\widehat{P}\,\Vert\,\widehat{\Sigma})\ ,
    \label{second_main0}
  \end{align}
  and as a consequence,
  \begin{align}
    {\overline{S}}(\widehat{P}\,\Vert\,\widehat{\Sigma})
    = {\underline{S}}(\widehat{P}\,\Vert\,\widehat{\Sigma})
    = {S}_{1}(\widehat{P}\,\Vert\,\widehat{\Sigma})\ .
    \label{second_main}
  \end{align}
\end{theorem}

\section{An alternative proof of \cref{quantum_theorem_main2}}

Here we provide an alternative proof of \cref{quantum_theorem_main2} presented
above, in the case of a one-dimensional chain, by combining a known result by
Hiai, Mosonyi, and Ogawa~\cite{Hiai2007JMP_correlated} with the ergodic theorem of
Bjelakovi\'c~\cite{Bjelakovic2004CMP_ergodic}.  We state these results here:

\begin{proposition}[{Hiai, Mosonyi, and Ogawa~\protect\cite[Lemma~4.2]{Hiai2007JMP_correlated}}]
  \label{thm:Hiai-local-Gibbs}
  Let $\hat\sigma_n$ be the reduced local Gibbs state on $n$ sites in one
  dimension.  There exist $\alpha_1, \alpha_2 > 0$ and $m_0 \in \mathbb N$ such
  that for all $m \geqslant m_0$ and $k \in \mathbb N$ we have
  \begin{align}
    \alpha_1^{k-1} \hat\sigma_m^{\otimes k}
    \leqslant\hat\sigma_{km}
    \leqslant\alpha_2^{k-1} \hat\sigma_m^{\otimes k}\ .
    \label{Gibbs_product}
  \end{align}
\end{proposition}

\begin{proposition}[{Bjelakovi\'c and
  Siegmund-Schultze~\protect\cite[Theorem~2.1]{Bjelakovic2004CMP_ergodic}}]
  \label{thm:Bjelakovic-ergodic-product-Gibbs}
  Suppose that $\widehat{P}$ is translation-invariant and ergodic, and
  $\widehat{\Sigma}= \{ \sigma^{\otimes n} \}$ is i.i.d. Then, for any $0 < \eta < 1$,
  \begin{align}
    {S}_{\mathrm{H}}^{\eta}(\widehat{P}\,\Vert\,\widehat{\Sigma}) = {S}_{1}(\widehat{P}\,\Vert\,\widehat{\Sigma})\ .
  \end{align}
\end{proposition}

The proof strategy is thus to use \cref{thm:Hiai-local-Gibbs} to reduce the
problem for a local Gibbs state to a problem with a tensor product Gibbs state,
by coarse-graining the $n$-site chain into $k$ blocks of $m$ sites.  The problem
then falls in the scope of \cref{thm:Bjelakovic-ergodic-product-Gibbs} which
gives the desired result.

\begin{proof}[{Alternative proof of \cref{quantum_theorem_main2} in one
    dimension.}]
  \sloppy
  We fix $m \in \mathbb N$, and let $n=km + r$ with $1 \leqslant r \leqslant m-1$.  First
  we argue that we can essentially ignore the $r$ remaining sites and focus on
  the $km$ sites.
  From the monotonicity of the hypothesis testing divergence under CPTP maps, and therefore under the partial trace, we
  have for any $0<\eta<1$,
  \begin{align}
    {S}_{\mathrm{H}}^{\eta}(\hat\rho_n\,\Vert\,\hat\sigma_n)
    \geqslant{S}_{\mathrm{H}}^{\eta}(\hat\rho_{km}\,\Vert\,\hat\sigma_{km})\ .
  \end{align} 
  Fix $0<\eta<1$ and let $\hat Q_{km}$ denote an optimal operator
  in~\eqref{def_SH} such that
  $\eta^{-1}\operatorname{tr}\bigl(\hat Q_{km}\hat\sigma_{km}\bigr) =
  \exp\bigl(-{S}_{\mathrm{H}}^{\eta}(\hat\rho_{km}\,\Vert\,\hat\sigma_m^{\otimes k})\bigr)$.  Then, from
  \cref{thm:Hiai-local-Gibbs},
  \begin{align}
    {S}_{\mathrm{H}}^{\eta}(\hat\rho_{km}\,\Vert\,\hat\sigma_{km})
    &\geqslant- \ln \bigl( \eta^{-1}\operatorname{tr}[\hat Q_{km} \hat\sigma_{km}] \bigr)
      \nonumber\\
    &\geqslant- \ln \bigl( \eta^{-1} \operatorname{tr}[\hat Q_{km} \hat\sigma_m^{\otimes k}]\bigr)
      - (k-1)\ln \alpha_2
      \nonumber\\
    &= {S}_{\mathrm{H}}^{\eta}(\hat\rho_{km}\,\Vert\,\hat\sigma_m^{\otimes k})
      - (k-1)\ln \alpha_2\ .
  \end{align}
  Therefore,
  \begin{align}
    {S}_{\mathrm{H}}^{\eta}(\hat\rho_n\,\Vert\,\hat\sigma_n)
    \geqslant{S}_{\mathrm{H}}^{\eta}(\hat\rho_{km}\,\Vert\,\hat\sigma_m^{\otimes k})
    - (k-1)\ln \alpha_2\ .
  \end{align}
  From \cref{thm:Bjelakovic-ergodic-product-Gibbs}, we have for large $k$ and at
  fixed $m$,
  \begin{align}
    \frac{1}{k} {S}_{\mathrm{H}}^{\eta}(\hat\rho_{km}\,\Vert\,\hat\sigma_m^{\otimes k})
    = \frac{1}{k} {S}_{1}(\hat\rho_{km}\,\Vert\,\hat\sigma_m^{\otimes k}) + \delta_k\ ,
  \end{align}
  where $\lim_{k \to \infty} \delta_k = 0$.  Using the fact that the logarithm
  is an operator monotone and with~\eqref{Gibbs_product},
  \begin{align}
    {S}_{1}(\hat\rho_{km}\,\Vert\,\hat\sigma_m^{\otimes k})
    \geqslant{S}_{1}(\hat\rho_{km}\,\Vert\,\hat\sigma_{km}) + (k-1) \ln \alpha_1\ .
  \end{align}
  Hence, we obtain
  \begin{align}
    \frac{1}{n} {S}_{\mathrm{H}}^{\eta}(\hat\rho_n\,\Vert\,\hat\sigma_n)
    \geqslant\frac{1}{n} {S}_{1}(\hat\rho_{km}\,\Vert\,\hat\sigma_{km})
    + \frac{k-1}{n} \ln\frac{\alpha_1}{\alpha_2} + \frac{k}{n} \delta_k\ .
  \end{align}
  Taking  $\liminf_{n \to \infty}$ while fixing $m$, we obtain
  \begin{align}
    \liminf_{n \to \infty} \frac{1}{n} {S}_{\mathrm{H}}^{\eta}(\hat\rho_n\,\Vert\,\hat\sigma_n)
    \geqslant{S}_{1}(\widehat{P}\,\Vert\,\widehat{\Sigma}) + \frac{1}{m} \ln \frac{\alpha_1}{\alpha_2}\ ,
  \end{align}
  where we used \cref{Gibbs_KL_rate} to get the first term on the right-hand
  side.  Since $m$ can be taken arbitrarily large, we obtain
  \begin{equation}
    \liminf_{n \to \infty} \frac{1}{n} {S}_{\mathrm{H}}^{\eta}(\hat\rho_n\,\Vert\,\hat\sigma_n)
    \geqslant{S}_{1}(\widehat{P}\,\Vert\,\widehat{\Sigma})\ .
    \label{alt_1}
  \end{equation}

  We next show the opposite direction.  Again from the monotonicity of
  the hypothesis testing divergence under partial trace,
  \begin{equation}
    {S}_{\mathrm{H}}^{\eta}(\hat\rho_n\,\Vert\,\hat\sigma_n)
    \leqslant{S}_{\mathrm{H}}^{\eta}(\hat\rho_{(k+1)m}\,\Vert\,\hat\sigma_{(k+1)m})\ .
  \end{equation} 
  Fix $0<\eta<1$ and let $\hat Q'_{(k+1)m}$ denote an optimal operator
  in~\eqref{def_SH} such that
  $\eta^{-1}\operatorname{tr}\bigl(\hat Q_{(k+1)m}\hat\sigma_{(k+1)m}\bigr) =
  \exp\bigl(-{S}_{\mathrm{H}}^{\eta}(\hat\rho_{(k+1)m}\,\Vert\,\hat\sigma_{(k+1)m})\bigr)$.
  Then, using \cref{thm:Hiai-local-Gibbs},
  \begin{align}
    {S}_{\mathrm{H}}^{\eta}(\hat\rho_{(k+1)m}\,\Vert\,\hat\sigma_{(k+1)m})
    &=  - \ln\bigl( \eta^{-1} \operatorname{tr}[\hat Q'_{(k+1)m} \hat\sigma_{(k+1)m}] \bigr)
      \nonumber\\
    &\leqslant- \ln\bigl( \eta^{-1} \operatorname{tr}[\hat Q'_{(k+1)m} \hat\sigma_m^{\otimes (k+1)}] \bigr)
      - k\ln \alpha_1
      \nonumber\\
    &\leqslant{S}_{\mathrm{H}}^{\eta}(\hat\rho_{(k+1)m}\,\Vert\,\hat\sigma_m^{\otimes (k+1)})
      - k\ln \alpha_1\ .
  \end{align}
  Therefore, 
  \begin{align}
    {S}_{\mathrm{H}}^{\eta}(\hat\rho_n\,\Vert\,\hat\sigma_n)
    \leqslant{S}_{\mathrm{H}}^{\eta}(\hat\rho_{(k+1)m}\,\Vert\,\hat\sigma_m^{\otimes (k+1)})
    - k\ln \alpha_1\ .
  \end{align}
  From \cref{thm:Bjelakovic-ergodic-product-Gibbs}, we have for large $k$ and
  for fixed $m$,
  \begin{align}
    \frac{1}{k+1}  {S}_{\mathrm{H}}^{\eta}(\hat\rho_{(k+1)m}\,\Vert\,\hat\sigma_m^{\otimes (k+1)})
    = \frac{1}{k+1} {S}_{1}(\hat\rho_{km}\,\Vert\,\hat\sigma_m^{\otimes (k+1)}) + \delta'_k\ ,
  \end{align}
  where $\lim_{k \to \infty} \delta'_k = 0$.  Since the logarithm is an operator
  monotone, we have from inequality~\eqref{Gibbs_product},
  \begin{align}
    {S}_{1}(\hat\rho_{(k+1)m}\,\Vert\,\hat\sigma_m^{\otimes (k+1)})
    \leqslant{S}_{1}(\hat\rho_{(k+1)m}\,\Vert\,\hat\sigma_{(k+1)m}) + k \ln \alpha_2\ .
  \end{align}
  Therefore, we obtain
  \begin{align}
    \frac{1}{n} {S}_{\mathrm{H}}^{\eta}(\hat\rho_n\,\Vert\,\hat\sigma_n)
    \leqslant\frac{1}{n} {S}_{1}(\hat\rho_{(k+1)m}\,\Vert\,\hat\sigma_{(k+1)m})
    + \frac{k}{n} \ln \frac{\alpha_2}{\alpha_1} + \frac{k+1}{n} \delta'_k\ .
  \end{align}
  By taking $\limsup_{n \to \infty}$ while fixing $m$, we obtain
  \begin{align}
    \limsup_{n \to \infty} \frac{1}{n} {S}_{\mathrm{H}}^{\eta}(\hat\rho_n\,\Vert\,\hat\sigma_n)
    \leqslant{S}_{1}(\widehat{P}\,\Vert\,\widehat{\Sigma}) + \frac{1}{m} \ln \frac{\alpha_2}{\alpha_1}\ ,
  \end{align}
  where we again used \cref{Gibbs_KL_rate}.  Since $m$ can be taken arbitrarily
  large, we obtain
  \begin{align}
    \limsup_{n \to \infty}\frac{1}{n}{S}_{\mathrm{H}}^{\eta}(\hat\rho_n\,\Vert\,\hat\sigma_n)
    \leqslant{S}_{1}(\widehat{P}\,\Vert\,\widehat{\Sigma})\ .
    \label{alt_2}
  \end{align}
  \Cref{second_main0} then follows from inequalities~\eqref{alt_1}
  and~\eqref{alt_2}.
\end{proof}

\section{The classical case}

If we restrict the $C^\ast$-algebra $\mathcal A$ in one dimension to a
commutative subalgebra, we obtain a classical stochastic
process. 
Here, we flesh out explicitly the classical ergodic theorem that our argument in \cref{sec:main-result-for-ergodic-local-Gibbs} and  \cref{appx:Cstar-algebras-construction}
reduces to in the classical case.

The classical counterpart of the setup in
these sections is a two-sided stochastic process
over $\mathbb Z$ with finite alphabets.  Let $\{ x_\ell \}_{\ell\in \mathbb Z}$
be the stochastic process, where $x_l \in B$ with $B$ being a
finite set of alphabets, and let
$X_n := (x_{-\ell}, x_{-\ell+1}, \cdots, x_\ell)$ with $n:= 2\ell+1$.  We
consider sequences of probability distributions
$\widehat{P}:= \{ \rho_n (X_n) \}_{n \in \mathbb N}$ and
$\widehat{\Sigma}:= \{ \sigma_n (X_n ) \}_{n \in \mathbb N}$.

First of all, we briefly comment on mathematical details about the correspondence
between the classical case and the quantum case (see also
Refs.~\cite{Bjelakovic2004IM_lattice,Bjelakovic2004CMP_ergodic}).   Let $\mathcal A$ be the
$C^\ast$-algebra of an infinite spin chain.  We consider a unital Abelian
$C^\ast$-subalgebra $\mathcal B \subset \mathcal A$, which is interpreted as a
set of classical observables.  Let $\Phi$ be a quantum state on $\mathcal A$,
and $\Phi |_{\mathcal B}$ be its restriction to $\mathcal B$.  From the
Gelfand-Naimark theorem, $\mathcal B$ is identified with the Banach space
$C_0(K)$, which is the space of $\mathbb C$-valued continuous functions on a
compact Hausdorff space $K$.  In our setup, $K = B^{\mathbb Z}$, which is compact from the Tychonoff's theorem.  From
the Riesz-Markov-Kakutani representation theorem, the dual of $C_0 (K)$ is the
space of regular Borel measures on $K$.  Thus $\Phi |_{\mathcal B}$ is
identified with a probability measure on $K$ (i.e., a stochastic process over
$\mathbb Z$).

Classical ergodicity can be defined in the same manner as in the quantum case
(\cref{defn:ergodic}), i.e., as a commutative case of quantum ergodicity.  On
the other hand, the standard definition of classical ergodicity is that any
subset of trajectories in a stochastic process that is invariant under $T$ has
measure $0$ or $1$.  These definitions are equivalent for the finite-alphabet
case.  In fact, a classical stochastic process is translation-invariant ergodic
if and only if it is an extremal point of the set of translation-invariant
processes.  Also, as mentioned before, \cref{defn:ergodic} is equivalent to the
definition by extremality for quantum spin
systems~\cite{BookRuelle_StatMechRigorous,Bjelakovic2005QIP_compression}.

All the quantum divergences introduced in \cref{sec:preliminaries} can be
computed using as arguments a probability distribution and a vector of positive
entries of same length, by embedding both classical vectors into the diagonal
entries of an operator in a Hilbert space whose dimension is the same as the
number of entries in the vectors.

In the following, we argue that, explicitly for the classical case, if $\widehat{P}$ and $\widehat{\Sigma}$ satisfy a relative asymptotic
equipartition property (relative AEP), then the lower and the upper divergence
rates coincide, and they must equal the KL divergence.  We first
define the relative AEP in the form of a convergence in probability, a classical
counterpart of our quantum formulation in
\cref{sec:main-result-for-ergodic-local-Gibbs}.

\begin{definition}[Relative asymptotic equipartition property (relative AEP)]
  \label{def_AEP}
  We say that $\widehat{P}$ and $\widehat{\Sigma}$ satisfy the relative AEP if the
  KL divergence rate ${S}_{1}(\widehat{P}\,\Vert\,\widehat{\Sigma})$ exists and if
  $\frac{1}{n} \ln \frac{\rho_n (X_n) }{\sigma_n (X_n)}$ converges to
  ${S}_{1}(\widehat{P}\,\Vert\,\widehat{\Sigma})$ in probability by sampling $X_n$ according to
  $\rho_n$.
\end{definition}
This is equivalently formulated as follows
(see, for example, Theorem~11.8.2 of Ref.~\cite{BookCoverThomas2006_InfTheory}):
\begin{proposition}
  \label{prop:rel-AEP-rel-typ-set}
  Suppose that ${S}_{1}(\widehat{P}\,\Vert\,\widehat{\Sigma})$ exists.  The sequences $\widehat{P}$ and
  $\widehat{\Sigma}$ satisfy the relative AEP if and only if for any $\varepsilon> 0$,
  there exists a set $Q_n \subset \{ X_n \}$ (the relative typical set) such
  that for sufficiently large $n$:
  \begin{enumerate}[label=(\alph*)]
  \item\label{propitem:rel-AEP-rel-typ-set--equipartition}
    For any $X_n \in Q_n$,
    \begin{align}
      \exp( n({S}_{1}(\widehat{P}\,\Vert\,\widehat{\Sigma}) - \varepsilon) )
      \leqslant\frac{\rho_n (X_n)}{\sigma_n (X_n)}
      \leqslant\exp( n({S}_{1}(\widehat{P}\,\Vert\,\widehat{\Sigma}) + \varepsilon) )\ ;
    \end{align}
  \item\label{propitem:rel-AEP-rel-typ-set--rhoweight}
 $\rho_n [Q_n] > 1-\varepsilon$; and
  \item\label{propitem:rel-AEP-rel-typ-set--sigma}
    $(1-\varepsilon) \exp (-n ({S}_{1}(\widehat{P}\,\Vert\,\widehat{\Sigma}) + \varepsilon)) < \sigma_n
    [Q_n] < \exp (-n ({S}_{1}(\widehat{P}\,\Vert\,\widehat{\Sigma}) - \varepsilon))$.
  \end{enumerate}
  Here, $\rho_n [Q_n]$ and $\sigma_n [Q_n]$ represent the probability of $Q_n$
  according to distributions $\rho_n$ and $\sigma_n$, respectively.
\end{proposition}

The relative AEP ensures that the min and max divergence rates converge to the
KL divergence rate:

\begin{proposition}
  \label{AEP_spectral}
  If $\widehat{P}$ and $\widehat{\Sigma}$ satisfy the relative AEP, we have
  \begin{align}
    {\underline{S}}(\widehat{P}\,\Vert\,\widehat{\Sigma})
    = {\overline{S}}(\widehat{P}\,\Vert\,\widehat{\Sigma})
    = {S}_{1}(\widehat{P}\,\Vert\,\widehat{\Sigma})\ .
    \label{c_spectral_KL}
  \end{align}
\end{proposition}

\begin{proof}
  Although this proposition follows easily from
  \cref{spectral_Nagaoka1,spectral_Nagaoka2}, we here note an alternative proof
  based on \cref{def_spectral} with a slightly different intuition.
  We consider a subnormalized probability distribution $\tau_n (X_n)$ defined by
  $\tau_n (X_n) := \rho_n (X_n) $ for $X_n \in Q_n$ and $\tau_n (X_n) := 0$ for
  $X_n \notin Q_n$.  From $\operatorname{tr}[ \tau_n ] > 1-\varepsilon$ and with
  \cref{gentle_lemma}, we see that $\tau_n$ is a candidate for the maximization
  in ${S}_{0}^{{2\sqrt{\varepsilon}}}(\rho_n\,\Vert\,\sigma_n)$.  Therefore,
  \begin{align}
    {S}_{0}^{{2\sqrt{\varepsilon}}}(\rho_n\,\Vert\,\sigma_n)
    \geqslant{S}_{0}(\tau_n\,\Vert\,\sigma_n)
    \geqslant- \ln \sigma_n [Q_n]\ ,
  \end{align}
  where we used that $Q_n$ cannot be smaller than the support of $\tau_n$ to
  obtain the right inequality.  From the right inequality
  of~\ref{propitem:rel-AEP-rel-typ-set--sigma} in
  \cref{prop:rel-AEP-rel-typ-set}, we have
  \begin{align}
    {S}_{0}^{{2\sqrt{\varepsilon}}}(\rho_n\,\Vert\,\sigma_n)
    > n( {S}_{1}(\widehat{P}\,\Vert\,\widehat{\Sigma}) - \varepsilon)\ .
  \end{align}
  By taking the limit $n\to\infty$, we obtain
  \begin{align}
    {\underline{S}}(\widehat{P}\,\Vert\,\widehat{\Sigma}) \geqslant{S}_{1}(\widehat{P}\,\Vert\,\widehat{\Sigma})\ .
    \label{c_lower_KL}
  \end{align}
  Similarly, we have
  \begin{align}
    {S}_{\infty}^{{2 \sqrt{\varepsilon}}}(\rho_n\,\Vert\,\sigma_n)
    \leqslant{S}_{\infty}(\tau_n\,\Vert\,\sigma_n)
    = \ln \max_{X_n \in Q_n} \frac{\rho_n (X_n)}{\sigma_n (X_n)}\ .
  \end{align}
  From the right hand side of
  \cref{prop:rel-AEP-rel-typ-set}~\ref{propitem:rel-AEP-rel-typ-set--equipartition},
  we have
  \begin{align}
    {S}_{\infty}^{{2 \sqrt{\varepsilon}}}(\rho_n\,\Vert\,\sigma_n)
    < n({S}_{1}(\rho_n\,\Vert\,\sigma_n) + \varepsilon)\ .
  \end{align}
  By taking the limit, we obtain
  \begin{align}
    {\overline{S}}(\widehat{P}\,\Vert\,\widehat{\Sigma}) \leqslant{S}_{1}(\widehat{P}\,\Vert\,\widehat{\Sigma})\ .
    \label{c_upper_KL}
  \end{align}
  By combining \cref{c_lower_KL,c_upper_KL}, we obtain~\eqref{c_spectral_KL}.
\end{proof}

In the following, we assume that $\widehat{P}:= \{ \rho_n \}$ is
translation-invariant (i.e., stationary) and ergodic.  In this case the
non-relative AEP 
(i.e., the classical counterpart of \cref{ergodic_prop})
is satisfied, as a consequence of the Shannon-McMillan
theorem.

As in the quantum case, we define the reduced state
$\widehat{\Sigma}:= \{ \sigma_n \}$ of the global Gibbs state $\sigma$ of a local and
translation-invariant Hamiltonian in one dimension, where
$\sigma_n (X_n ) := \sigma (X_n)$ (i.e., $\sigma_n$ is a marginal distribution
of $\sigma$).  We can also define the truncated Gibbs state
$\widehat{\Sigma}^{\Box}:= \{ \sigma^{\Box}_n \}$.  The global Gibbs state $\sigma$ is
obtained as the limit of the truncated Gibbs states~\cite{Ruelle1968CMP_onedimensional}:
\begin{align}
  \sigma := \lim_{n \to \infty}\sigma^{\Box}_n\ ,
\end{align}
where convergence is given by the weak-$\ast$ topology (or the vague topology)
of the dual of the Banach space $C_0(K)$.

We remark that the case of the reduced Gibbs state $\widehat{\Sigma}$ can also be
obtained from a well-known fact that the relative AEP is satisfied for a
translation-invariant ergodic process with respect to a translation-invariant
Markov process.  (The relative AEP has also been proved in a stronger sense
(i.e., almost surely convergence).  See Ref.~\cite{Algoet1988AoP_sandwich} and references
therein. For our purpose here, however, convergence in probability is enough.)
In fact, we have the following lemma.

\begin{lemma}
  \label{transfer_lemma}
  The global Gibbs state $\sigma$ of a local and translation-invariant
  Hamiltonian in one dimension is translation-invariant Markovian.
\end{lemma}

\begin{proof}
  From the Hammersley-Clifford theorem~\cite{Hammersley1971_Markov} (see also
  Ref.~\cite{Kato2019CMP_Markov}), it is known that the Gibbs state of a local Hamiltonian
  on an arbitrary finite graph is Markovian.
  On the other hand, here we directly prove this lemma by
  explicitly calculating the global Gibbs distribution $\sigma$, without using
  the Hammersley-Clifford theorem.

  For simplicity, we assume that the local interaction is given in the form of
  $h_i = h(x_i, x_{i+1})$ and satisfies $h(x,y) = h(y,x)$.  We introduce the
  transfer matrix $T$, whose $(x_i,x_{i+1})$-element is given by
  \begin{align}
    \langle {x_i}\hspace*{0.2ex}\vert\hspace*{0.2ex}{ T }\hspace*{0.2ex}\vert\hspace*{0.2ex}{ x_{i+1} }\rangle  := \exp (-\beta h(x_i, x_{i+1}))\ .
  \end{align}
  Here, we used the bra-ket notation to represent the classical probability vectors.
  We denote the spectral decomposition of $T$ as
  \begin{align}
    T = \sum_{\lambda} e^{\lambda} \lvert {\lambda}\rangle\hspace*{-0.25ex}\langle{\lambda}\rvert \ .
    \label{T_spectrum}
  \end{align}
  We also assume that $T$ has a non-degenerate maximum eigenvalue
  $e^{\lambda_\ast}$. 

  For the truncated Hamiltonian
  $H_{[-\ell,\ell]} := \sum_{i=-\ell}^\ell h (x_i, x_{i+1})$, the truncated
  Gibbs distribution is given by
  \begin{align}
    \sigma^{\Box}_{[-\ell,\ell]} (x_{-\ell}, \cdots, x_\ell, x_{\ell+1})
    = \frac{\prod_{i=-\ell}^\ell \langle { x_i }\hspace*{0.2ex}\vert\hspace*{0.2ex}{ T }\hspace*{0.2ex}\vert\hspace*{0.2ex}{ x_{i+1} }\rangle  }{\langle { 1 }\hspace*{0.2ex}\vert\hspace*{0.2ex}{ T^{2\ell+1} }\hspace*{0.2ex}\vert\hspace*{0.2ex}{ 1 }\rangle  }\ ,
  \end{align}
  where $\lvert {1}\rangle  := \sum_{x_i} \lvert {x_i}\rangle $ is the column vector whose entries are
  all unity.  Its marginal distribution for an interval $[-\ell',n]$ with
  $\ell'<\ell$, $n<\ell$ is given by
  \begin{align}
    \sigma^{\Box}_{[-\ell,\ell]} (x_{-\ell'}, \cdots, x_n)
    = \frac{
      \langle {1}\hspace*{0.2ex}\vert\hspace*{0.2ex}{ T^{\ell-\ell'} }\hspace*{0.2ex}\vert\hspace*{0.2ex}{ x_{-\ell'} }\rangle 
      \prod_{i=-\ell'}^{n-1} \langle {x_i}\hspace*{0.2ex}\vert\hspace*{0.2ex}{ T }\hspace*{0.2ex}\vert\hspace*{0.2ex}{x_{i+1}}\rangle  \langle {x_n}\hspace*{0.2ex}\vert\hspace*{0.2ex}{ T^{\ell-n} }\hspace*{0.2ex}\vert\hspace*{0.2ex}{1}\rangle 
    }{
      \langle {1}\hspace*{0.2ex}\vert\hspace*{0.2ex}{ T^{2\ell+1} }\hspace*{0.2ex}\vert\hspace*{0.2ex}{1}\rangle 
    }\ .
  \end{align}
  The conditional probability is then given by
  \begin{align}
    \sigma^{\Box}_{[-\ell,\ell]} (\mathopen{} x_n \mathclose{}\,|\,\mathopen{}
    x_{n-1}, \cdots, x_{-\ell'}\mathclose{})
    &=  \frac{
        \langle {1}\hspace*{0.2ex}\vert\hspace*{0.2ex}{T^{\ell-\ell'}}\hspace*{0.2ex}\vert\hspace*{0.2ex}{x_{-\ell'}}\rangle 
        \prod_{i=-\ell'}^{n-1} \langle {x_i}\hspace*{0.2ex}\vert\hspace*{0.2ex}{T}\hspace*{0.2ex}\vert\hspace*{0.2ex}{x_{i+1}}\rangle  \langle {x_n}\hspace*{0.2ex}\vert\hspace*{0.2ex}{T^{\ell-n}}\hspace*{0.2ex}\vert\hspace*{0.2ex}{1}\rangle 
      }{
        \langle {1}\hspace*{0.2ex}\vert\hspace*{0.2ex}{T^{\ell-\ell'}}\hspace*{0.2ex}\vert\hspace*{0.2ex}{x_{-\ell'}}\rangle   \prod_{i=-\ell'}^{n-2} \langle {x_i}\hspace*{0.2ex}\vert\hspace*{0.2ex}{T}\hspace*{0.2ex}\vert\hspace*{0.2ex}{x_{i+1}}\rangle 
        \langle {x_{n-1}}\hspace*{0.2ex}\vert\hspace*{0.2ex}{T^{\ell-n+1}}\hspace*{0.2ex}\vert\hspace*{0.2ex}{1}\rangle 
      }
      \nonumber\\
    &= \frac{ \langle {x_n}\hspace*{0.2ex}\vert\hspace*{0.2ex}{T^{\ell-n}}\hspace*{0.2ex}\vert\hspace*{0.2ex}{1}\rangle  }{ \langle {x_{n-1}}\hspace*{0.2ex}\vert\hspace*{0.2ex}{T^{\ell-n+1}}\hspace*{0.2ex}\vert\hspace*{0.2ex}{1}\rangle  }
      \langle { x_{n-1} }\hspace*{0.2ex}\vert\hspace*{0.2ex}{ T }\hspace*{0.2ex}\vert\hspace*{0.2ex}{ x_n }\rangle \ ,
  \end{align}
  which depends only on $x_{n-1}$ and $x_n$~--- as expected from the
  Hammersley-Clifford theorem~--- with also an explicit dependency on $n$.
  From~\eqref{T_spectrum},
  \begin{align}
    \sigma^{\Box}_{[-\ell,\ell]} (\mathopen{}x_n\mathclose{}\,|\,\mathopen{}
    x_{n-1}, \cdots, x_{-\ell'}\mathclose{})
    = \frac{
      \sum_\lambda e^{\lambda (\ell-n)} \langle {x_n}\hspace*{0.2ex}\vert\hspace*{0.2ex}{\lambda}\rangle  \langle {\lambda}\hspace*{0.2ex}\vert\hspace*{0.2ex}{1}\rangle 
    }{
      \sum_\lambda e^{\lambda (\ell-n+1)} \langle { x_{n-1} }\hspace*{0.2ex}\vert\hspace*{0.2ex}{ \lambda }\rangle  \langle {\lambda}\hspace*{0.2ex}\vert\hspace*{0.2ex}{1}\rangle 
    }
    \langle { x_{n-1} }\hspace*{0.2ex}\vert\hspace*{0.2ex}{ T }\hspace*{0.2ex}\vert\hspace*{0.2ex}{ x_n }\rangle \ .
  \end{align}
  By taking the limit of $\ell$ while fixing $\ell'$ and $n$, we obtain
  \begin{align}
  \label{max_eingenvector}
    \lim_{\ell \to \infty} \sigma^{\Box}_{[-\ell,\ell]}(\mathopen{}x_n \mathclose{}\,|\,\mathopen{}
    x_{n-1}, \cdots, x_{-\ell'}\mathclose{})
    = \frac{ \langle {x_n}\hspace*{0.2ex}\vert\hspace*{0.2ex}{\lambda_\ast}\rangle  }{ e^{\lambda_\ast} \langle { x_{n-1} }\hspace*{0.2ex}\vert\hspace*{0.2ex}{ \lambda_\ast }\rangle  }
    \langle { x_{n-1} }\hspace*{0.2ex}\vert\hspace*{0.2ex}{ T }\hspace*{0.2ex}\vert\hspace*{0.2ex}{ x_n }\rangle \ ,
  \end{align}
  where the right-hand side depends only on $x_n$ and $x_{n+1}$ and no longer
  explicitly depends on $n$.  Therefore, the global Gibbs distribution $\sigma$
  satisfies
  \begin{equation}
    \sigma(\mathopen{}x_n\mathclose{}\,|\,\mathopen{} x_{n-1}, \cdots, x_{-\ell'}\mathclose{})
    = \sigma(\mathopen{}x_n\mathclose{}\,|\,\mathopen{} x_{n-1} \mathclose{})\ .
    \label{Doob_Markov}
  \end{equation}
  We note that it is straightforward to remove the assumption that $T$ has a non-degenerate maximum eigenvalue.
  In fact, we can just replace the right-hand side of \eqref{max_eingenvector} by multiple eigenvectors with the maximum eigenvalue of $T$.

  In general, a stochastic process $\sigma$ is defined as Markovian, if for any
  $n$
  \begin{align}
    \sigma (\mathopen{} x_n \mathclose{}\,|\,\mathopen{} x_{n-1}, x_{n-2}, \cdots \mathclose{})
    = \sigma (\mathopen{} x_n \mathclose{}\,|\,\mathopen{} x_{n-1} \mathclose{})
  \end{align}
  holds almost surely (see, for example, Chapter~2 of Ref.~\cite{BookDoob_StochasticProcesses}).
  Also, from the Levy's martingale convergence theorem,
  $\lim_{\ell' \to \infty} \sigma(\mathopen{}x_n\mathclose{}\,|\,\mathopen{}
  x_{n-1}, \cdots, x_{-\ell'}\mathclose{}) = \sigma(\mathopen{}x_n
  \mathclose{}\,|\,\mathopen{} x_{n-1}, \cdots \mathclose{})$ holds almost
  surely.  The claim then follows from \cref{Doob_Markov}.
\end{proof}

\def\ {\unskip\space}
\def\doibase#110.{https://doi.org/10.}
\catcode`\&=12\relax
\bibsep=2pt\relax
\def\selectlanguage#1{}

\end{document}